\documentclass{article}
\usepackage{mathtools,amssymb,amsthm,tikz,tikz-cd,multirow}
\usepackage{fullpage}
\usepackage{etoolbox}
\usepackage{extarrows}
\usepackage{enumitem}
\usepackage{makecell}
\usepackage{tensor}
\numberwithin{equation}{section}
\usetikzlibrary{calc}
\usepackage[english]{babel}
\usepackage[utf8]{inputenc}
\usepackage{amsmath}
\usepackage{amsfonts}
\usepackage{graphicx}
\usepackage{blindtext}
\usepackage{amssymb}
\usepackage{mathtools,tikz,tikz-cd,multirow}

\usepackage{amsthm}

\newtheorem{theorem}{Theorem}[section]

\newtheorem{lemma}[theorem]{Lemma}
\newtheorem{proposition}[theorem]{Proposition}
\theoremstyle{remark}

\usepackage{mathtools,amssymb,amsthm,tikz,tikz-cd,multirow}

\definecolor{darkgreen}{RGB}{0,180,0}
\colorlet[named]{green}{darkgreen}
\theoremstyle{definition}

\title{Stochastic symplectic ice}

\author{Chenyang Zhong \thanks{Department of Statistics, Stanford University}}

\date{\today}

\begin{document}
\maketitle

\newcommand{\caps}[2]{\begin{tikzpicture}
  \draw (0,-0.5) arc(-90:90:0.5);
  \filldraw[black] (0.5,0) circle (2pt);
  \node at (0,-0.5) [anchor=east] {$#2$};
  \node at (0,0.5) [anchor=east] {$#1$};
\end{tikzpicture}}

\newcommand{\capsC}[3]{\begin{tikzpicture}
   \draw (0,-0.5) to [out = 0, in = 180] (0.5,0);
  \draw (0,0.5) to [out = 0, in = 180] (0.5,0);
  \filldraw[black] (0.5,0) circle (2pt);
  \node at (0,-0.5) [anchor=east] {$#2$};
  \node at (0,0.5) [anchor=east] {$#1$};
  \node at (0.5,0)   [anchor=west] {$#3$};
\end{tikzpicture}}

\newcommand{\newcaps}[2]{\begin{tikzpicture}
  \draw (0,-0.5) to [out = 0, in = 180] (0.5,0);
  \draw (0,0.5) to [out = 0, in = 180] (0.5,0);
  \filldraw[black] (0.5,0) circle (2pt);
  \node at (0,-0.5) [anchor=east] {$#2$};
  \node at (0,0.5) [anchor=east] {$#1$};
\end{tikzpicture}}

\newcommand{\gammaicen}[4]{\begin{tikzpicture}
\coordinate (a) at (-.75, 0);
\coordinate (b) at (0, .75);
\coordinate (c) at (.75, 0);
\coordinate (d) at (0, -.75);
\coordinate (aa) at (-.75,.5);
\coordinate (cc) at (.75,.5);
\draw (a)--(c);
\draw (b)--(d);
\draw[fill=white] (a) circle (.25);
\draw[fill=white] (b) circle (.25);
\draw[fill=white] (c) circle (.25);
\draw[fill=white] (d) circle (.25);
\node at (0,1) { };
\node at (a) {$#1$};
\node at (b) {$#2$};
\node at (c) {$#3$};
\node at (d) {$#4$};
\path[fill=white] (0,0) circle (.2);
\node at (0,0) {$z_n$};
\end{tikzpicture}}

\newcommand{\gammaicei}[4]{\begin{tikzpicture}
\coordinate (a) at (-.75, 0);
\coordinate (b) at (0, .75);
\coordinate (c) at (.75, 0);
\coordinate (d) at (0, -.75);
\coordinate (aa) at (-.75,.5);
\coordinate (cc) at (.75,.5);
\draw (a)--(c);
\draw (b)--(d);
\draw[fill=white] (a) circle (.25);
\draw[fill=white] (b) circle (.25);
\draw[fill=white] (c) circle (.25);
\draw[fill=white] (d) circle (.25);
\node at (0,1) { };
\node at (a) {$#1$};
\node at (b) {$#2$};
\node at (c) {$#3$};
\node at (d) {$#4$};
\end{tikzpicture}}

\newcommand{\gammaice}[4]{\begin{tikzpicture}
\coordinate (a) at (-.75, 0);
\coordinate (b) at (0, .75);
\coordinate (c) at (.75, 0);
\coordinate (d) at (0, -.75);
\coordinate (aa) at (-.75,.5);
\coordinate (cc) at (.75,.5);
\draw (a)--(c);
\draw (b)--(d);
\draw[fill=white] (a) circle (.25);
\draw[fill=white] (b) circle (.25);
\draw[fill=white] (c) circle (.25);
\draw[fill=white] (d) circle (.25);
\node at (0,1) { };
\node at (a) {$#1$};
\node at (b) {$#2$};
\node at (c) {$#3$};
\node at (d) {$#4$};
\path[fill=white] (0,0) circle (.2);
\node at (0,0) {$z_i$};
\end{tikzpicture}}

\begin{abstract}
In this paper, we construct solvable ice models (six-vertex models) with stochastic weights and U-turn right boundary, which we term ``stochastic symplectic ice''. The models consist of alternating rows of two types of vertices. The probabilistic interpretation of the models leads to novel interacting particle systems where particles alternately jump to the right and then to the left. Two colored versions of the models and related stochastic dynamics are also introduced. Using the Yang-Baxter equations, we establish functional equations and recursive relations for the partition functions of these models. In particular, the recursive relations satisfied by the partition function of one of the colored models are closely related to Demazure-Lusztig operators of type C.
\end{abstract}

\section{Introduction}\label{Sect.1}

Since the pioneering investigation by Baxter (\cite{Bax2,Bax}), exactly solvable lattice models have found applications to various fields of mathematics and mathematical physics. Here, ``exactly solvable'' means that the Yang-Baxter equation, or ``star-triangle relation'', is satisfied by the system. We refer the reader to e.g. \cite{WZ,BBBG4,BSW,Kup,Kup2} for applications to algebraic combinatorics and to e.g. \cite{Bax3,BPZ,Di} for applications to quantum field theory.

Recently, there has been a series of works (see e.g. \cite{BBF,BBCFG,BBB,BBBG,BBBG2,BBBG3}) relating representation theory (for example, Tokuyama-type formulas and non-archimedean Whittaker functions) to exactly solvable lattice models. In the seminal work \cite{BBF}, a six-vertex model with free fermionic Boltzmann weights is introduced, whose partition function is shown to be equal to the product of a Schur polynomial and a deformation of Weyl's denominator. The culmination of this series of works is \cite{BBBG3}, where metaplectic Iwahori Whittaker functions are related to a supersymmetric solvable lattice model. These models are related to Cartan type A, and are based on the Yang-Baxter equation for free fermionic Boltzmann weights as introduced in \cite{BBF}. 

For Cartan type C, a parallel line of works has been initiated by Hamel and King (\cite{HK1,HK}) and Ivanov (\cite{Iva}). Their works show that the partition functions of certain six-vertex models equal the product of a deformation of Weyl's denominator and an irreducible character of the symplectic group $\mathrm{Sp}(2n,\mathbb{C})$. Ivanov's approach uses the Yang-Baxter equation as developed in \cite{BBF}. His lattice model consists of alternating rows of two types of ice (called $\Gamma$ ice and $\Delta$ ice) and has U-turn boundary on the right end. A different U-turn lattice model is related to Whittaker functions on the metaplectic double cover of $\mathrm{Sp}(2n,F)$, where $F$ is a non-archimedean local field (\cite{BBCG5}). Later work (\cite{Gra}) extends the results to metaplectic ice for Cartan type C. In \cite{BS3}, new deformations of Weyl's character formula for Cartan types B,C and D, and a character formula of Proctor for type BC, are obtained using the Yang-Baxter equation. A different approach based on a discrete time evolution operator on one-dimensional Fermionic Fock space is in \cite{BS2}. Further developments, including a dual version of the model in \cite{Iva}, generalizations of the models in \cite{Iva} and \cite{BBCG5}, and the $U_q(\mathfrak{sl}_2)$ six-vertex model with reflecting end, are in \cite{M1,M2,M4}. 

Another sequence of recent works, which comes from the subject of ``integrable probability'', relates stochastic systems, such as the asymmetric simple exclusion process (ASEP) and the Kardar-Parisi-Zhang equation, to certain solvable lattice models called the stochastic higher spin six-vertex models (see e.g. \cite{BCG,Bor,CP,Agg,OP,B4,BBCW} for some recent developments). Exactly solvable lattice models provide a powerful tool for analyzing probabilistic properties of these stochastic systems (for example, for proving Tracy-Widom type fluctuation results). The reader is also referred to \cite{BP} for a useful tutorial. The works on stochastic vertex models so far have been mainly restricted to models that are related to Cartan type A. 

The focus of this paper is to provide the first \underline{stochastic} model for Cartan type C, which we term ``stochastic symplectic ice''. This serves as an attempt to connect the above two sequences of works. Specifically, we construct two types of ice models (six-vertex models) with stochastic weights and U-turn boundary on the right end. The difference between the two types of models lies in the boundary condition at the U-turns (reflecting, or absorbing-and-emitting, see Section \ref{Sect.2.1} for details).   
The rows of the model alternate with two types of vertices, which we term ``stochastic $\Gamma$ vertex'' and ``stochastic $\Delta$ vertex'', in analogy to the terms used for Ivanov's symplectic ice model. The model shares some features with Ivanov's model, but the Boltzmann weights are quite different from previous works.

The stochastic symplectic ice model also leads to novel stochastic dynamics for interacting particle systems. Under the probabilistic interpretation, the particles alternately jump to the right and then jump to the left. The particles are reflected from the U-turn boundary, or absorbed into\slash emitted from the U-turn (depending on the type of the model). The partition function represents the probability of obtaining a particular particle configuration at the end of the evolution.

The models are \underline{solvable}, in that they satisfy the Yang-Baxter equation. In fact, four sets of Yang-Baxter equations are found for these models. Combining the Yang-Baxter equations with two further relations, the caduceus relation and the fish relation, we derive functional equations for both models (see Theorem \ref{MainT}).

A recent work by Borodin and Wheeler introduces a new attribute called ``color'' to stochastic vertex models (\cite{BW}). The colored stochastic vertex models are related to the quantum group $U_q(\widehat{\mathfrak{sl}}_{n+1})$, and degenerate to multi-species versions of interacting particle systems (such as multi-species ASEP). In \cite{BW}, recursive relations for the partition functions of these colored models are derived using the Yang-Baxter equation, and are related to Demazure-Lusztig operators of type A. For earlier developments related to colored vertex models, we refer the reader to \cite{CdW,BBW,FW}.

In this paper, we also construct two colored versions of the stochastic symplectic ice model. The colored models are also stochastic, and can be interpreted as stochastic dynamics of interacting particles with colors. In one of the colored models (which we introduce and study in Section \ref{Sect.3}), each particle carries a ``signed color'', whose sign changes to the opposite when reflected from the U-turn boundary. This seems to be a novel feature compared to previous colored stochastic vertex models. When specifying the boundary conditions, two signed permutations from the hyperoctahedral group (which is the Weyl group for type C) are involved. Previous works have mainly been focusing on boundary conditions specified by the symmetric group (which is the Weyl group for type A).

The colored models are also solvable, but only three sets of Yang-Baxter equations are found to be satisfied by the models. In order to derive recursive relations for the partition functions based on the three Yang-Baxter equations, we resort to a different relation, the reflection equation (see Theorems \ref{Refl} and \ref{Refl2N} below). The recursive relations for one of the colored models are further related to Demazure-Lusztig operators of type C.

We remark that following the ``fusion'' procedure as outlined in \cite[Appendix B]{BW}, higher spin versions of the models can be considered. We hope to investigate these higher spin models in the future. We also note that the ``height function'' as defined in \cite{BCG} can be similarly defined for the models here, and we leave the study of probabilistic properties of the models (similar to those in \cite[Theorems 1.1-1.2]{BCG}) to future works.

A few days before the first arXiv version of this paper was posted, a preprint by Buciumas and Scrimshaw (\cite{BS}) appeared on the arXiv. Their work constructed colored lattice models with partition functions representing symplectic and odd orthogonal Demazure characters and atoms. The work in this paper was done independently of and concurrently with their work. We also note that the colored models in this paper are quite different from theirs: the colored models in this paper are six-vertex models, while theirs are five-vertex models (in that the $b_1$ patterns for both $\Gamma$ ice and $\Delta$ ice have Boltzmann weight $0$ in their paper); the Boltzmann weights in this paper are also quite different from theirs (for example, in this paper, caps with negative color input only allow positive color output, while in their paper, such caps only allow negative color output; the weights in this paper are stochastic, while theirs are not); the recursive relations in this paper are related to Demazure-Lusztig operators of type C, while theirs are related to Demazure atoms and characters.

Section \ref{Sect.2} of the paper introduces the two types of the stochastic symplectic ice model. The Yang-Baxter equations, the caduceus relation and the fish relation are also given there. These relations are then used to derive the functional equations for the partition functions. Section \ref{Sect.3} introduces the first colored model for the stochastic symplectic ice. The Yang-Baxter equations and the reflection equation are also introduced in this section. Using these relations, the recursive relations for the partition functions are obtained. The recursive relations are further related to Demazure-Lusztig operators of type C. Finally we introduce and study the second colored model in Section \ref{Sect.4}.

\subsection{Acknowledgement}
The author wishes to thank Daniel Bump for his encouragement and many helpful conversations. The author thanks Valentin Buciumas, Nathan Gray, Slava Naprienko, and Travis Scrimshaw for their comments, and thanks Ben Brubaker for his help. The author also thanks Kohei Motegi for bringing the reference \cite{M4} to his attention.

\section{Stochastic symplectic ice}\label{Sect.2}

In this section, we introduce and study two classes of stochastic symplectic ice. They are termed ``reflecting stochastic symplectic ice'' and ``absorbing-and-emitting stochastic symplectic ice''. Section \ref{Sect.2.1} introduces the models and related Boltzmann weights. Section \ref{Sect.2.2} gives the Boltzmann weights for the R-matrices and shows the Yang-Baxter equation. An additional relation, the ``caduceus relation'', is shown in Section \ref{Sect.2.2.5}. By combining the Yang-Baxter equation with the caduceus relation and a further relation, the ``fish relation'', we establish functional equations satisfied by the partition functions in Section \ref{Sect.2.3}. We remark here that an explicit form of the partition function of the reflecting stochastic symplectic ice can also be computed based on the above-stated relations. We will present the details of the derivation in a subsequent work.

\subsection{The models}\label{Sect.2.1}

First we introduce some notations. By ``ice model'', we mean a planar lattice where every edge is assigned a $+$ or $-$ spin. To each vertex in the lattice, we assign a Boltzmann weight, which is a number that depends on the type of the vertex (there are two types of vertices for stochastic symplectic ice, see the next paragraph) and the $+$ or $-$ spins assigned to the four adjacent edges. A configuration\slash state means a labeling of the edges of the graph by $+$ or $-$ spins, and the Boltzmann weight of a configuration is the product of the Boltzmann weights of all the vertices for the configuration. An admissible state is a state where the assignment of spins to the edges adjacent to each vertex is one of the allowed assignments for that vertex (the allowed assignments are listed in tables later in the paper). The partition function of the ice model is the sum of Boltzmann weights for all admissible configurations.

In the stochastic symplectic ice model, two types of vertices are involved. They are termed ``stochastic $\Gamma$ vertex'' and ``stochastic $\Delta$ vertex'' in this paper, in analogy to the $\Gamma$ ice and the $\Delta$ ice used in Ivanov's symplectic ice model (see \cite{Iva}). The model depends on $n+1$ parameters $z_1,\cdots,z_n,q$, where $z_1,\cdots,z_n$ are called ``spectral parameters'' and $q$ is called the ``deformation parameter''. Throughout the paper, we also define
\begin{equation*}
    z_i'=q+1-\frac{1}{z_i},
\end{equation*}
for every $1\leq i\leq n$.

The Boltzmann weights for the stochastic $\Gamma$ vertex and the stochastic $\Delta$ vertex (with spectral parameter $z_i$) are listed in Figures \ref{Figure1}-\ref{Figure2}. Throughout the paper, an assignment of spins to the adjacent edges of a vertex that is not listed in the corresponding table has Boltzmann weight $0$. 

\begin{figure}
\[
\begin{array}{|c|c|c|c|c|c|}
\hline
  \tt{a}_1&\tt{a}_2&\tt{b}_1&\tt{b}_2&\tt{c}_1&\tt{c}_2\\
\hline
\gammaice{+}{+}{+}{+} &
  \gammaice{-}{-}{-}{-} &
  \gammaice{+}{-}{+}{-} &
  \gammaice{-}{+}{-}{+} &
  \gammaice{-}{+}{+}{-} &
  \gammaice{+}{-}{-}{+}\\
\hline
   1 & 1 & z_i & q z_i &1-q z_i & 1-z_i\\
\hline\end{array}\]
\caption{Boltzmann weights for stochastic $\Gamma$ vertex with spectral parameter $z_i$}
\label{Figure1}
\end{figure}

\begin{figure}
\[
\begin{array}{|c|c|c|c|c|c|}
\hline
  \tt{a}_1&\tt{a}_2&\tt{b}_1&\tt{b}_2&\tt{d}_1&\tt{d}_2\\
\hline
\gammaice{+}{+}{+}{+} &
  \gammaice{-}{-}{-}{-} &
  \gammaice{+}{-}{+}{-} &
  \gammaice{-}{+}{-}{+} &
  \gammaice{-}{-}{+}{+} &
  \gammaice{+}{+}{-}{-}\\
\hline
   1 & 1 & z_i' & \frac{1}{q} z_i' &  1-z_i' & 1-\frac{1}{q}z_i'\\
\hline\end{array}\]
\caption{Boltzmann weights for stochastic $\Delta$ vertex with spectral parameter $z_i$, where $z_i'=q+1-\frac{1}{z_i}$}
\label{Figure2}
\end{figure}

Now we introduce the stochastic symplectic ice model. We consider a rectangular lattice with $2n$ rows and $L$ columns. The rows are numbered $1,2,\cdots, 2n$ from bottom to top, and the columns are numbered $1,2,\cdots,L$ from right to left. Every odd-numbered row is a row of stochastic $\Delta$ vertices, and every even-numbered row is a row of stochastic $\Gamma$ vertices. The spectral parameter for the $i$th row of stochastic $\Gamma$ vertices and the $i$th row of stochastic $\Delta$ vertices is given by $z_i$.

The model also depends on a partition $\lambda=(\lambda_1,\cdots,\lambda_{n'})\in \mathbb{Z}^{n'}$, where $\lambda_1\geq \cdots \geq \lambda_{n'}$ and $n'\in\mathbb{N}_{+}$. We assume that $L\geq \lambda_1+n'$. The assignment of spins to boundary edges of the rectangular lattice is given as follows: on the left column, we assign $-$ to each row of stochastic $\Gamma$ vertex, and $+$ to each row of stochastic $\Delta$ vertex; on the top, we assign $+$ to each boundary edge; on the bottom, we assign $-$ to each column labeled $\lambda_i+n'+1-i$, for $1\leq i \leq n'$; on the right, the $i$th row of stochastic $\Gamma$ vertices and the $i$th row of stochastic $\Delta$ vertices are connected by a ``cap''.

For example, when $n=2$, $\lambda=(2,1)$, $L=4$, the model configuration is shown in Figure \ref{Configuration1}.

\begin{figure}[h]
\centering
\begin{tikzpicture}[scale=1]
\draw (0,1)--(5,1);
\draw (0,2)--(5,2);
\draw (0,3)--(5,3);
\draw (0,4)--(5,4);
\draw (1,0.5)--(1,4.5);
\draw (2,0.5)--(2,4.5);
\draw (3,0.5)--(3,4.5);
\draw (4,0.5)--(4,4.5);
\filldraw[black] (1,1) circle (1pt);
\filldraw[black] (2,1) circle (1pt);
\filldraw[black] (3,1) circle (1pt);
\filldraw[black] (4,1) circle (1pt);
\filldraw[black] (3,2) circle (1pt);
\filldraw[black] (2,2) circle (1pt);
\filldraw[black] (1,2) circle (1pt);
\filldraw[black] (4,2) circle (1pt);
\filldraw[black] (1,3) circle (1pt);
\filldraw[black] (2,3) circle (1pt);
\filldraw[black] (3,3) circle (1pt);
\filldraw[black] (4,3) circle (1pt);
\filldraw[black] (1,4) circle (1pt);
\filldraw[black] (2,4) circle (1pt);
\filldraw[black] (3,4) circle (1pt);
\filldraw[black] (4,4) circle (1pt);
\filldraw[black] (5.5,1.5) circle (1pt);
\filldraw[black] (5.5,3.5) circle (1pt);
\draw (5,1) arc(-90:90:0.5);
\draw (5,3) arc(-90:90:0.5);
\node at (0,1) [anchor=east] {$+$};
\node at (0,2) [anchor=east] {$-$};
\node at (0,3) [anchor=east] {$+$};
\node at (0,4) [anchor=east] {$-$};
\node at (-0.5,1) [anchor=east] {$1$};
\node at (-0.5,2) [anchor=east] {$2$};
\node at (-0.5,3) [anchor=east] {$3$};
\node at (-0.5,4) [anchor=east] {$4$};
\node at (-1,1) [anchor=east] {row};
\node at (0.5,1) [anchor=south] {$\Delta$};
\node at (0.5,2) [anchor=south] {$\Gamma$};
\node at (0.5,3) [anchor=south] {$\Delta$};
\node at (0.5,4) [anchor=south] {$\Gamma$};
\node at (4.5,1) [anchor=south] {$z_1$};
\node at (4.5,2) [anchor=south] {$z_1$};
\node at (4.5,3) [anchor=south] {$z_2$};
\node at (4.5,4) [anchor=south] {$z_2$};
\node at (1,4.5) [anchor=south] {$+$};
\node at (2,4.5) [anchor=south] {$+$};
\node at (3,4.5) [anchor=south] {$+$};
\node at (4,4.5) [anchor=south] {$+$};
\node at (1,5) [anchor=south] {$4$};
\node at (2,5) [anchor=south] {$3$};
\node at (3,5) [anchor=south] {$2$};
\node at (4,5) [anchor=south] {$1$};
\node at (0,5) [anchor=south] {column};
\node at (1,0.5) [anchor=north] {$-$};
\node at (2,0.5) [anchor=north] {$+$};
\node at (3,0.5) [anchor=north] {$-$};
\node at (4,0.5) [anchor=north] {$+$};
\end{tikzpicture}
\caption{Model configuration when $n=2,\lambda=(2,1),L=4$}
\label{Configuration1}
\end{figure}

We now discuss the Boltzmann weights for the caps. There are two choices of Boltzmann weights for the caps, which lead to two types of stochastic symplectic ice: ``reflecting stochastic symplectic ice'' and ``absorbing-and-emitting stochastic symplectic ice''. The Boltzmann weights of the caps for the two models are listed in Figures \ref{cap1}-\ref{cap2}, respectively. For reflecting symplectic ice, we always assume that $n'=n$ when taking the partition $\lambda$ for boundary conditions (by particle conservation).

\begin{figure}[h]
\[
\begin{array}{|c|c|c|c|c|c|}
\hline
\text{Cap} &\caps{+}{+} & \caps{-}{-} \\
\hline
\text{Boltzmann weight}  &  1 & 1 \\
\hline\end{array}\]
\caption{Boltzmann weights for caps: reflecting stochastic symplectic ice}
\label{cap1}
\end{figure}

\begin{figure}[h]
\[
\begin{array}{|c|c|c|c|c|c|}
\hline
\text{Cap} &\caps{+}{-} & \caps{-}{+} \\
\hline
\text{Boltzmann weight}  &  1 & 1 \\
\hline\end{array}\]
\caption{Boltzmann weights for caps: absorbing-and-emitting stochastic symplectic ice}
\label{cap2}
\end{figure}

Throughout the paper, we denote by $z=(z_1,z_2,\cdots,z_n)$ the vector formed by the $n$ spectral parameters. We also denote by $\mathcal{S}_{n,L,\lambda,z}$ the collection of admissible configurations of the reflecting stochastic symplectic ice with $2n$ rows, $L$ columns, spectral parameters $z_1,\cdots,z_n$ and bottom boundary condition given by $\lambda$. We also let $Z(\mathcal{S}_{n,L,\lambda,z})$ be the corresponding partition function. The collection of admissible configurations and the partition function of the absorbing-and-emitting stochastic symplectic ice are denoted by $\mathcal{T}_{n,L,\lambda,z}$ and $Z(\mathcal{T}_{n,L,\lambda,z})$, respectively.

We note that the Boltzmann weights of the models are \underline{stochastic}. For stochastic $\Gamma$ vertex, we view the left and top edges adjacent to the vertex as input, and the other two as output; for stochastic $\Delta$ vertex, we view the right and top edges adjacent to the vertex as input, and the other two as output; for caps, we view the top spin as input and the bottom spin as output. It can be seen that the possible Boltzmann weights for a given vertex (either the stochastic $\Gamma$ vertex, the stochastic $\Delta$ vertex or the cap) with given input sum up to $1$. Moreover, if $z$ satisfies the condition
\begin{equation}\label{Co1}
    \max\{0,\frac{1}{q+1}\} \leq  z_i\leq \min\{\frac{1}{q},1\}, \text{ for every }1 \leq i\leq n,
\end{equation}
then all possible Boltzmann weights are non-negative. Therefore, when the condition (\ref{Co1}) is satisfied, the Boltzmann weight of a given vertex can be interpreted as the probability of obtaining the output given the input at that vertex.

We also note that if the condition (\ref{Co1}) is satisfied by $z$, then both types of stochastic symplectic ice can be interpreted as an interacting particle system. We put $x-y$ coordinates on the rectangular lattice (see, for example, Figure \ref{Configuration1}) such that the $x$ coordinate of the $i$th row is $i$ and the $y$ coordinate of the $j$th column is $j$. For each $t=0,1,\cdots,2n$, we consider the set of vertical edges of the lattice that have a non-empty intersection with the line $x=2n-t+\frac{1}{2}$ and carry a $-$ spin. The positions of the particles at time $t$ are just the $y$ coordinates of these vertical edges. Therefore, an admissible state of the stochastic symplectic ice gives a possible evolution of the particles, and the Boltzmann weight for that state represents the probability of the particular evolution.

Now we describe the stochastic dynamics of the particles. For $t=0,1,\cdots,2n-1$, if $t$ is even, the particles attempt to jump to the right; if $t$ is odd, the particles attempt to jump to the left. 

The detailed rule is as follows. When $t$ is even, the particles are ordered from left to right. There is a new particle entering from the left boundary (we call it $0$th particle), which jumps to the right with geometric jump size (with parameter $q z_{n-\frac{t}{2}}$) unless it hits the $1$st particle; if the particle hits the $1$st particle, then it stops to move further. Starting from $l=1$, if the $l$th particle wasn't hit by any particle on its left, we flip a coin with head probability $z_{n-\frac{t}{2}}$ to determine whether it will stay at its current position or not; if the coin comes up tail, then the particle jumps to the right with geometric jump size (with parameter $q z_{n-\frac{t}{2}}$) unless it hits the $(l+1)$th particle; if the particle hits the $(l+1)$th particle, then it stops to move further. If the $l$th particle was hit by the $(l-1)$th particle, it jumps to the right by $1$ and the following move is the same as the previous case (except for the first step determining whether it will stay at the current position). Then the $(l+1)$th particle begins to move. If the rightmost particle moves beyond the first column (meaning that it hits the cap), for reflecting stochastic symplectic ice it is reflected by the cap (so it will start to move leftward from the first column at time $t+1$), while for absorbing-and-emitting stochastic symplectic ice it is absorbed. For absorbing-and-emitting stochastic symplectic ice, if no particle hits the cap, a new particle will be emitted from the first column at time $t+1$.

When $t$ is odd, the particles are ordered from right to left. If there is a particle reflected/emitted from the cap (we call it $0$th particle), it jumps to the left with geometric jump size (with parameter $\frac{1}{q}z_{n-\frac{t-1}{2}}'$) unless it hits the $1$st particle; if the particle hits the $1$st particle, then it stops to move further. Starting from $l=1$, if the $l$th particle wasn't hit by any particle on its right, we flip a coin with head probability $z_{n-\frac{t-1}{2}}'$ to determine whether it will stay at its current position or not; if the coin comes up tail, then the particle jumps to the left with geometric jump size (with parameter $\frac{1}{q}  z_{n-\frac{t-1}{2}}'$) unless it hits the $(l+1)$th particle; if the particle hits the $(l+1)$th particle, then it stops to move further. If the $l$th particle was hit by the $(l-1)$th particle, it jumps to the left by $1$ and the following move is the same as the previous case (except for the first step determining whether it will stay at the current position). Then the $(l+1)$th particle begins to move. 

Under this probabilistic interpretation, the partition functions $Z(\mathcal{S}_{n,L,\lambda,z})$ and $Z(\mathcal{T}_{n,L,\lambda,z})$ represent the probability that the particle configuration at time $t=2n$ corresponds to the partition $\lambda$ (meaning that the $i$th particle, ordered from left to right, has coordinate $\lambda_i+n'+1-i$ for $1\leq i\leq n'$) and no particle has ever moved left of the $L$th column.

\subsection{The R-matrix and the Yang-Baxter equation}\label{Sect.2.2}

The Yang-Baxter equation is a powerful tool for studying solvable lattice models. It involves two ordinary vertices (for our model, the stochastic $\Gamma$ vertex or the stochastic $\Delta$ vertex) and one additional rotated vertex called the R-vertex (also called the ``R-matrix'', due to connections to quantum group theory). 

The stochastic symplectic ice, as introduced in Section \ref{Sect.2.1}, is a solvable lattice model, in that we can find four types of R-matrices such that four sets of Yang-Baxter equations are satisfied by the model. In this section, we introduce the R-matrices for the stochastic symplectic ice, and show that the R-matrices together with the ordinary vertices (stochastic $\Gamma$ and $\Delta$ vertex) satisfy the Yang-Baxter equations, in the form of Theorem \ref{YBE1} below.

The four types of R-matrices are termed ``stochastic $\Gamma-\Gamma$ vertex'', ``stochastic $\Gamma-\Delta$ vertex'', ``stochastic $\Delta-\Delta$ vertex'' and ``stochastic $\Delta-\Gamma$ vertex'', in analogy to the terms used in \cite{Iva} for the symplectic ice. The Boltzmann weights for these R-matrices are given in Figures \ref{Rmatrix1}-\ref{Rmatrix4}.

\begin{figure}[h]
\[\begin{array}{|c|c|c|c|c|c|}
\hline
\begin{tikzpicture}[scale=0.7]
\draw (0,0) to [out = 0, in = 180] (2,2);
\draw (0,2) to [out = 0, in = 180] (2,0);
\draw[fill=white] (0,0) circle (.35);
\draw[fill=white] (0,2) circle (.35);
\draw[fill=white] (2,0) circle (.35);
\draw[fill=white] (2,2) circle (.35);
\node at (0,0) {$+$};
\node at (0,2) {$+$};
\node at (2,2) {$+$};
\node at (2,0) {$+$};
\node at (2,0) {$+$};
\path[fill=white] (1,1) circle (.3);
\node at (1,1) {$R_{z_i,z_j}$};
\end{tikzpicture}&
\begin{tikzpicture}[scale=0.7]
\draw (0,0) to [out = 0, in = 180] (2,2);
\draw (0,2) to [out = 0, in = 180] (2,0);
\draw[fill=white] (0,0) circle (.35);
\draw[fill=white] (0,2) circle (.35);
\draw[fill=white] (2,0) circle (.35);
\draw[fill=white] (2,2) circle (.35);
\node at (0,0) {$-$};
\node at (0,2) {$-$};
\node at (2,2) {$-$};
\node at (2,0) {$-$};
\path[fill=white] (1,1) circle (.3);
\node at (1,1) {$R_{z_i,z_j}$};
\end{tikzpicture}&
\begin{tikzpicture}[scale=0.7]
\draw (0,0) to [out = 0, in = 180] (2,2);
\draw (0,2) to [out = 0, in = 180] (2,0);
\draw[fill=white] (0,0) circle (.35);
\draw[fill=white] (0,2) circle (.35);
\draw[fill=white] (2,0) circle (.35);
\draw[fill=white] (2,2) circle (.35);
\node at (0,0) {$+$};
\node at (0,2) {$-$};
\node at (2,2) {$+$};
\node at (2,0) {$-$};
\path[fill=white] (1,1) circle (.3);
\node at (1,1) {$R_{z_i,z_j}$};
\end{tikzpicture}&
\begin{tikzpicture}[scale=0.7]
\draw (0,0) to [out = 0, in = 180] (2,2);
\draw (0,2) to [out = 0, in = 180] (2,0);
\draw[fill=white] (0,0) circle (.35);
\draw[fill=white] (0,2) circle (.35);
\draw[fill=white] (2,0) circle (.35);
\draw[fill=white] (2,2) circle (.35);
\node at (0,0) {$-$};
\node at (0,2) {$+$};
\node at (2,2) {$-$};
\node at (2,0) {$+$};
\path[fill=white] (1,1) circle (.3);
\node at (1,1) {$R_{z_i,z_j}$};
\end{tikzpicture}&
\begin{tikzpicture}[scale=0.7]
\draw (0,0) to [out = 0, in = 180] (2,2);
\draw (0,2) to [out = 0, in = 180] (2,0);
\draw[fill=white] (0,0) circle (.35);
\draw[fill=white] (0,2) circle (.35);
\draw[fill=white] (2,0) circle (.35);
\draw[fill=white] (2,2) circle (.35);
\node at (0,0) {$-$};
\node at (0,2) {$+$};
\node at (2,2) {$+$};
\node at (2,0) {$-$};
\path[fill=white] (1,1) circle (.3);
\node at (1,1) {$R_{z_i,z_j}$};
\end{tikzpicture}&
\begin{tikzpicture}[scale=0.7]
\draw (0,0) to [out = 0, in = 180] (2,2);
\draw (0,2) to [out = 0, in = 180] (2,0);
\draw[fill=white] (0,0) circle (.35);
\draw[fill=white] (0,2) circle (.35);
\draw[fill=white] (2,0) circle (.35);
\draw[fill=white] (2,2) circle (.35);
\node at (0,0) {$+$};
\node at (0,2) {$-$};
\node at (2,2) {$-$};
\node at (2,0) {$+$};
\path[fill=white] (1,1) circle (.3);
\node at (1,1) {$R_{z_i,z_j}$};
\end{tikzpicture}\\
\hline
1&1
&\frac{z_i-z_j}{1-(q+1)z_j+q z_i z_j}
&\frac{q(z_i-z_j)}{1-(q+1)z_j+q z_i z_j}
&\frac{(1-q z_i)(1-z_j)}{1-(q+1)z_j+q z_i z_j}
&\frac{(1-z_i)(1-q z_j)}{1-(q+1)z_j+q z_i z_j}\\
\hline
\end{array}\]
\caption{Boltzmann weights for stochastic $\Gamma-\Gamma$ vertex with spectral parameters $z_i$ and $z_j$}
\label{Rmatrix1}
\end{figure}

\begin{figure}[h]
\[\begin{array}{|c|c|c|c|c|c|}
\hline
\begin{tikzpicture}[scale=0.7]
\draw (0,0) to [out = 0, in = 180] (2,2);
\draw (0,2) to [out = 0, in = 180] (2,0);
\draw[fill=white] (0,0) circle (.35);
\draw[fill=white] (0,2) circle (.35);
\draw[fill=white] (2,0) circle (.35);
\draw[fill=white] (2,2) circle (.35);
\node at (0,0) {$+$};
\node at (0,2) {$+$};
\node at (2,2) {$+$};
\node at (2,0) {$+$};
\node at (2,0) {$+$};
\path[fill=white] (1,1) circle (.3);
\node at (1,1) {$R_{z_i,z_j}$};
\end{tikzpicture}&
\begin{tikzpicture}[scale=0.7]
\draw (0,0) to [out = 0, in = 180] (2,2);
\draw (0,2) to [out = 0, in = 180] (2,0);
\draw[fill=white] (0,0) circle (.35);
\draw[fill=white] (0,2) circle (.35);
\draw[fill=white] (2,0) circle (.35);
\draw[fill=white] (2,2) circle (.35);
\node at (0,0) {$-$};
\node at (0,2) {$-$};
\node at (2,2) {$-$};
\node at (2,0) {$-$};
\path[fill=white] (1,1) circle (.3);
\node at (1,1) {$R_{z_i,z_j}$};
\end{tikzpicture}&
\begin{tikzpicture}[scale=0.7]
\draw (0,0) to [out = 0, in = 180] (2,2);
\draw (0,2) to [out = 0, in = 180] (2,0);
\draw[fill=white] (0,0) circle (.35);
\draw[fill=white] (0,2) circle (.35);
\draw[fill=white] (2,0) circle (.35);
\draw[fill=white] (2,2) circle (.35);
\node at (0,0) {$+$};
\node at (0,2) {$-$};
\node at (2,2) {$+$};
\node at (2,0) {$-$};
\path[fill=white] (1,1) circle (.3);
\node at (1,1) {$R_{z_i,z_j}$};
\end{tikzpicture}&
\begin{tikzpicture}[scale=0.7]
\draw (0,0) to [out = 0, in = 180] (2,2);
\draw (0,2) to [out = 0, in = 180] (2,0);
\draw[fill=white] (0,0) circle (.35);
\draw[fill=white] (0,2) circle (.35);
\draw[fill=white] (2,0) circle (.35);
\draw[fill=white] (2,2) circle (.35);
\node at (0,0) {$-$};
\node at (0,2) {$+$};
\node at (2,2) {$-$};
\node at (2,0) {$+$};
\path[fill=white] (1,1) circle (.3);
\node at (1,1) {$R_{z_i,z_j}$};
\end{tikzpicture}&
\begin{tikzpicture}[scale=0.7]
\draw (0,0) to [out = 0, in = 180] (2,2);
\draw (0,2) to [out = 0, in = 180] (2,0);
\draw[fill=white] (0,0) circle (.35);
\draw[fill=white] (0,2) circle (.35);
\draw[fill=white] (2,0) circle (.35);
\draw[fill=white] (2,2) circle (.35);
\node at (0,0) {$-$};
\node at (0,2) {$-$};
\node at (2,2) {$+$};
\node at (2,0) {$+$};
\path[fill=white] (1,1) circle (.3);
\node at (1,1) {$R_{z_i,z_j}$};
\end{tikzpicture}&
\begin{tikzpicture}[scale=0.7]
\draw (0,0) to [out = 0, in = 180] (2,2);
\draw (0,2) to [out = 0, in = 180] (2,0);
\draw[fill=white] (0,0) circle (.35);
\draw[fill=white] (0,2) circle (.35);
\draw[fill=white] (2,0) circle (.35);
\draw[fill=white] (2,2) circle (.35);
\node at (0,0) {$+$};
\node at (0,2) {$+$};
\node at (2,2) {$-$};
\node at (2,0) {$-$};
\path[fill=white] (1,1) circle (.3);
\node at (1,1) {$R_{z_i,z_j}$};
\end{tikzpicture}\\
\hline
1&1
&\frac{z_i'+q z_j-(q+1)z_i'z_j}{1-z_i'z_j}
&\frac{q^{-1}z_i'+z_j-(1+q^{-1})z_i'z_j}{1-z_i'z_j}
&\frac{(1-z_i')(1-q z_j)}{1-z_i'z_j}
&\frac{(1-q^{-1}z_i')(1-z_j)}{1-z_i'z_j}\\
\hline
\end{array}\]
\caption{Boltzmann weights for stochastic $\Delta-\Gamma$ vertex with spectral parameters $z_i$ and $z_j$}
\label{Rmatrix2}
\end{figure}

\begin{figure}[h]
\[\begin{array}{|c|c|c|c|c|c|}
\hline
\begin{tikzpicture}[scale=0.7]
\draw (0,0) to [out = 0, in = 180] (2,2);
\draw (0,2) to [out = 0, in = 180] (2,0);
\draw[fill=white] (0,0) circle (.35);
\draw[fill=white] (0,2) circle (.35);
\draw[fill=white] (2,0) circle (.35);
\draw[fill=white] (2,2) circle (.35);
\node at (0,0) {$+$};
\node at (0,2) {$+$};
\node at (2,2) {$+$};
\node at (2,0) {$+$};
\node at (2,0) {$+$};
\path[fill=white] (1,1) circle (.3);
\node at (1,1) {$R_{z_i,z_j}$};
\end{tikzpicture}&
\begin{tikzpicture}[scale=0.7]
\draw (0,0) to [out = 0, in = 180] (2,2);
\draw (0,2) to [out = 0, in = 180] (2,0);
\draw[fill=white] (0,0) circle (.35);
\draw[fill=white] (0,2) circle (.35);
\draw[fill=white] (2,0) circle (.35);
\draw[fill=white] (2,2) circle (.35);
\node at (0,0) {$-$};
\node at (0,2) {$-$};
\node at (2,2) {$-$};
\node at (2,0) {$-$};
\path[fill=white] (1,1) circle (.3);
\node at (1,1) {$R_{z_i,z_j}$};
\end{tikzpicture}&
\begin{tikzpicture}[scale=0.7]
\draw (0,0) to [out = 0, in = 180] (2,2);
\draw (0,2) to [out = 0, in = 180] (2,0);
\draw[fill=white] (0,0) circle (.35);
\draw[fill=white] (0,2) circle (.35);
\draw[fill=white] (2,0) circle (.35);
\draw[fill=white] (2,2) circle (.35);
\node at (0,0) {$+$};
\node at (0,2) {$-$};
\node at (2,2) {$+$};
\node at (2,0) {$-$};
\path[fill=white] (1,1) circle (.3);
\node at (1,1) {$R_{z_i,z_j}$};
\end{tikzpicture}&
\begin{tikzpicture}[scale=0.7]
\draw (0,0) to [out = 0, in = 180] (2,2);
\draw (0,2) to [out = 0, in = 180] (2,0);
\draw[fill=white] (0,0) circle (.35);
\draw[fill=white] (0,2) circle (.35);
\draw[fill=white] (2,0) circle (.35);
\draw[fill=white] (2,2) circle (.35);
\node at (0,0) {$-$};
\node at (0,2) {$+$};
\node at (2,2) {$-$};
\node at (2,0) {$+$};
\path[fill=white] (1,1) circle (.3);
\node at (1,1) {$R_{z_i,z_j}$};
\end{tikzpicture}&
\begin{tikzpicture}[scale=0.7]
\draw (0,0) to [out = 0, in = 180] (2,2);
\draw (0,2) to [out = 0, in = 180] (2,0);
\draw[fill=white] (0,0) circle (.35);
\draw[fill=white] (0,2) circle (.35);
\draw[fill=white] (2,0) circle (.35);
\draw[fill=white] (2,2) circle (.35);
\node at (0,0) {$-$};
\node at (0,2) {$+$};
\node at (2,2) {$+$};
\node at (2,0) {$-$};
\path[fill=white] (1,1) circle (.3);
\node at (1,1) {$R_{z_i,z_j}$};
\end{tikzpicture}&
\begin{tikzpicture}[scale=0.7]
\draw (0,0) to [out = 0, in = 180] (2,2);
\draw (0,2) to [out = 0, in = 180] (2,0);
\draw[fill=white] (0,0) circle (.35);
\draw[fill=white] (0,2) circle (.35);
\draw[fill=white] (2,0) circle (.35);
\draw[fill=white] (2,2) circle (.35);
\node at (0,0) {$+$};
\node at (0,2) {$-$};
\node at (2,2) {$-$};
\node at (2,0) {$+$};
\path[fill=white] (1,1) circle (.3);
\node at (1,1) {$R_{z_i,z_j}$};
\end{tikzpicture}\\
\hline
1&1
&\frac{z_j'-z_i'}{q-(q+1)z_i'+z_i'z_j'}
&\frac{q(z_j'-z_i')}{q-(q+1)z_i'+z_i'z_j'}
&\frac{(1-z_i')(q-z_j')}{q-(q+1)z_i'+z_i'z_j'}
&\frac{(1-z_j')(q-z_i')}{q-(q+1)z_i'+z_i'z_j'}\\
\hline
\end{array}\]
\caption{Boltzmann weights for stochastic $\Delta-\Delta$ vertex with spectral parameters $z_i$ and $z_j$}
\label{Rmatrix3}
\end{figure}

\begin{figure}[h]
\[\begin{array}{|c|c|c|c|c|c|}
\hline
\begin{tikzpicture}[scale=0.7]
\draw (0,0) to [out = 0, in = 180] (2,2);
\draw (0,2) to [out = 0, in = 180] (2,0);
\draw[fill=white] (0,0) circle (.35);
\draw[fill=white] (0,2) circle (.35);
\draw[fill=white] (2,0) circle (.35);
\draw[fill=white] (2,2) circle (.35);
\node at (0,0) {$+$};
\node at (0,2) {$+$};
\node at (2,2) {$+$};
\node at (2,0) {$+$};
\node at (2,0) {$+$};
\path[fill=white] (1,1) circle (.3);
\node at (1,1) {$R_{z_i,z_j}$};
\end{tikzpicture}&
\begin{tikzpicture}[scale=0.7]
\draw (0,0) to [out = 0, in = 180] (2,2);
\draw (0,2) to [out = 0, in = 180] (2,0);
\draw[fill=white] (0,0) circle (.35);
\draw[fill=white] (0,2) circle (.35);
\draw[fill=white] (2,0) circle (.35);
\draw[fill=white] (2,2) circle (.35);
\node at (0,0) {$-$};
\node at (0,2) {$-$};
\node at (2,2) {$-$};
\node at (2,0) {$-$};
\path[fill=white] (1,1) circle (.3);
\node at (1,1) {$R_{z_i,z_j}$};
\end{tikzpicture}&
\begin{tikzpicture}[scale=0.7]
\draw (0,0) to [out = 0, in = 180] (2,2);
\draw (0,2) to [out = 0, in = 180] (2,0);
\draw[fill=white] (0,0) circle (.35);
\draw[fill=white] (0,2) circle (.35);
\draw[fill=white] (2,0) circle (.35);
\draw[fill=white] (2,2) circle (.35);
\node at (0,0) {$+$};
\node at (0,2) {$-$};
\node at (2,2) {$+$};
\node at (2,0) {$-$};
\path[fill=white] (1,1) circle (.3);
\node at (1,1) {$R_{z_i,z_j}$};
\end{tikzpicture}&
\begin{tikzpicture}[scale=0.7]
\draw (0,0) to [out = 0, in = 180] (2,2);
\draw (0,2) to [out = 0, in = 180] (2,0);
\draw[fill=white] (0,0) circle (.35);
\draw[fill=white] (0,2) circle (.35);
\draw[fill=white] (2,0) circle (.35);
\draw[fill=white] (2,2) circle (.35);
\node at (0,0) {$-$};
\node at (0,2) {$+$};
\node at (2,2) {$-$};
\node at (2,0) {$+$};
\path[fill=white] (1,1) circle (.3);
\node at (1,1) {$R_{z_i,z_j}$};
\end{tikzpicture}&
\begin{tikzpicture}[scale=0.7]
\draw (0,0) to [out = 0, in = 180] (2,2);
\draw (0,2) to [out = 0, in = 180] (2,0);
\draw[fill=white] (0,0) circle (.35);
\draw[fill=white] (0,2) circle (.35);
\draw[fill=white] (2,0) circle (.35);
\draw[fill=white] (2,2) circle (.35);
\node at (0,0) {$-$};
\node at (0,2) {$-$};
\node at (2,2) {$+$};
\node at (2,0) {$+$};
\path[fill=white] (1,1) circle (.3);
\node at (1,1) {$R_{z_i,z_j}$};
\end{tikzpicture}&
\begin{tikzpicture}[scale=0.7]
\draw (0,0) to [out = 0, in = 180] (2,2);
\draw (0,2) to [out = 0, in = 180] (2,0);
\draw[fill=white] (0,0) circle (.35);
\draw[fill=white] (0,2) circle (.35);
\draw[fill=white] (2,0) circle (.35);
\draw[fill=white] (2,2) circle (.35);
\node at (0,0) {$+$};
\node at (0,2) {$+$};
\node at (2,2) {$-$};
\node at (2,0) {$-$};
\path[fill=white] (1,1) circle (.3);
\node at (1,1) {$R_{z_i,z_j}$};
\end{tikzpicture}\\
\hline
1&1
&\frac{qz_i+z_j'-(1+q)}{z_iz_j'-1}
&\frac{qz_i+z_j'-(1+q)}{q(z_iz_j'-1)}
&\frac{(1-qz_i)(1-z_j')}{z_iz_j'-1}
&\frac{(1-z_i)(q-z_j')}{q(z_iz_j'-1)}\\
\hline
\end{array}\]
\caption{Boltzmann weights for stochastic $\Gamma-\Delta$ vertex with spectral parameters $z_i$ and $z_j$}
\label{Rmatrix4}
\end{figure}

The following theorem gives the four sets of Yang-Baxter equations for the stochastic symplectic ice.

\begin{theorem}\label{YBE1}
For any $X,Y\in\{\Gamma,\Delta\}$ the following holds. Assume that $S$ is stochastic $X$ vertex with spectral parameter $z_i$, $T$ is stochastic $Y$ vertex with spectral parameter $z_j$, and $R$ is stochastic $X-Y$ vertex with spectral parameters $z_i,z_j$. Then the partition functions of the following two configurations are equal for any fixed combination of spins $a,b,c,d,e,f$.
\begin{equation}
\hfill
\begin{tikzpicture}[baseline=(current bounding box.center)]
  \draw (0,1) to [out = 0, in = 180] (2,3) to (4,3);
  \draw (0,3) to [out = 0, in = 180] (2,1) to (4,1);
  \draw (3,0) to (3,4);
  \draw[fill=white] (0,1) circle (.3);
  \draw[fill=white] (0,3) circle (.3);
  \draw[fill=white] (3,4) circle (.3);
  \draw[fill=white] (4,3) circle (.3);
  \draw[fill=white] (4,1) circle (.3);
  \draw[fill=white] (3,0) circle (.3);
  \draw[fill=white] (2,3) circle (.3);
  \draw[fill=white] (2,1) circle (.3);
  \draw[fill=white] (3,2) circle (.3);
  \node at (0,1) {$a$};
  \node at (0,3) {$b$};
  \node at (3,4) {$c$};
  \node at (4,3) {$d$};
  \node at (4,1) {$e$};
  \node at (3,0) {$f$};
  \node at (2,3) {$g$};
  \node at (3,2) {$h$};
  \node at (2,1) {$i$};
\filldraw[black] (3,3) circle (2pt);
\node at (3,3) [anchor=south west] {$S$};
\filldraw[black] (3,1) circle (2pt);
\node at (3,1) [anchor=north west] {$T$};
\filldraw[black] (1,2) circle (2pt);
\node at (1,2) [anchor=west] {$R$};
\end{tikzpicture}\qquad\qquad
\begin{tikzpicture}[baseline=(current bounding box.center)]
  \draw (0,1) to (2,1) to [out = 0, in = 180] (4,3);
  \draw (0,3) to (2,3) to [out = 0, in = 180] (4,1);
  \draw (1,0) to (1,4);
  \draw[fill=white] (0,1) circle (.3);
  \draw[fill=white] (0,3) circle (.3);
  \draw[fill=white] (1,4) circle (.3);
  \draw[fill=white] (4,3) circle (.3);
  \draw[fill=white] (4,1) circle (.3);
  \draw[fill=white] (1,0) circle (.3);
  \draw[fill=white] (2,3) circle (.3);
  \draw[fill=white] (1,2) circle (.3);
  \draw[fill=white] (2,1) circle (.3);
  \node at (0,1) {$a$};
  \node at (0,3) {$b$};
  \node at (1,4) {$c$};
  \node at (4,3) {$d$};
  \node at (4,1) {$e$};
  \node at (1,0) {$f$};
  \node at (2,3) {$j$};
  \node at (1,2) {$k$};
  \node at (2,1) {$l$};
\filldraw[black] (1,3) circle (2pt);
\node at (1,3) [anchor=south west] {$T$};
\filldraw[black] (1,1) circle (2pt);
\node at (1,1) [anchor=north west]{$S$};
\filldraw[black] (3,2) circle (2pt);
\node at (3,2) [anchor=west] {$R$};
\end{tikzpicture}
\end{equation}
\end{theorem}
\begin{proof}
There are in total $2^6=64$ possible combinations of the boundary spins $a,b,c,d,e,f$. These identities are checked using a SAGE program.
\end{proof}

\subsection{The caduceus relation}\label{Sect.2.2.5}

In addition to the Yang-Baxter equation, the stochastic symplectic ice also satisfies a further relation called the ``caduceus relation'', which plays an important role in deriving functional equations for the partition functions in Section \ref{Sect.2.3}. Namely, we have the following 

\begin{theorem}\label{caduceus}
Assume that $A$ is stochastic $\Gamma-\Gamma$ vertex, $B$ is stochastic $\Delta-\Delta$ vertex, $C$ is stochastic $\Delta-\Gamma$ vertex, and $D$ is stochastic $\Gamma-\Delta$ vertex. Also assume that the spectral parameters of the four vertices $A,B,C,D$ are all $z_i,z_j$. Denote by $Z(I_1(\epsilon_1,\epsilon_2,\epsilon_3,\epsilon_4))$ the partition function of the following configuration with fixed combination of spins $\epsilon_1,\epsilon_2,\epsilon_3,\epsilon_4$.
\begin{equation}
\label{eqn:caduceus1}
\hfill
I_1(\epsilon_1,\epsilon_2,\epsilon_3,\epsilon_4)=
\begin{tikzpicture}[baseline=(current bounding box.center)]
  \draw (0,0) to (0.6,0) to [out=0, in=180] (3,2);
  \draw (0,3) to (0.6,3) to [out=0, in=180] (3,1);
  \draw (0,1) to [out=0, in=180] (2.4,3) to (3,3) ;
  \draw (0,2) to [out=0, in=180] (2.4,0) to (3,0);
  \draw (3,2) arc(-90:90:0.5);
  \draw (3,0) arc(-90:90:0.5);
  \filldraw[black] (3.5,0.5) circle (2pt);
  \filldraw[black] (3.5,2.5) circle (2pt);
  \filldraw[black] (0.9,1.5) circle (2pt);
  \filldraw[black] (2.1,1.5) circle (2pt);
  \filldraw[black] (1.5,0.5) circle (2pt);
  \filldraw[black] (1.5,2.5) circle (2pt);
  \node at (0.9,1.5) [anchor=south] {$D$};
  \node at (2.1,1.5) [anchor=south] {$C$};
  \node at (1.5,0.5) [anchor=south] {$B$};
  \node at (1.5,2.5) [anchor=south] {$A$};
  \node at (0,0) [anchor=east] {$\epsilon_4$};
  \node at (0,1) [anchor=east] {$\epsilon_3$};
  \node at (0,3) [anchor=east] {$\epsilon_1$};
  \node at (0,2) [anchor=east] {$\epsilon_2$};
\end{tikzpicture}
\end{equation}
Also denote by $Z(I_2(\epsilon_1,\epsilon_2,\epsilon_3,\epsilon_4))$ the partition function of the following configuration with fixed combination of spins $\epsilon_1,\epsilon_2,\epsilon_3,\epsilon_4$.
\begin{equation}
\hfill
I_2(\epsilon_1,\epsilon_2,\epsilon_3,\epsilon_4)=
\begin{tikzpicture}[baseline=(current bounding box.center)]
  \draw (0,2) arc(-90:90:0.5);
  \draw (0,0) arc(-90:90:0.5);
  \filldraw[black] (0.5,0.5) circle (2pt);
  \filldraw[black] (0.5,2.5) circle (2pt);
  \node at (0,0) [anchor=east] {$\epsilon_4$};
  \node at (0,1) [anchor=east] {$\epsilon_3$};
  \node at (0,3) [anchor=east] {$\epsilon_1$};
  \node at (0,2) [anchor=east] {$\epsilon_2$};
\end{tikzpicture}
\end{equation}
Then for any fixed combination of spins $\epsilon_1,\epsilon_2,\epsilon_3,\epsilon_4$, and for both choices of cap weights given in Figures \ref{cap1}-\ref{cap2}, we have
\begin{equation}
    Z(I_1(\epsilon_1,\epsilon_2,\epsilon_3,\epsilon_4))=\frac{(qz_iz_j-1)(1-(q+1)(z_i+z_j)+(q^2+q+1)z_iz_j)}{q(z_i+z_j-(q+1)z_iz_j)^2}Z(I_2(\epsilon_1,\epsilon_2,\epsilon_3,\epsilon_4)).
\end{equation}
\end{theorem}
\begin{proof}
The are in total $2^4=16$ possible combinations of the boundary spins $\epsilon_1,\epsilon_2,\epsilon_3,\epsilon_4$. The identities are checked using a SAGE program.
\end{proof}

\subsection{Functional equations satisfied by the partition functions}\label{Sect.2.3}

In this section, we derive functional equations satisfied by the partition functions. The main result is the following

\begin{theorem}\label{MainT}
Let 
\begin{equation}
D_1(n,L,z)=\prod_{i=1}^n z_i^L \prod_{i=1}^n (1-(q+1)z_i+q z_i z_i'^{-1})
\end{equation}
and 
\begin{equation}
    D_2(n,L,z)=\prod_{i=1}^n z_i^L.
\end{equation}
Then $\frac{Z(\mathcal{S}_{n,L,\lambda,z})}{D_1(n,L,z)}$ and $\frac{Z(\mathcal{T}_{n,L,\lambda,z})}{D_2(n,L,z)}$ are invariant under any permutation of $z_1,\cdots,z_n$ and any interchange $z_i\leftrightarrow z_i'^{-1}$.
\end{theorem}

Theorem \ref{MainT} follows from Propositions \ref{Transposition}-\ref{Interchange} below. Proposition \ref{Transposition} gives a functional equation when $z_1,\cdots,z_n$ are permuted, and Proposition \ref{Interchange} gives another functional equation under the interchange $z_n \leftrightarrow  \frac{1}{z_n'}$. Note that $\frac{1}{z_i}+z_i'=q+1$ for every $1\leq i\leq n$.

\begin{proposition}\label{Transposition}
The partition functions of the two types of stochastic symplectic ice, namely, $Z(\mathcal{S}_{n,L,\lambda,z})$ and $Z(\mathcal{T}_{n,L,\lambda,z})$, are both invariant under any permutation of $z_1,\cdots,z_n$.
\end{proposition}

\begin{proposition}\label{Interchange}
Let $s_n z:=(z_1,\cdots,z_{n-1},\frac{1}{z_n'})$. Then we have
\begin{equation}
    Z(\mathcal{S}_{n,L,\lambda,s_nz})=(\frac{1}{z_n z_n'})^L\frac{1-(q+1)z_n'^{-1}+q z_n z_n'^{-1}}{1-(q+1)z_n+q z_n z_n'^{-1}} Z(\mathcal{S}_{n,L,\lambda,z}),
\end{equation}
\begin{equation}
    Z(\mathcal{T}_{n,L,\lambda,s_nz})=(\frac{1}{ z_n z_n'})^L  Z(\mathcal{T}_{n,L,\lambda,z}).
\end{equation}
\end{proposition}

The rest of this section is devoted to the proof of Propositions \ref{Transposition}-\ref{Interchange}. The proof of Proposition \ref{Transposition} is based on the Yang-Baxter equation and the caduceus relation. The proof of Proposition \ref{Interchange} is based on the Yang-Baxter equation and another relation called the ``fish relation'' (see Proposition \ref{fish} below).

\subsubsection{Proof of Proposition \ref{Transposition}}

Based on Theorems \ref{YBE1}-\ref{caduceus}, we give the proof of Proposition \ref{Transposition} as follows.

\begin{proof}[Proof of Proposition \ref{Transposition}]
We note that the symmetric group $S_n$ is generated by adjacent transpositions $(i,i+1)$ for $1\leq i\leq n-1$. Therefore it suffices to verify the invariance of the partition functions under the transposition of $z_i$ and $z_{i+1}$. We show the details below for $Z(\mathcal{S}_{n,L,\lambda,z})$. The argument for $Z(\mathcal{T}_{n,L,\lambda,z})$ is essentially the same.

We attach a braid to the left boundary of the rows $2i-1,2i,2i+1,2i+2$ of $\mathcal{S}_{n,L,\lambda,z}$, and obtain the following
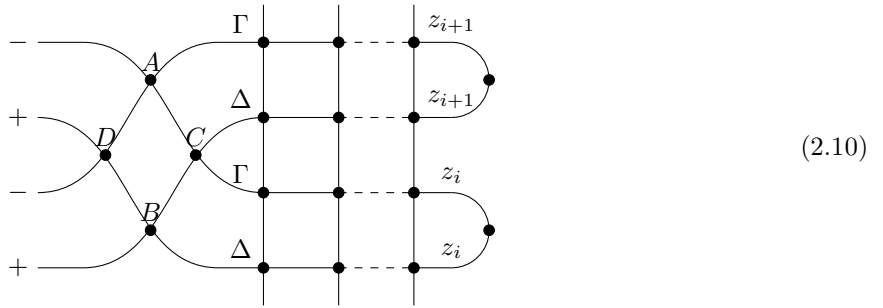
\begin{equation}
    \begin{tikzpicture}[baseline=(current bounding box.center)]
  \draw (0,0) to (0.6,0) to [out=0, in=180] (3,2) to (4,2);
  \draw (0,3) to (0.6,3) to [out=0, in=180] (3,1) to (4,1);
  \draw (0,1) to [out=0, in=180] (2.4,3) to (3,3) to (4,3);
  \draw (0,2) to [out=0, in=180] (2.4,0) to (3,0) to (4,0);
  \draw [dashed] (4,2) to (5,2);
  \draw [dashed] (4,1) to (5,1);
  \draw [dashed] (4,3) to (5,3);
  \draw [dashed] (4,0) to (5,0);
  \draw (5,2) to (5.5,2);
  \draw (5,1) to (5.5,1);
  \draw (5,3) to (5.5,3);
  \draw (5,0) to (5.5,0);
  \draw (5.5,2) arc(-90:90:0.5);
  \draw (5.5,0) arc(-90:90:0.5);
  \draw (3,-0.5) to (3,3.5);
  \draw (4,-0.5) to (4,3.5);
  \draw (5,-0.5) to (5,3.5);
  \filldraw[black] (6,0.5) circle (2pt);
  \filldraw[black] (6,2.5) circle (2pt);
  \filldraw[black] (0.9,1.5) circle (2pt);
  \filldraw[black] (2.1,1.5) circle (2pt);
  \filldraw[black] (1.5,0.5) circle (2pt);
  \filldraw[black] (1.5,2.5) circle (2pt);
  \filldraw[black] (3,0) circle (2pt);
  \filldraw[black] (4,0) circle (2pt);
  \filldraw[black] (5,0) circle (2pt);
  \filldraw[black] (3,1) circle (2pt);
  \filldraw[black] (4,1) circle (2pt);
  \filldraw[black] (5,1) circle (2pt);
  \filldraw[black] (3,2) circle (2pt);
  \filldraw[black] (4,2) circle (2pt);
  \filldraw[black] (5,2) circle (2pt);
  \filldraw[black] (3,3) circle (2pt);
  \filldraw[black] (4,3) circle (2pt);
  \filldraw[black] (5,3) circle (2pt);
  \node at (0.9,1.5) [anchor=south] {$D$};
  \node at (2.1,1.5) [anchor=south] {$C$};
  \node at (1.5,0.5) [anchor=south] {$B$};
  \node at (1.5,2.5) [anchor=south] {$A$};
  \node at (0,0) [anchor=east] {$+$};
  \node at (0,1) [anchor=east] {$-$};
  \node at (0,2) [anchor=east] {$+$};
  \node at (0,3) [anchor=east] {$-$};
  \node at (5.5,0) [anchor=south] {$z_i$};
  \node at (5.5,1) [anchor=south] {$z_i$};
  \node at (5.5,2) [anchor=south] {$z_{i+1}$};
  \node at (5.5,3) [anchor=south] {$z_{i+1}$};
  \node at (2.7,0) [anchor=south] {$\Delta$};
  \node at (2.7,1) [anchor=south] {$\Gamma$};
  \node at (2.7,2) [anchor=south] {$\Delta$};
  \node at (2.7,3) [anchor=south] {$\Gamma$};
\end{tikzpicture}
\end{equation}
where we have omitted the other rows of $\mathcal{S}_{n,L,\lambda,z}$. We denote by $Z(J_1)$ the partition function of this new ice model.

Note that the only admissible configuration of the braid is given as follows: 
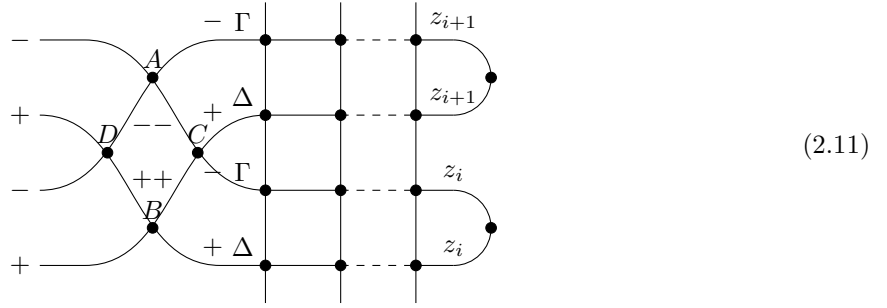
\begin{equation}
    \begin{tikzpicture}[baseline=(current bounding box.center)]
  \draw (0,0) to (0.6,0) to [out=0, in=180] (3,2) to (4,2);
  \draw (0,3) to (0.6,3) to [out=0, in=180] (3,1) to (4,1);
  \draw (0,1) to [out=0, in=180] (2.4,3) to (3,3) to (4,3);
  \draw (0,2) to [out=0, in=180] (2.4,0) to (3,0) to (4,0);
  \draw [dashed] (4,2) to (5,2);
  \draw [dashed] (4,1) to (5,1);
  \draw [dashed] (4,3) to (5,3);
  \draw [dashed] (4,0) to (5,0);
  \draw (5,2) to (5.5,2);
  \draw (5,1) to (5.5,1);
  \draw (5,3) to (5.5,3);
  \draw (5,0) to (5.5,0);
  \draw (5.5,2) arc(-90:90:0.5);
  \draw (5.5,0) arc(-90:90:0.5);
  \draw (3,-0.5) to (3,3.5);
  \draw (4,-0.5) to (4,3.5);
  \draw (5,-0.5) to (5,3.5);
  \filldraw[black] (6,0.5) circle (2pt);
  \filldraw[black] (6,2.5) circle (2pt);
  \filldraw[black] (0.9,1.5) circle (2pt);
  \filldraw[black] (2.1,1.5) circle (2pt);
  \filldraw[black] (1.5,0.5) circle (2pt);
  \filldraw[black] (1.5,2.5) circle (2pt);
  \filldraw[black] (3,0) circle (2pt);
  \filldraw[black] (4,0) circle (2pt);
  \filldraw[black] (5,0) circle (2pt);
  \filldraw[black] (3,1) circle (2pt);
  \filldraw[black] (4,1) circle (2pt);
  \filldraw[black] (5,1) circle (2pt);
  \filldraw[black] (3,2) circle (2pt);
  \filldraw[black] (4,2) circle (2pt);
  \filldraw[black] (5,2) circle (2pt);
  \filldraw[black] (3,3) circle (2pt);
  \filldraw[black] (4,3) circle (2pt);
  \filldraw[black] (5,3) circle (2pt);
  \node at (0.9,1.5) [anchor=south] {$D$};
  \node at (2.1,1.5) [anchor=south] {$C$};
  \node at (1.5,0.5) [anchor=south] {$B$};
  \node at (1.5,2.5) [anchor=south] {$A$};
  \node at (0,0) [anchor=east] {$+$};
  \node at (0,1) [anchor=east] {$-$};
  \node at (0,2) [anchor=east] {$+$};
  \node at (0,3) [anchor=east] {$-$};
  \node at (5.5,0) [anchor=south] {$z_i$};
  \node at (5.5,1) [anchor=south] {$z_i$};
  \node at (5.5,2) [anchor=south] {$z_{i+1}$};
  \node at (5.5,3) [anchor=south] {$z_{i+1}$};
  \node at (2.7,0) [anchor=south] {$\Delta$};
  \node at (2.7,1) [anchor=south] {$\Gamma$};
  \node at (2.7,2) [anchor=south] {$\Delta$};
  \node at (2.7,3) [anchor=south] {$\Gamma$};
  \node at (1.1,2.1) [anchor=north west] {$-$};
  \node at (1.9,2.1) [anchor=north east] {$-$};
  \node at (1.1,0.9) [anchor=south west] {$+$};
  \node at (1.9,0.9) [anchor=south east] {$+$};
  \node at (2.3,3) [anchor=south] {$-$};
  \node at (2.3,1.8) [anchor=south] {$+$};
  \node at (2.3,1) [anchor=south] {$-$};
  \node at (2.3,0) [anchor=south] {$+$};
\end{tikzpicture}
\end{equation}
Therefore $Z(J_1)$ is the product of the partition function of the braid and $Z(\mathcal{S}_{n,L,\lambda,z})$. Let
\begin{equation}
    L(z,q,i)=\frac{(qz_{i}z_{i+1}-1)(1-(q+1)(z_i+z_{i+1})+(q^2+q+1)z_iz_{i+1})}{q(z_i+z_{i+1}-(q+1)z_iz_{i+1})^2}.
\end{equation}
By computation, we obtain that
\begin{equation}\label{E13}
    Z(J_1)=L(z,q,i) Z(\mathcal{S}_{n,L,\lambda,z}).
\end{equation}

Now using the four sets of Yang-Baxter equations (Theorem \ref{YBE1}), we can move the four vertices (in the order of $C,A,B,D$) of the braid to the right without changing the partition function. Namely, if we denote by $Z(J_2)$ the partition function of the following
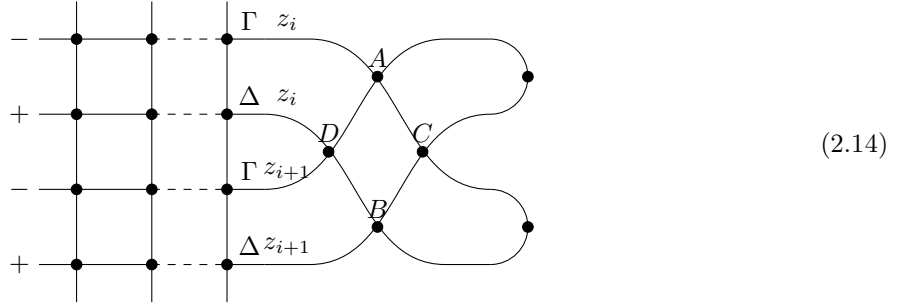
\begin{equation}
    \begin{tikzpicture}[baseline=(current bounding box.center)]
  \draw (2.5,0) to (3.1,0) to [out=0, in=180] (5.5,2);
  \draw (2.5,3) to (3.1,3) to [out=0, in=180] (5.5,1);
  \draw (2.5,1) to [out=0, in=180] (4.9,3) to (5.5,3);
  \draw (2.5,2) to [out=0, in=180] (4.9,0) to (5.5,0);
  \draw (-0.5,2) to (1,2);
  \draw (-0.5,1) to (1,1);
  \draw (-0.5,3) to (1,3);
  \draw (-0.5,0) to (1,0);
  \draw [dashed] (1,2) to (2,2);
  \draw [dashed] (1,1) to (2,1);
  \draw [dashed] (1,3) to (2,3);
  \draw [dashed] (1,0) to (2,0);
  \draw (2,2) to (2.5,2);
  \draw (2,1) to (2.5,1);
  \draw (2,3) to (2.5,3);
  \draw (2,0) to (2.5,0);
  \draw (5.5,2) arc(-90:90:0.5);
  \draw (5.5,0) arc(-90:90:0.5);
  \draw (0,-0.5) to (0,3.5);
  \draw (1,-0.5) to (1,3.5);
  \draw (2,-0.5) to (2,3.5);
  \filldraw[black] (6,0.5) circle (2pt);
  \filldraw[black] (6,2.5) circle (2pt);
  \filldraw[black] (3.35,1.5) circle (2pt);
  \filldraw[black] (4.6,1.5) circle (2pt);
  \filldraw[black] (4,0.5) circle (2pt);
  \filldraw[black] (4,2.5) circle (2pt);
  \filldraw[black] (0,0) circle (2pt);
  \filldraw[black] (1,0) circle (2pt);
  \filldraw[black] (2,0) circle (2pt);
  \filldraw[black] (0,1) circle (2pt);
  \filldraw[black] (1,1) circle (2pt);
  \filldraw[black] (2,1) circle (2pt);
  \filldraw[black] (0,2) circle (2pt);
  \filldraw[black] (1,2) circle (2pt);
  \filldraw[black] (2,2) circle (2pt);
  \filldraw[black] (0,3) circle (2pt);
  \filldraw[black] (1,3) circle (2pt);
  \filldraw[black] (2,3) circle (2pt);
  \node at (3.35,1.5) [anchor=south] {$D$};
  \node at (4.6,1.5) [anchor=south] {$C$};
  \node at (4,0.5) [anchor=south] {$B$};
  \node at (4,2.5) [anchor=south] {$A$};
  \node at (-0.5,0) [anchor=east] {$+$};
  \node at (-0.5,1) [anchor=east] {$-$};
  \node at (-0.5,2) [anchor=east] {$+$};
  \node at (-0.5,3) [anchor=east] {$-$};
  \node at (2.8,0) [anchor=south] {$z_{i+1}$};
  \node at (2.8,1) [anchor=south] {$z_{i+1}$};
  \node at (2.8,2) [anchor=south] {$z_{i}$};
  \node at (2.8,3) [anchor=south] {$z_{i}$};
  \node at (2.3,0) [anchor=south] {$\Delta$};
  \node at (2.3,1) [anchor=south] {$\Gamma$};
  \node at (2.3,2) [anchor=south] {$\Delta$};
  \node at (2.3,3) [anchor=south] {$\Gamma$};
\end{tikzpicture}
\end{equation}
then we have that 
\begin{equation}\label{E12}
    Z(J_1)=Z(J_2)
\end{equation}

Let $J_3(\epsilon_1,\epsilon_2,\epsilon_3,\epsilon_4)$ be given as follows, and recall the definition of $I_1(\epsilon_1,\epsilon_2,\epsilon_3,\epsilon_4)$ and $I_2(\epsilon_1,\epsilon_2,\epsilon_3,\epsilon_4)$ from the statement of Theorem \ref{caduceus} (taking the spectral parameters to be $z_{i+1},z_i$).
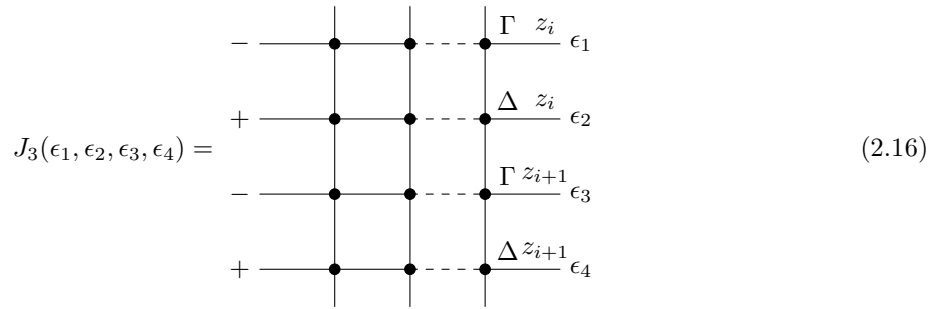
\begin{equation}
J_3(\epsilon_1,\epsilon_2,\epsilon_3,\epsilon_4)=
\begin{tikzpicture}[baseline=(current bounding box.center)]
  \draw (-1,2) to (1,2);
  \draw (-1,1) to (1,1);
  \draw (-1,3) to (1,3);
  \draw (-1,0) to (1,0);
  \draw [dashed] (1,2) to (2,2);
  \draw [dashed] (1,1) to (2,1);
  \draw [dashed] (1,3) to (2,3);
  \draw [dashed] (1,0) to (2,0);
  \draw (2,2) to (3,2);
  \draw (2,1) to (3,1);
  \draw (2,3) to (3,3);
  \draw (2,0) to (3,0);
  \draw (0,-0.5) to (0,3.5);
  \draw (1,-0.5) to (1,3.5);
  \draw (2,-0.5) to (2,3.5);
  \filldraw[black] (0,0) circle (2pt);
  \filldraw[black] (1,0) circle (2pt);
  \filldraw[black] (2,0) circle (2pt);
  \filldraw[black] (0,1) circle (2pt);
  \filldraw[black] (1,1) circle (2pt);
  \filldraw[black] (2,1) circle (2pt);
  \filldraw[black] (0,2) circle (2pt);
  \filldraw[black] (1,2) circle (2pt);
  \filldraw[black] (2,2) circle (2pt);
  \filldraw[black] (0,3) circle (2pt);
  \filldraw[black] (1,3) circle (2pt);
  \filldraw[black] (2,3) circle (2pt);
  \node at (-1,0) [anchor=east] {$+$};
  \node at (-1,1) [anchor=east] {$-$};
  \node at (-1,2) [anchor=east] {$+$};
  \node at (-1,3) [anchor=east] {$-$};
  \node at (2.8,0) [anchor=south] {$z_{i+1}$};
  \node at (2.8,1) [anchor=south] {$z_{i+1}$};
  \node at (2.8,2) [anchor=south] {$z_{i}$};
  \node at (2.8,3) [anchor=south] {$z_{i}$};
  \node at (2.3,0) [anchor=south] {$\Delta$};
  \node at (2.3,1) [anchor=south] {$\Gamma$};
  \node at (2.3,2) [anchor=south] {$\Delta$};
  \node at (2.3,3) [anchor=south] {$\Gamma$};
  \node at (3,0) [anchor=west] {$\epsilon_4$};
  \node at (3,1) [anchor=west] {$\epsilon_3$};
  \node at (3,2) [anchor=west] {$\epsilon_2$};
  \node at (3,3) [anchor=west] {$\epsilon_1$};
\end{tikzpicture}
\end{equation}

Now note that by Theorem \ref{caduceus},
\begin{eqnarray}\label{E11}
    Z(J_2)&=&\sum_{\epsilon_1,\epsilon_2,\epsilon_3,\epsilon_4\in\{-,+\}}Z(I_1(\epsilon_1,\epsilon_2,\epsilon_3,\epsilon_4))Z(J_3(\epsilon_1,\epsilon_2,\epsilon_3,\epsilon_4))\nonumber\\
    &=& L(z,q,i)\sum_{\epsilon_1,\epsilon_2,\epsilon_3,\epsilon_4\in\{-,+\}}Z(I_2(\epsilon_1,\epsilon_2,\epsilon_3,\epsilon_4))Z(J_3(\epsilon_1,\epsilon_2,\epsilon_3,\epsilon_4))\nonumber\\
    &=& L(z,q,i)Z(\mathcal{S}_{n,L,\lambda,s_iz}),
\end{eqnarray}
where $s_iz$ is the vector obtained by interchanging $z_i,z_{i+1}$ from $z$.

By combining (\ref{E13}),(\ref{E12}),(\ref{E11}), we obtain that
\begin{equation}
    Z(\mathcal{S}_{n,L,\lambda,z})=Z(\mathcal{S}_{n,L,\lambda,s_i z}),
\end{equation}
which finishes the proof.
\end{proof}

\subsubsection{Proof of Proposition \ref{Interchange}}

Before the proof of Proposition \ref{Interchange}, we make the following observation. As all the boundary edges on the top carry the $+$ spin, we conclude that only the three states in Figure \ref{F0} are involved in the $2n$th row. Now we simultaneously change the sign of the spins in the $2n$th row (interchanging $-$ and $+$ spins), change the Boltzmann weights of the vertices in the $2n$th row to those in Figure \ref{F1}, and change the Boltzmann weights for the cap connecting the last two rows to those in Figure \ref{F2} or \ref{F3} (depending on the type of the stochastic symplectic ice). For each admissible state, the Boltzmann weight of each vertex in the $2n$th row is now scaled by a factor of $\frac{1}{q z_n}$. Therefore the partition functions of the new system, denoted by $Z(\mathcal{S}'_{n,L,\lambda,z})$ and $Z(\mathcal{T}'_{n,L,\lambda,z})$ respectively, satisfy the following
\begin{equation}
    Z(\mathcal{S}'_{n,L,\lambda,z})=\frac{1}{(q z_n)^L}Z(\mathcal{S}_{n,L,\lambda,z}),
\end{equation}
\begin{equation}
    Z(\mathcal{T}'_{n,L,\lambda,z})=\frac{1}{(q z_n)^L}Z(\mathcal{T}_{n,L,\lambda,z}).
\end{equation}

\begin{figure}[h]
\[
\begin{array}{|c|c|c|c|c|c|}
\hline
\gammaicen{+}{+}{+}{+} &
\gammaicen{-}{+}{-}{+} &
\gammaicen{-}{+}{+}{-}\\
\hline
   1 & q z_n & 1-q z_n \\
\hline\end{array}\]
\caption{Boltzmann weights involved in the $2n$th row}
\label{F0}
\end{figure}

\begin{figure}[h]
\[
\begin{array}{|c|c|c|c|c|c|}
\hline
  \tt{a}_1&\tt{a}_2&\tt{b}_1&\tt{b}_2&\tt{d}_1&\tt{d}_2\\
\hline
\gammaicen{+}{+}{+}{+} &
  \gammaicen{-}{-}{-}{-} &
  \gammaicen{+}{-}{+}{-} &
  \gammaicen{-}{+}{-}{+} &
  \gammaicen{-}{-}{+}{+} &
  \gammaicen{+}{+}{-}{-}\\
\hline
   1 & 1 & \frac{1}{z_n} & \frac{1}{q z_n} &\frac{1}{z_n}-1 & \frac{1}{qz_n}-1\\
\hline\end{array}\]
\caption{New Boltzmann weights for the $2n$th row}
\label{F1}
\end{figure}

\begin{figure}[h]
\[
\begin{array}{|c|c|c|c|c|c|}
\hline
\text{New cap} &\newcaps{+}{-} & \newcaps{-}{+} \\
\hline
\text{Boltzmann weight}  &  1 & 1 \\
\hline\end{array}\]
\caption{New Boltzmann weights for the cap connecting the last two rows: reflecting stochastic symplectic ice}
\label{F2}
\end{figure}

\begin{figure}[h]
\[
\begin{array}{|c|c|c|c|c|c|}
\hline
\text{New cap} &\newcaps{+}{+} & \newcaps{-}{-} \\
\hline
\text{Boltzmann weight}  &  1 & 1 \\
\hline\end{array}\]
\caption{New Boltzmann weights for the cap connecting the last two rows: absorbing-and-emitting stochastic symplectic ice}
\label{F3}
\end{figure}

The following lemma gives a new set of Yang-Baxter equations, which will be used in the proof of Proposition \ref{Interchange}.

\begin{lemma}\label{Lem1}
Assume that $t_1,t_2\in \mathbb{C}$. Also assume that the Boltzmann weights of $S$ are given by Figure \ref{Fig1}, the Boltzmann weights of $T$ are given by Figure \ref{Fig2}, and the Boltzmann weights of $R$ are given by Figure \ref{Fig3}. Then the partition functions of the following two configurations are equal for any fixed combination of spins $a,b,c,d,e,f$.
\begin{figure}[h]
\[
\begin{array}{|c|c|c|c|c|c|}
\hline
  \tt{a}_1&\tt{a}_2&\tt{b}_1&\tt{b}_2&\tt{d}_1&\tt{d}_2\\
\hline
\gammaicei{+}{+}{+}{+} &
  \gammaicei{-}{-}{-}{-} &
  \gammaicei{+}{-}{+}{-} &
  \gammaicei{-}{+}{-}{+} &
  \gammaicei{-}{-}{+}{+} &
  \gammaicei{+}{+}{-}{-}\\
\hline
   1 & 1 & q t_1 & t_1 & q t_1-1 & t_1-1\\
\hline\end{array}\]
\caption{Boltzmann weights for $S$: Lemma \ref{Lem1}}
\label{Fig1}

\end{figure}
\begin{figure}[h]
\[
\begin{array}{|c|c|c|c|c|c|}
\hline
  \tt{a}_1&\tt{a}_2&\tt{b}_1&\tt{b}_2&\tt{d}_1&\tt{d}_2\\
\hline
\gammaicei{+}{+}{+}{+} &
  \gammaicei{-}{-}{-}{-} &
  \gammaicei{+}{-}{+}{-} &
  \gammaicei{-}{+}{-}{+} &
  \gammaicei{-}{-}{+}{+} &
  \gammaicei{+}{+}{-}{-}\\
\hline
   1 & 1 & q t_2 & t_2 & 1-q t_2 & 1-t_2\\
\hline\end{array}\]
\caption{Boltzmann weights for $T$: Lemma \ref{Lem1}}
\label{Fig2}
\end{figure}

\begin{figure}[h]
\[\begin{array}{|c|c|c|c|c|c|}
\hline
  \tt{a}_1&\tt{a}_2&\tt{b}_1&\tt{b}_2&\tt{c}_1&\tt{c}_2\\
\hline
\begin{tikzpicture}[scale=0.7]
\draw (0,0) to [out = 0, in = 180] (2,2);
\draw (0,2) to [out = 0, in = 180] (2,0);
\draw[fill=white] (0,0) circle (.35);
\draw[fill=white] (0,2) circle (.35);
\draw[fill=white] (2,0) circle (.35);
\draw[fill=white] (2,2) circle (.35);
\node at (0,0) {$+$};
\node at (0,2) {$+$};
\node at (2,2) {$+$};
\node at (2,0) {$+$};
\node at (2,0) {$+$};
\end{tikzpicture}&
\begin{tikzpicture}[scale=0.7]
\draw (0,0) to [out = 0, in = 180] (2,2);
\draw (0,2) to [out = 0, in = 180] (2,0);
\draw[fill=white] (0,0) circle (.35);
\draw[fill=white] (0,2) circle (.35);
\draw[fill=white] (2,0) circle (.35);
\draw[fill=white] (2,2) circle (.35);
\node at (0,0) {$-$};
\node at (0,2) {$-$};
\node at (2,2) {$-$};
\node at (2,0) {$-$};
\end{tikzpicture}&
\begin{tikzpicture}[scale=0.7]
\draw (0,0) to [out = 0, in = 180] (2,2);
\draw (0,2) to [out = 0, in = 180] (2,0);
\draw[fill=white] (0,0) circle (.35);
\draw[fill=white] (0,2) circle (.35);
\draw[fill=white] (2,0) circle (.35);
\draw[fill=white] (2,2) circle (.35);
\node at (0,0) {$+$};
\node at (0,2) {$-$};
\node at (2,2) {$+$};
\node at (2,0) {$-$};
\end{tikzpicture}&
\begin{tikzpicture}[scale=0.7]
\draw (0,0) to [out = 0, in = 180] (2,2);
\draw (0,2) to [out = 0, in = 180] (2,0);
\draw[fill=white] (0,0) circle (.35);
\draw[fill=white] (0,2) circle (.35);
\draw[fill=white] (2,0) circle (.35);
\draw[fill=white] (2,2) circle (.35);
\node at (0,0) {$-$};
\node at (0,2) {$+$};
\node at (2,2) {$-$};
\node at (2,0) {$+$};
\end{tikzpicture}&
\begin{tikzpicture}[scale=0.7]
\draw (0,0) to [out = 0, in = 180] (2,2);
\draw (0,2) to [out = 0, in = 180] (2,0);
\draw[fill=white] (0,0) circle (.35);
\draw[fill=white] (0,2) circle (.35);
\draw[fill=white] (2,0) circle (.35);
\draw[fill=white] (2,2) circle (.35);
\node at (0,0) {$-$};
\node at (0,2) {$+$};
\node at (2,2) {$+$};
\node at (2,0) {$-$};
\end{tikzpicture}&
\begin{tikzpicture}[scale=0.7]
\draw (0,0) to [out = 0, in = 180] (2,2);
\draw (0,2) to [out = 0, in = 180] (2,0);
\draw[fill=white] (0,0) circle (.35);
\draw[fill=white] (0,2) circle (.35);
\draw[fill=white] (2,0) circle (.35);
\draw[fill=white] (2,2) circle (.35);
\node at (0,0) {$+$};
\node at (0,2) {$-$};
\node at (2,2) {$-$};
\node at (2,0) {$+$};
\end{tikzpicture}\\
\hline
1&1
&\frac{t_2-t_1}{1-(q+1)t_1+q t_1 t_2}
&\frac{q(t_2-t_1)}{1-(q+1)t_1+q t_1 t_2}
&-\frac{(1-t_2)(1-q t_1)}{1-(q+1)t_1+q t_1 t_2}
&-\frac{(1-t_1)(1-q t_2)}{1-(q+1)t_1+q t_1 t_2}\\
\hline
\end{array}\]
\caption{Boltzmann weights for $R$: Lemma \ref{Lem1}}
\label{Fig3}
\end{figure}

\begin{equation}
\label{eqn:ybenew}
\hfill
\begin{tikzpicture}[baseline=(current bounding box.center)]
  \draw (0,1) to [out = 0, in = 180] (2,3) to (4,3);
  \draw (0,3) to [out = 0, in = 180] (2,1) to (4,1);
  \draw (3,0) to (3,4);
  \draw[fill=white] (0,1) circle (.3);
  \draw[fill=white] (0,3) circle (.3);
  \draw[fill=white] (3,4) circle (.3);
  \draw[fill=white] (4,3) circle (.3);
  \draw[fill=white] (4,1) circle (.3);
  \draw[fill=white] (3,0) circle (.3);
  \draw[fill=white] (2,3) circle (.3);
  \draw[fill=white] (2,1) circle (.3);
  \draw[fill=white] (3,2) circle (.3);
  \node at (0,1) {$a$};
  \node at (0,3) {$b$};
  \node at (3,4) {$c$};
  \node at (4,3) {$d$};
  \node at (4,1) {$e$};
  \node at (3,0) {$f$};
  \node at (2,3) {$g$};
  \node at (3,2) {$h$};
  \node at (2,1) {$i$};
\filldraw[black] (3,3) circle (2pt);
\node at (3,3) [anchor=south west] {$S$};
\filldraw[black] (3,1) circle (2pt);
\node at (3,1) [anchor=north west] {$T$};
\filldraw[black] (1,2) circle (2pt);
\node at (1,2) [anchor=west] {$R$};
\end{tikzpicture}\qquad\qquad
\begin{tikzpicture}[baseline=(current bounding box.center)]
  \draw (0,1) to (2,1) to [out = 0, in = 180] (4,3);
  \draw (0,3) to (2,3) to [out = 0, in = 180] (4,1);
  \draw (1,0) to (1,4);
  \draw[fill=white] (0,1) circle (.3);
  \draw[fill=white] (0,3) circle (.3);
  \draw[fill=white] (1,4) circle (.3);
  \draw[fill=white] (4,3) circle (.3);
  \draw[fill=white] (4,1) circle (.3);
  \draw[fill=white] (1,0) circle (.3);
  \draw[fill=white] (2,3) circle (.3);
  \draw[fill=white] (1,2) circle (.3);
  \draw[fill=white] (2,1) circle (.3);
  \node at (0,1) {$a$};
  \node at (0,3) {$b$};
  \node at (1,4) {$c$};
  \node at (4,3) {$d$};
  \node at (4,1) {$e$};
  \node at (1,0) {$f$};
  \node at (2,3) {$j$};
  \node at (1,2) {$k$};
  \node at (2,1) {$l$};
\filldraw[black] (1,3) circle (2pt);
\node at (1,3) [anchor=south west] {$T$};
\filldraw[black] (1,1) circle (2pt);
\node at (1,1) [anchor=north west]{$S$};
\filldraw[black] (3,2) circle (2pt);
\node at (3,2) [anchor=west] {$R$};
\end{tikzpicture}
\end{equation}
\end{lemma}
\begin{proof}
There are $2^6=64$ possible combinations of boundary spins. We have checked the identities using a SAGE program.
\end{proof}

Now consider the R-matrix with Boltzmann weights given by Figure \ref{Rm}. It is obtained by taking $t_1=\frac{1}{q z_n}$ and $t_2=\frac{1}{q}z_n'$ in the Boltzmann weights from Figure \ref{Fig3}. The following theorem gives the ``fish relation'' satisfied by the new R-matrix and the new cap. 

\begin{figure}[h]
\[\begin{array}{|c|c|c|c|c|c|}
\hline
  \tt{a}_1&\tt{a}_2&\tt{b}_1&\tt{b}_2&\tt{c}_1&\tt{c}_2\\
\hline
\begin{tikzpicture}[scale=0.7]
\draw (0,0) to [out = 0, in = 180] (2,2);
\draw (0,2) to [out = 0, in = 180] (2,0);
\draw[fill=white] (0,0) circle (.35);
\draw[fill=white] (0,2) circle (.35);
\draw[fill=white] (2,0) circle (.35);
\draw[fill=white] (2,2) circle (.35);
\node at (0,0) {$+$};
\node at (0,2) {$+$};
\node at (2,2) {$+$};
\node at (2,0) {$+$};
\node at (2,0) {$+$};
\end{tikzpicture}&
\begin{tikzpicture}[scale=0.7]
\draw (0,0) to [out = 0, in = 180] (2,2);
\draw (0,2) to [out = 0, in = 180] (2,0);
\draw[fill=white] (0,0) circle (.35);
\draw[fill=white] (0,2) circle (.35);
\draw[fill=white] (2,0) circle (.35);
\draw[fill=white] (2,2) circle (.35);
\node at (0,0) {$-$};
\node at (0,2) {$-$};
\node at (2,2) {$-$};
\node at (2,0) {$-$};
\end{tikzpicture}&
\begin{tikzpicture}[scale=0.7]
\draw (0,0) to [out = 0, in = 180] (2,2);
\draw (0,2) to [out = 0, in = 180] (2,0);
\draw[fill=white] (0,0) circle (.35);
\draw[fill=white] (0,2) circle (.35);
\draw[fill=white] (2,0) circle (.35);
\draw[fill=white] (2,2) circle (.35);
\node at (0,0) {$+$};
\node at (0,2) {$-$};
\node at (2,2) {$+$};
\node at (2,0) {$-$};
\end{tikzpicture}&
\begin{tikzpicture}[scale=0.7]
\draw (0,0) to [out = 0, in = 180] (2,2);
\draw (0,2) to [out = 0, in = 180] (2,0);
\draw[fill=white] (0,0) circle (.35);
\draw[fill=white] (0,2) circle (.35);
\draw[fill=white] (2,0) circle (.35);
\draw[fill=white] (2,2) circle (.35);
\node at (0,0) {$-$};
\node at (0,2) {$+$};
\node at (2,2) {$-$};
\node at (2,0) {$+$};
\end{tikzpicture}&
\begin{tikzpicture}[scale=0.7]
\draw (0,0) to [out = 0, in = 180] (2,2);
\draw (0,2) to [out = 0, in = 180] (2,0);
\draw[fill=white] (0,0) circle (.35);
\draw[fill=white] (0,2) circle (.35);
\draw[fill=white] (2,0) circle (.35);
\draw[fill=white] (2,2) circle (.35);
\node at (0,0) {$-$};
\node at (0,2) {$+$};
\node at (2,2) {$+$};
\node at (2,0) {$-$};
\end{tikzpicture}&
\begin{tikzpicture}[scale=0.7]
\draw (0,0) to [out = 0, in = 180] (2,2);
\draw (0,2) to [out = 0, in = 180] (2,0);
\draw[fill=white] (0,0) circle (.35);
\draw[fill=white] (0,2) circle (.35);
\draw[fill=white] (2,0) circle (.35);
\draw[fill=white] (2,2) circle (.35);
\node at (0,0) {$+$};
\node at (0,2) {$-$};
\node at (2,2) {$-$};
\node at (2,0) {$+$};
\end{tikzpicture}\\
\hline
1&1
&\frac{z_n z_n'-1}{q z_n+z_n'-(q+1)}
&\frac{q(z_nz_n'-1)}{q z_n+z_n'-(q+1)}
&-\frac{(z_n-1)(q-z_n')}{q z_n+z_n'-(q+1)}
&-\frac{(q z_n-1)(1-z_n')}{q z_n+z_n'-(q+1)}\\
\hline
\end{array}\]
\caption{The R-matrix used in the proof of Proposition \ref{Interchange}}
\label{Rm}
\end{figure}

\begin{proposition}\label{fish}
Suppose the Boltzmann weights of $R$ in the following is given by Figure \ref{Rm}. Denote by $Z(I_3(\epsilon_1,\epsilon_2))$ the partition function of the following system.
\begin{equation}
\hfill
I_3(\epsilon_1,\epsilon_2)=
\begin{tikzpicture}[baseline=(current bounding box.center)]
  \draw (0,0) to [out=0, in=-150] (1,0.5) to [out=30, in=180] (2,1);
  \draw (0,1) to [out=0, in=150] (1,0.5) to [out=-30, in=180] (2,0);
  \draw (2,0) to [out = 0, in = 180] (2.5,0.5);
  \draw (2,1) to [out = 0, in = 180] (2.5,0.5);
  \filldraw[black] (2.5,0.5) circle (2pt);
  \node at (0,0) [anchor=east] {$\epsilon_2$};
  \node at (0,1) [anchor=east] {$\epsilon_1$};
  \filldraw[black] (1,0.5) circle (2pt);
  \node at (1,0.5) [anchor=south] {$R$};
\end{tikzpicture}
\end{equation}
Also denote by $Z(I_4(\epsilon_1,\epsilon_2))$ the partition function of the following system.
\begin{equation}
\hfill
I_4(\epsilon_1,\epsilon_2)=
\begin{tikzpicture}[baseline=(current bounding box.center)]
  \draw (0,0) to [out = 0, in = 180] (0.5,0.5);
  \draw (0,1) to [out = 0, in = 180] (0.5,0.5);
  \filldraw[black] (0.5,0.5) circle (2pt);
  \node at (0,0) [anchor=east] {$\epsilon_2$};
  \node at (0,1) [anchor=east] {$\epsilon_1$};
\end{tikzpicture}
\end{equation}
Then for reflecting stochastic symplectic ice (i.e. the Boltzmann weights for the new cap are given by Figure \ref{F2}), we have
\begin{equation}
    Z(I_3(\epsilon_1,\epsilon_2))=-\frac{1-(q+1)z_n+q z_n z_n'^{-1}}{1-(q+1) z_n'^{-1}+q z_n z_n'^{-1}} Z(I_4(\epsilon_1,\epsilon_2));
\end{equation}
for absorbing-and-emitting stochastic symplectic ice (i.e. the Boltzmann weights for the new cap are given by Figure \ref{F3}), we have
\begin{equation}
    Z(I_3(\epsilon_1,\epsilon_2))= Z(I_4(\epsilon_1,\epsilon_2)).
\end{equation}
\end{proposition}
\begin{proof}
We denote by $a_1,a_2,b_1,b_2,c_1,c_2$ the Boltzmann weights for the R-matrix. 

First consider reflecting stochastic symplectic ice. In this case $Z(I_3(+,+))=Z(I_3(-,-))=Z(I_4(+,+))=Z(I_4(-,-))=0$. Moreover,
\begin{equation}
    Z(I_3(+,-))=c_1+b_2=-\frac{1-(q+1)z_n+q z_n z_n'^{-1}}{1-(q+1) z_n'^{-1}+q z_n z_n'^{-1}}Z(I_4(+,-)),
\end{equation}
\begin{equation}
    Z(I_3(-,+))=c_2+b_1=-\frac{1-(q+1)z_n+q z_n z_n'^{-1}}{1-(q+1) z_n'^{-1}+q z_n z_n'^{-1}}Z(I_4(-,+)).
\end{equation}

Now consider absorbing-and-emitting stochastic symplectic ice. In this case $Z(I_3(+,-))=Z(I_3(-,+))=Z(I_4(+,-))=Z(I_4(-,+))=0$. Moreover,
\begin{equation}
    Z(I_3(+,+))=a_1=Z(I_4(+,+)),
\end{equation}
\begin{equation}
    Z(I_3(-,-))=a_2=Z(I_4(-,-)).
\end{equation}
\end{proof}

We finish the proof of Proposition \ref{Interchange} as follows.

\begin{proof}[Proof of Proposition \ref{Interchange}]
For ease of notations, we denote by $V(a_1,a_2,b_1,b_2,c_1,c_2,d_1,d_2)$ a vertex with Boltzmann weights given by $a_1,a_2,b_1,b_2,c_1,c_2,d_1,d_2$.

Now we attach the R-matrix given by Figure \ref{Rm} to the left boundary of the last two rows of the changed system:
\begin{equation}
    \begin{tikzpicture}[baseline=(current bounding box.center)]
  \draw (0,0) to [out=0, in=-150] (1,0.5) to [out=30, in=180] (2,1);
  \draw (0,1) to [out=0, in=150] (1,0.5) to [out=-30, in=180] (2,0);
  \draw (2,0) to (3.5,0);
  \draw  [dashed] (3.5,0) to (4.5,0);
  \draw (4.5,0) to (5.5,0);
  \draw (2,1) to (3.5,1);
  \draw  [dashed] (3.5,1) to (4.5,1);
  \draw (4.5,1) to (5.5,1);
  \draw (5.5,0) to [out = 0, in = 180] (6,0.5);
  \draw (5.5,1) to [out = 0, in = 180] (6,0.5);
  \draw (2.5,-0.5) to (2.5,1.5);
  \draw (3.5,-0.5) to (3.5,1.5);
  \draw (4.5,-0.5) to (4.5,1.5);
  \filldraw[black] (2.5,0) circle (2pt);
  \filldraw[black] (3.5,0) circle (2pt);
  \filldraw[black] (4.5,0) circle (2pt);
  \filldraw[black] (2.5,1) circle (2pt);
  \filldraw[black] (3.5,1) circle (2pt);
  \filldraw[black] (4.5,1) circle (2pt);
  \filldraw[black] (6,0.5) circle (2pt);
  \node at (0,0) [anchor=east] {$+$};
  \node at (0,1) [anchor=east] {$+$};
  \filldraw[black] (1,0.5) circle (2pt);
  \node at (1,0.5) [anchor=south] {$R$};
  \node at (8,1) [anchor=south] {$V(1,1,\frac{1}{z_n},\frac{1}{q z_n},0,0,\frac{1}{z_n}-1,\frac{1}{qz_n}-1)$};
  \node at (8,0) [anchor=north] {$V(1,1,z_n',\frac{1}{q}z_n',0,0,1-z_n',1-\frac{1}{q}z_n')$};
  \node at (0,0) [anchor=east] {$+$};
\end{tikzpicture}
\end{equation}
Note that the only admissible configuration of the R-matrix is given by
\begin{equation}
    \begin{tikzpicture}[baseline=(current bounding box.center)]
  \draw (0,0) to [out=0, in=-150] (1,0.5) to [out=30, in=180] (2,1);
  \draw (0,1) to [out=0, in=150] (1,0.5) to [out=-30, in=180] (2,0);
  \draw (2,0) to (3.5,0);
  \draw  [dashed] (3.5,0) to (4.5,0);
  \draw (4.5,0) to (5.5,0);
  \draw (2,1) to (3.5,1);
  \draw  [dashed] (3.5,1) to (4.5,1);
  \draw (4.5,1) to (5.5,1);
  \draw (5.5,0) to [out = 0, in = 180] (6,0.5);
  \draw (5.5,1) to [out = 0, in = 180] (6,0.5);
  \draw (2.5,-0.5) to (2.5,1.5);
  \draw (3.5,-0.5) to (3.5,1.5);
  \draw (4.5,-0.5) to (4.5,1.5);
  \filldraw[black] (2.5,0) circle (2pt);
  \filldraw[black] (3.5,0) circle (2pt);
  \filldraw[black] (4.5,0) circle (2pt);
  \filldraw[black] (2.5,1) circle (2pt);
  \filldraw[black] (3.5,1) circle (2pt);
  \filldraw[black] (4.5,1) circle (2pt);
  \filldraw[black] (6,0.5) circle (2pt);
  \node at (0,0) [anchor=east] {$+$};
  \node at (0,1) [anchor=east] {$+$};
  \filldraw[black] (1,0.5) circle (2pt);
  \node at (1,0.5) [anchor=south] {$R$};
  \node at (8,1) [anchor=south] {$V(1,1,\frac{1}{z_n},\frac{1}{q z_n},0,0,\frac{1}{z_n}-1,\frac{1}{qz_n}-1)$};
  \node at (8,0) [anchor=north] {$V(1,1,z_n',\frac{1}{q}z_n',0,0,1-z_n',1-\frac{1}{q}z_n')$};
  \node at (0,0) [anchor=east] {$+$};
  \node at (2,0) [anchor=north] {$+$};
  \node at (2,1) [anchor=south] {$+$};
\end{tikzpicture}
\end{equation}
Therefore, the partition function of the above system is equal to $Z(\mathcal{S}'_{n,L,\lambda,z})$ or $Z(\mathcal{T}'_{n,L,\lambda,z})$ according to the type of the stochastic symplectic ice.

By Lemma \ref{Lem1}, the R-matrix can be pushed to the right without changing the partition function. That is, the partition function of the above system is equal to the partition function of the following
\begin{equation}
\begin{tikzpicture}[baseline=(current bounding box.center)]
\draw (3.5,0) to [out=0, in=-150] (4.5,0.5) to [out=30, in=180] (5.5,1);
 \draw (3.5,1) to [out=0, in=150] (4.5,0.5) to [out=-30, in=180] (5.5,0);
  \draw (0,0) to (1.5,0);
  \draw  [dashed] (1.5,0) to (2.5,0);
  \draw (2.5,0) to (3.5,0);
  \draw (0,1) to (1.5,1);
  \draw  [dashed] (1.5,1) to (2.5,1);
  \draw (2.5,1) to (3.5,1);
  \draw (5.5,0) to [out = 0, in = 180] (6,0.5);
  \draw (5.5,1) to [out = 0, in = 180] (6,0.5);
  \draw (0.5,-0.5) to (0.5,1.5);
  \draw (1.5,-0.5) to (1.5,1.5);
  \draw (2.5,-0.5) to (2.5,1.5);
  \filldraw[black] (0.5,0) circle (2pt);
  \filldraw[black] (1.5,0) circle (2pt);
  \filldraw[black] (2.5,0) circle (2pt);
  \filldraw[black] (0.5,1) circle (2pt);
  \filldraw[black] (1.5,1) circle (2pt);
  \filldraw[black] (2.5,1) circle (2pt);
  \filldraw[black] (6,0.5) circle (2pt);
  \node at (0,0) [anchor=east] {$+$};
  \node at (0,1) [anchor=east] {$+$};
  \filldraw[black] (4.5,0.5) circle (2pt);
  \node at (4.5,0.5) [anchor=south] {$R$};
  \node at (5.6,0) [anchor=north] {$V(1,1,\frac{1}{z_n},\frac{1}{q z_n},0,0,\frac{1}{z_n}-1,\frac{1}{qz_n}-1)$};
  \node at (5.6,1) [anchor=south] {$V(1,1,z_n',\frac{1}{q}z_n',0,0,1-z_n',1-\frac{1}{q}z_n')$};
  \node at (0,0) [anchor=east] {$+$};
\end{tikzpicture}
\end{equation}

Consider the reflecting stochastic symplectic ice. By Proposition \ref{fish}, the above partition function is equal to $-\frac{1-(q+1)z_n+q z_n z_n'^{-1}}{1-(q+1) z_n'^{-1}+q z_n z_n'^{-1}}$ times the partition function of the following system (denoted by $Z_1$)
\begin{equation}
    \begin{tikzpicture}[baseline=(current bounding box.center)]
  \draw (0,0) to (1.5,0);
  \draw  [dashed] (1.5,0) to (2.5,0);
  \draw (2.5,0) to (3.5,0);
  \draw (0,1) to (1.5,1);
  \draw  [dashed] (1.5,1) to (2.5,1);
  \draw (2.5,1) to (3.5,1);
  \draw (3.5,0) to [out = 0, in = 180] (4,0.5);
  \draw (3.5,1) to [out = 0, in = 180] (4,0.5);
  \draw (0.5,-0.5) to (0.5,1.5);
  \draw (1.5,-0.5) to (1.5,1.5);
  \draw (2.5,-0.5) to (2.5,1.5);
  \filldraw[black] (0.5,0) circle (2pt);
  \filldraw[black] (1.5,0) circle (2pt);
  \filldraw[black] (2.5,0) circle (2pt);
  \filldraw[black] (0.5,1) circle (2pt);
  \filldraw[black] (1.5,1) circle (2pt);
  \filldraw[black] (2.5,1) circle (2pt);
  \filldraw[black] (4,0.5) circle (2pt);
  \node at (0,0) [anchor=east] {$+$};
  \node at (0,1) [anchor=east] {$+$};
  \node at (5.6,0) [anchor=north] {$V(1,1,\frac{1}{z_n},\frac{1}{q z_n},0,0,\frac{1}{z_n}-1,\frac{1}{qz_n}-1)$};
  \node at (5.6,1) [anchor=south] {$V(1,1,z_n',\frac{1}{q}z_n',0,0,1-z_n',1-\frac{1}{q}z_n')$};
  \node at (0,0) [anchor=east] {$+$};
\end{tikzpicture}
\end{equation}

We note again that the top boundary edges of the system all carry $+$ spin. We change the sign of the spins of the $2n$th row again (interchanging $+$ and $-$), and also change the Boltzmann weights of the $2n$th row as in the following configuration. The Boltzmann weights for the cap connecting the last two rows are changed back to the original one given in Figure \ref{cap1}. We denote the partition function of the following system by $Z_2$.
\begin{equation}
    \begin{tikzpicture}[baseline=(current bounding box.center)]
  \draw (0,0) to (1.5,0);
  \draw  [dashed] (1.5,0) to (2.5,0);
  \draw (2.5,0) to (3.5,0);
  \draw (0,1) to (1.5,1);
  \draw  [dashed] (1.5,1) to (2.5,1);
  \draw (2.5,1) to (3.5,1);
  \draw (3.5,0) arc(-90:90:0.5);
  \draw (0.5,-0.5) to (0.5,1.5);
  \draw (1.5,-0.5) to (1.5,1.5);
  \draw (2.5,-0.5) to (2.5,1.5);
  \filldraw[black] (0.5,0) circle (2pt);
  \filldraw[black] (1.5,0) circle (2pt);
  \filldraw[black] (2.5,0) circle (2pt);
  \filldraw[black] (0.5,1) circle (2pt);
  \filldraw[black] (1.5,1) circle (2pt);
  \filldraw[black] (2.5,1) circle (2pt);
  \filldraw[black] (4,0.5) circle (2pt);
  \node at (0,0) [anchor=east] {$+$};
  \node at (0,1) [anchor=east] {$-$};
  \node at (5.6,0) [anchor=north] {$V(1,1,\frac{1}{z_n},\frac{1}{q z_n},0,0,\frac{1}{z_n}-1,\frac{1}{qz_n}-1)$};
  \node at (5.6,1) [anchor=south] {$V(1,1,\frac{1}{z_n'},\frac{q}{z_n'},\frac{q}{z_n'}-1,\frac{1}{z_n'}-1,0,0)$};
\end{tikzpicture}
\end{equation}
Similar to the previous argument, we conclude that 
\begin{equation}
    Z_2=(\frac{q}{z_n'})^L Z_1.
\end{equation}

Now note that the total number of $c_1,c_2,d_1,d_2$ patterns in the last two rows is an odd number (as can be seen by interpreting $-$ spins as paths and considering all possibilities). Hence $Z_2$ is equal to $-1$ times the partition function of the following configuration, which is $Z(\mathcal{S}_{n,L,\lambda,s_nz})$.
\begin{equation}
    \begin{tikzpicture}[baseline=(current bounding box.center)]
  \draw (0,0) to (1.5,0);
  \draw  [dashed] (1.5,0) to (2.5,0);
  \draw (2.5,0) to (3.5,0);
  \draw (0,1) to (1.5,1);
  \draw  [dashed] (1.5,1) to (2.5,1);
  \draw (2.5,1) to (3.5,1);
  \draw (3.5,0) arc(-90:90:0.5);
  \draw (0.5,-0.5) to (0.5,1.5);
  \draw (1.5,-0.5) to (1.5,1.5);
  \draw (2.5,-0.5) to (2.5,1.5);
  \filldraw[black] (0.5,0) circle (2pt);
  \filldraw[black] (1.5,0) circle (2pt);
  \filldraw[black] (2.5,0) circle (2pt);
  \filldraw[black] (0.5,1) circle (2pt);
  \filldraw[black] (1.5,1) circle (2pt);
  \filldraw[black] (2.5,1) circle (2pt);
  \filldraw[black] (4,0.5) circle (2pt);
  \node at (0,0) [anchor=east] {$+$};
  \node at (0,1) [anchor=east] {$-$};
  \node at (5.6,0) [anchor=north] {$V(1,1,\frac{1}{z_n},\frac{1}{q z_n},0,0,1-\frac{1}{z_n},1-\frac{1}{qz_n})$};
  \node at (5.6,1) [anchor=south] {$V(1,1,\frac{1}{z_n'},\frac{q}{z_n'},1-\frac{q}{z_n'},1-\frac{1}{z_n'},0,0)$};
\end{tikzpicture}
\end{equation}

Therefore we conclude that
\begin{equation}
    Z(\mathcal{S}_{n,L,\lambda,s_nz})=(\frac{1}{z_n z_n'})^L\frac{1-(q+1)z_n'^{-1}+q z_n z_n'^{-1}}{1-(q+1)z_n+q z_n z_n'^{-1}} Z(\mathcal{S}_{n,L,\lambda,z}).
\end{equation}

The conclusion for $Z(\mathcal{T}_{n,L,\lambda,s_nz})$ can be obtained similarly, noting that the number of $c_1,c_2,d_1,d_2$ patterns in the last two rows is an even number for this case.

\end{proof}

\section{Colored stochastic symplectic ice}\label{Sect.3}

In this section, we introduce a colored version of the stochastic symplectic ice model. For each edge of the rectangular lattice, instead of assigning either a $+$ or $-$ spin, we now associate either $+$ or one of $2n$ colors to it. The $2n$ colors are labeled by $[\pm n ]=\{\overline{n},\cdots,\overline{1},1,\cdots,n\}$.

The colored model is closely related to Cartan type C: part of the boundary conditions are specified by two elements $\sigma$ and $\tau$ of the hyperoctahedral group--the Weyl group of type C; the recursive relations for the partition function, upon a change of variables, are related to Demazure-Lusztig operators of type C.

We start by introducing the colored model in Section \ref{Sect.3.1}. Then we introduce the R-matrix and prove the Yang-Baxter equation in Section \ref{Sect.3.2}. In Section \ref{Sect.3.2.5}, we compute the partition function when $\sigma(i)=\overline{\tau(i)}$ for every $1\leq i\leq n$. Then we present a new relation, the reflection equation, in Section \ref{Sect.3.2.7}. By combing the Yang-Baxter equation and the reflection equation, we derive the recursive relations for the partition function in Section \ref{Sect.3.3}. The recursive relations are further related to Demazure-Lusztig operators of type C in Section \ref{Sect.3.4}. 

We briefly introduce the hyperoctahedral group--denoted by $B_n$--here. The hyperoctahedral group has the following presentation
\begin{eqnarray*}
  B_n&=&\langle s_1,\cdots,s_n|s_i^2=1, 1\leq i\leq n; (s_is_{i+1})^3=1, 1\leq i\leq n-2;\\
  &&(s_{n-1}s_n)^4=1; (s_is_j)^2=1,1\leq i<j\leq n, |i-j|>1\rangle.
\end{eqnarray*}
The group $B_n$ is the Weyl group for the root system of type $C_n$. Elements of $B_n$ can be viewed as permutations $\sigma$ of $[\pm n]$ such that $\sigma(-i)=-\sigma(i)$ for every $1\leq i\leq n$. 

We make the convention that elements of $B_n$ are multiplied from right to left, and that $\overline{\overline{i}}=i$ for each $1\leq i\leq n$. Thus for each $i\in [\pm n]$ and $\sigma,\tau\in B_n$, we have $\sigma\tau (i)=\sigma(\tau(i))$. Note that $s_i$ is the transposition $(i,i+1)$ for each $1\leq i\leq n-1$, and that $s_n(i)=i$ for $1\leq i\leq n-1$ and $s_n(n)=\overline{n}$. 

\subsection{The colored model}\label{Sect.3.1}
We introduce the colored version of the stochastic symplectic ice in this section. The main difference from the uncolored model is that now every edge of the lattice can either take $+$ or one of the $2n$ colors labeled by $[\pm n]$.

We denote the colors by $c_{\overline{n}},\cdots,c_n$. Hereafter we refer to $c_i$ and $c_{\overline{i}}$ as mutually opposite colors for every $1\leq i\leq n$. For convenience of notations, we also let $c_0:=+$. We take the following order on $[\pm n]\cup\{0\}$:
\begin{equation}
    \bar{n}<\cdots<\bar{1}<0<1<\cdots n.
\end{equation}

For the colored model, there are also two types of vertices. They are termed ``colored stochastic $\Gamma$ vertex'' and ``colored stochastic $\Delta$ vertex''. The model depends on $n$ spectral parameters $z_1,\cdots,z_n$ and a deformation parameter $q$. Again we take
\begin{equation*}
    z_i'=q+1-\frac{1}{z_i}.
\end{equation*}
The Boltzmann weights for the two types of vertices are listed in Figures \ref{Gamma}-\ref{Delta}. 

\begin{figure}[h]
\[
\begin{array}{|c|c|c|c|c|c|}
\hline
  \tt{a}_1\slash \tt{a}_2&\tt{a}_1\slash\tt{a}_2& \tt{b}_1&\tt{b}_2&\tt{c}_1&\tt{c}_2\\
\hline
\gammaice{c_{\alpha}}{c_{\alpha}}{c_{\alpha}}{c_{\alpha}} &
  \gammaice{c_{\beta}}{c_{\beta}}{c_{\beta}}{c_{\beta}} &
  \gammaice{c_{\alpha}}{c_{\beta}}{c_{\alpha}}{c_{\beta}} &
  \gammaice{c_{\beta}}{c_{\alpha}}{c_{\beta}}{c_{\alpha}} &
  \gammaice{c_{\beta}}{c_{\alpha}}{c_{\alpha}}{c_{\beta}} &
  \gammaice{c_{\alpha}}{c_{\beta}}{c_{\beta}}{c_{\alpha}}\\
\hline
   1 & 1 & z_i & q z_i &1-q z_i & 1-z_i\\
\hline\end{array}\]
\caption{Boltzmann weights for colored stochastic $\Gamma$ vertex with spectral parameter $z_i$, where $\alpha<\beta$}
\label{Gamma}
\end{figure}

\begin{figure}[h]
\[
\begin{array}{|c|c|c|c|c|c|}
\hline
  \tt{a}_1 \slash \tt{a}_2 &\tt{a}_1\slash \tt{a}_2&\tt{b}_1&\tt{b}_2&\tt{d}_1&\tt{d}_2\\
\hline
\gammaice{c_{\alpha}}{c_{\alpha}}{c_{\alpha}}{c_{\alpha}} &
  \gammaice{c_{\beta}}{c_{\beta}}{c_{\beta}}{c_{\beta}} &
  \gammaice{c_{\alpha}}{c_{\beta}}{c_{\alpha}}{c_{\beta}} &
  \gammaice{c_{\beta}}{c_{\alpha}}{c_{\beta}}{c_{\alpha}} &
  \gammaice{c_{\beta}}{c_{\beta}}{c_{\alpha}}{c_{\alpha}} &
  \gammaice{c_{\alpha}}{c_{\alpha}}{c_{\beta}}{c_{\beta}}\\
\hline
   1 & 1 & z_i' & \frac{1}{q} z_i' &  1-z_i' & 1-\frac{1}{q}z_i'\\
\hline\end{array}\]
\caption{Boltzmann weights for colored stochastic $\Delta$ vertex with spectral parameter $z_i$, where $\alpha<\beta$ and $z_i'=q+1-\frac{1}{z_i}$}
\label{Delta}
\end{figure}

The colored model consists of a rectangular lattice with $2n$ rows and $L$ columns. The rows are numbered $1,2,\cdots,2n$ from bottom to top, and the columns are numbered $1,2,\cdots,L$ from right to left. Every odd-numbered row is a row of colored stochastic $\Delta$ vertex, and every even-numbered row is a row of colored stochastic $\Gamma$ vertex. The spectral parameter for the $i$th row of colored stochastic $\Gamma$ vertices and the $i$th row of the colored stochastic $\Delta$ vertices is $z_i$. 

The model also depends on a partition $\lambda=(\lambda_1,\cdots,\lambda_n)\in\mathbb{Z}^n$ (with $\lambda_1\geq \cdots \geq \lambda_n$) and two elements $\sigma,\tau\in B_n$, where $B_n$ is the hyperoctahedral group as introduced previously. We assume that $L\geq \lambda_1+n$. The assignment of boundary conditions is given as follows: on the left column, we assign color $c_{\sigma(i)}$ to the $i$th row of colored stochastic $\Gamma$ vertex, and $+$ to each row of colored stochastic $\Delta$ vertex; on the top, we assign $+$ to each boundary edge; on the bottom, we assign color $c_{\tau(i)}$ to each column labeled $\lambda_i+n+1-i$ for $1\leq i\leq n$, and assign $+$ to the other columns. On the right, the $i$th row of colored stochastic $\Gamma$ vertex and the $i$th row of colored stochastic $\Delta$ vertex are connected with a cap. The Boltzmann weights for the caps are given in Figure \ref{ColorCap}. 

\begin{figure}[h]
\[
\begin{array}{|c|c|c|c|c|c|}
\hline
\text{Cap} &\caps{+}{+} & \caps{c_{\alpha}}{c_{\bar{\alpha}}} & \caps{c_{\bar{\alpha}}}{c_{\alpha}} \\
\hline
\text{Boltzmann weight}  &  1 & 1 & 1 \\
\hline\end{array}\]
\caption{Boltzmann weights of the caps for colored stochastic symplectic ice: where $\alpha\in \{1,2,\cdots,n\}$}
\label{ColorCap}
\end{figure}

Hereafter we denote by $\mathcal{S}_{n,L,\lambda,\sigma,\tau,z}$ the collection of admissible configurations with the above-specified data. We also denote by $Z(\mathcal{S}_{n,L,\lambda,\sigma,\tau,z})$ the corresponding partition function.

We note that the Boltzmann weights for both types of vertices and the caps are also stochastic, as in the uncolored case. When $z$ satisfies the condition (\ref{Co1}), a probabilistic interpretation for each vertex can similarly be obtained.

We also note that the colored model can be interpreted as an interacting particle system as in the uncolored case if condition (\ref{Co1}) is satisfied. The interpretation is similar to the uncolored case, except that now each particle carries a color, and that the updating rule for a particle depends on its color.

The detailed rule is as follows. When $t$ is even, the particles are ordered from left to right. There is a new particle entering from the left boundary with color $c_{\sigma(n-\frac{t}{2})}$ (we call it $0$th particle), which jumps to the right with geometric jump size (with parameter $q z_{n-\frac{t}{2}}$ if $\sigma(n-\frac{t}{2})>0$, or $z_{n-\frac{t}{2}}$ otherwise) unless it hits the $1$st particle; if the particle hits the $1$st particle, the updating rule will be described later. Starting from $l=1$, if the $l$th particle wasn't hit by any particle on its left, we flip a coin with head probability $z_{n-\frac{t}{2}}$ (if the color of the particle is $c_{\alpha}$ with $\alpha>0)$ or $q z_{n-\frac{t}{2}}$ (if $\alpha<0$) to determine whether it will stay at its current position; if the coin comes up tail, then the particle jumps to the right with geometric jump size (with parameter $q z_{n-\frac{t}{2}}$ if $\alpha>0$, or $z_{n-\frac{t}{2}}$ if $\alpha<0$) unless it hits the $(l+1)$th particle; if the particle hits the $(l+1)$th particle, the updating rule will be described later. If the $l$th particle was hit by the $(l-1)$th particle, either the $(l-1)$th or the $l$th particle (depending on the updating rule as will be described later) jumps to the right by $1$ and the following move is the same as the previous case except for the first step determining whether it will stay at the current position. Then the $(l+1)$th particle begins to move. If the rightmost particle moves beyond the first column (meaning that it hits the cap), it is reflected by the cap (meaning that it will start to move leftward from the first column at time $t+1$), and its color $c_{\alpha}$ is changed to $c_{\overline{\alpha}}$.

When $t$ is odd, the particles are ordered from right to left. If there is a particle reflected from the cap (we call it $0$th particle), it jumps to the left with geometric jump size (with parameter $\frac{1}{q}  z_{n-\frac{t-1}{2}}'$ if the color of the particle is $c_{\alpha}$ with $\alpha>0$, or $z_{n-\frac{t-1}{2}}'$ otherwise) unless it hits the $1$st particle; if the particle hits the $1$st particle, the updating rule will be described later. Starting from $l=1$, if the $l$th particle wasn't hit by any particle on its right, we flip a coin with head probability $z_{n-\frac{t-1}{2}}'$ (if the color of the particle is $c_{\alpha}$ with $\alpha>0$) or $\frac{1}{q}z_{n-\frac{t-1}{2}}'$ (if $\alpha<0$) to determine whether it will stay at its current position; if the coin comes up tail, then the particle jumps to the left with geometric jump size (with parameter $\frac{1}{q}  z_{n-\frac{t-1}{2}}'$ if $\alpha>0$, or $z_{n-\frac{t-1}{2}}'$ if $\alpha<0$) unless it hits the $(l+1)$th particle; if the particle hits the $(l+1)$th particle, the updating rule will be described later. If the $l$th particle was hit by the $(l-1)$th particle, either the $(l-1)$th or the $l$th particle (depending on the updating rule as will be described later) jumps to the left by $1$ and the following move is the same as the previous case except for the first step determining whether it will stay at the current position. Then the $(l+1)$th particle begins to move.

Now we describe the updating rule for the case when a particle hits another.
When the time $t$ is even, consider the situation when a particle of color $c_{\alpha}$ hits another particle of color $c_{\beta}$ from the left. If $\alpha<\beta$, with probability $z_{n-\frac{t}{2}}$ the two particles are swapped, with the particle of color $c_{\beta}$ staying at the original position and the other particle continuing to move; with probability $1-z_{n-\frac{t}{2}}$, the particle of color $c_{\alpha}$ stays at the current position and the other particle starts to move. If $\alpha>\beta$, with probability $q z_{n-\frac{t}{2}}$ the two particles are swapped, with the particle of color $c_{\beta}$ staying at the original position and the other particle continuing to move; with probability $1-q z_{n-\frac{t}{2}}$, the particle of color $c_{\alpha}$ stays at the current position and the other particle starts to move.

When the time $t$ is odd, we also consider the situation when a particle of color $c_{\alpha}$ hits another particle of color $c_{\beta}$ from the right. If $\alpha<\beta$, with probability $z_{n-\frac{t-1}{2}}'$ the two particles are swapped, with the particle of color $c_{\beta}$ staying at the original position and the other particle continuing to move; with probability $1-z_{n-\frac{t-1}{2}}'$, the particle of color $c_{\alpha}$ stays at the current position and the other particle starts to move. If $\alpha>\beta$, with probability $\frac{1}{q} z_{n-\frac{t-1}{2}}'$ the two particles are swapped, with the particle of color $c_{\beta}$ staying at the original position and the other particle continuing to move; with probability $1-\frac{1}{q} z_{n-\frac{t-1}{2}}'$, the particle of color $c_{\alpha}$ stays at the current position and the other particle starts to move. 

Under this probabilistic interpretation, the partition function $Z(\mathcal{S}_{n,L,\lambda,\sigma,\tau,z})$ represents the probability that (with the entering order of particle colors specified by $\sigma$) the particle configuration at time $t=2n$ is given by $\mu$ and $\tau$, with $\mu$ specifying the particle locations and $\tau$ specifying the particle colors.

\subsection{The R-matrix and the Yang-Baxter equation}\label{Sect.3.2}
For the colored stochastic symplectic ice, we find three sets of Yang-Baxter equations. The corresponding R-matrices are termed ``colored stochastic $\Gamma-\Gamma$ vertex'', ``colored stochastic $\Delta-\Gamma$ vertex'' and ``colored stochastic $\Delta-\Delta$ vertex''. In Section \ref{Sect.3.3} we will show that, when combined with the reflection equation, these three sets of Yang-Baxter equations are enough for us to derive the recursive relations for the partition functions.

Throughout the paper we denote the Boltzmann weights of an R-matrix of type $XY$ and spectral parameters $z_i,z_j$ as shown in Figure \ref{RmatrixC} by $R_{XY}(c_{\alpha},c_{\beta},c_{\gamma},c_{\delta}; z_i,z_j)$, where $(X,Y)\in\{(\Gamma,\Gamma),(\Delta,\Delta),(\Delta,\Gamma)\}$ and $\alpha,\beta,\gamma,\delta\in\{\overline{n},\cdots,\overline{1},0,1,\cdots,n\}$.
The Boltzmann weights for the three types of R-matrices are given in Figures \ref{RmatrixC1}-\ref{RmatrixC3}.

\begin{figure}[h]
\centering
 \begin{tikzpicture}[scale=0.7]
\draw (0,0) to [out = 0, in = 180] (2,2);
\draw (0,2) to [out = 0, in = 180] (2,0);
\draw[fill=white] (0,0) circle (.35);
\draw[fill=white] (0,2) circle (.35);
\draw[fill=white] (2,0) circle (.35);
\draw[fill=white] (2,2) circle (.35);
\node at (0,0) {$c_{\alpha}$};
\node at (0,2) {$c_{\beta}$};
\node at (2,2) {$c_{\gamma}$};
\node at (2,0) {$c_{\delta}$};
\path[fill=white] (1,1) circle (.3);
\node at (1,1) {$R_{z_i,z_j}$};
\end{tikzpicture}
\caption{R-matrix}
\label{RmatrixC}
\end{figure}
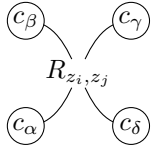

\begin{figure}[h]
\[\begin{array}{|c|c|c|c|c|c|}
\hline
\begin{tikzpicture}[scale=0.7]
\draw (0,0) to [out = 0, in = 180] (2,2);
\draw (0,2) to [out = 0, in = 180] (2,0);
\draw[fill=white] (0,0) circle (.35);
\draw[fill=white] (0,2) circle (.35);
\draw[fill=white] (2,0) circle (.35);
\draw[fill=white] (2,2) circle (.35);
\node at (0,0) {$c_{\alpha}$};
\node at (0,2) {$c_{\alpha}$};
\node at (2,2) {$c_{\alpha}$};
\node at (2,0) {$c_{\alpha}$};
\path[fill=white] (1,1) circle (.3);
\node at (1,1) {$R_{z_i,z_j}$};
\end{tikzpicture}&
\begin{tikzpicture}[scale=0.7]
\draw (0,0) to [out = 0, in = 180] (2,2);
\draw (0,2) to [out = 0, in = 180] (2,0);
\draw[fill=white] (0,0) circle (.35);
\draw[fill=white] (0,2) circle (.35);
\draw[fill=white] (2,0) circle (.35);
\draw[fill=white] (2,2) circle (.35);
\node at (0,0) {$c_{\beta}$};
\node at (0,2) {$c_{\beta}$};
\node at (2,2) {$c_{\beta}$};
\node at (2,0) {$c_{\beta}$};
\path[fill=white] (1,1) circle (.3);
\node at (1,1) {$R_{z_i,z_j}$};
\end{tikzpicture}&
\begin{tikzpicture}[scale=0.7]
\draw (0,0) to [out = 0, in = 180] (2,2);
\draw (0,2) to [out = 0, in = 180] (2,0);
\draw[fill=white] (0,0) circle (.35);
\draw[fill=white] (0,2) circle (.35);
\draw[fill=white] (2,0) circle (.35);
\draw[fill=white] (2,2) circle (.35);
\node at (0,0) {$c_{\alpha}$};
\node at (0,2) {$c_{\beta}$};
\node at (2,2) {$c_{\alpha}$};
\node at (2,0) {$c_{\beta}$};
\path[fill=white] (1,1) circle (.3);
\node at (1,1) {$R_{z_i,z_j}$};
\end{tikzpicture}&
\begin{tikzpicture}[scale=0.7]
\draw (0,0) to [out = 0, in = 180] (2,2);
\draw (0,2) to [out = 0, in = 180] (2,0);
\draw[fill=white] (0,0) circle (.35);
\draw[fill=white] (0,2) circle (.35);
\draw[fill=white] (2,0) circle (.35);
\draw[fill=white] (2,2) circle (.35);
\node at (0,0) {$c_{\beta}$};
\node at (0,2) {$c_{\alpha}$};
\node at (2,2) {$c_{\beta}$};
\node at (2,0) {$c_{\alpha}$};
\path[fill=white] (1,1) circle (.3);
\node at (1,1) {$R_{z_i,z_j}$};
\end{tikzpicture}&
\begin{tikzpicture}[scale=0.7]
\draw (0,0) to [out = 0, in = 180] (2,2);
\draw (0,2) to [out = 0, in = 180] (2,0);
\draw[fill=white] (0,0) circle (.35);
\draw[fill=white] (0,2) circle (.35);
\draw[fill=white] (2,0) circle (.35);
\draw[fill=white] (2,2) circle (.35);
\node at (0,0) {$c_{\beta}$};
\node at (0,2) {$c_{\alpha}$};
\node at (2,2) {$c_{\alpha}$};
\node at (2,0) {$c_{\beta}$};
\path[fill=white] (1,1) circle (.3);
\node at (1,1) {$R_{z_i,z_j}$};
\end{tikzpicture}&
\begin{tikzpicture}[scale=0.7]
\draw (0,0) to [out = 0, in = 180] (2,2);
\draw (0,2) to [out = 0, in = 180] (2,0);
\draw[fill=white] (0,0) circle (.35);
\draw[fill=white] (0,2) circle (.35);
\draw[fill=white] (2,0) circle (.35);
\draw[fill=white] (2,2) circle (.35);
\node at (0,0) {$c_{\alpha}$};
\node at (0,2) {$c_{\beta}$};
\node at (2,2) {$c_{\beta}$};
\node at (2,0) {$c_{\alpha}$};
\path[fill=white] (1,1) circle (.3);
\node at (1,1) {$R_{z_i,z_j}$};
\end{tikzpicture}\\
\hline
1&1
&\frac{z_i-z_j}{1-(q+1)z_j+q z_i z_j}
&\frac{q(z_i-z_j)}{1-(q+1)z_j+q z_i z_j}
&\frac{(1-q z_i)(1-z_j)}{1-(q+1)z_j+q z_i z_j}
&\frac{(1-z_i)(1-q z_j)}{1-(q+1)z_j+q z_i z_j}\\
\hline
\end{array}\]
\caption{Boltzmann weights for colored stochastic $\Gamma-\Gamma$ vertex with spectral parameters $z_i$ and $z_j$: where $\alpha<\beta$}
\label{RmatrixC1}
\end{figure}

\begin{figure}[h]
\[\begin{array}{|c|c|c|c|c|c|}
\hline
\begin{tikzpicture}[scale=0.7]
\draw (0,0) to [out = 0, in = 180] (2,2);
\draw (0,2) to [out = 0, in = 180] (2,0);
\draw[fill=white] (0,0) circle (.35);
\draw[fill=white] (0,2) circle (.35);
\draw[fill=white] (2,0) circle (.35);
\draw[fill=white] (2,2) circle (.35);
\node at (0,0) {$c_{\alpha}$};
\node at (0,2) {$c_{\alpha}$};
\node at (2,2) {$c_{\alpha}$};
\node at (2,0) {$c_{\alpha}$};
\path[fill=white] (1,1) circle (.3);
\node at (1,1) {$R_{z_i,z_j}$};
\end{tikzpicture}&
\begin{tikzpicture}[scale=0.7]
\draw (0,0) to [out = 0, in = 180] (2,2);
\draw (0,2) to [out = 0, in = 180] (2,0);
\draw[fill=white] (0,0) circle (.35);
\draw[fill=white] (0,2) circle (.35);
\draw[fill=white] (2,0) circle (.35);
\draw[fill=white] (2,2) circle (.35);
\node at (0,0) {$c_{\beta}$};
\node at (0,2) {$c_{\beta}$};
\node at (2,2) {$c_{\beta}$};
\node at (2,0) {$c_{\beta}$};
\path[fill=white] (1,1) circle (.3);
\node at (1,1) {$R_{z_i,z_j}$};
\end{tikzpicture}&
\begin{tikzpicture}[scale=0.7]
\draw (0,0) to [out = 0, in = 180] (2,2);
\draw (0,2) to [out = 0, in = 180] (2,0);
\draw[fill=white] (0,0) circle (.35);
\draw[fill=white] (0,2) circle (.35);
\draw[fill=white] (2,0) circle (.35);
\draw[fill=white] (2,2) circle (.35);
\node at (0,0) {$c_{\alpha}$};
\node at (0,2) {$c_{\beta}$};
\node at (2,2) {$c_{\alpha}$};
\node at (2,0) {$c_{\beta}$};
\path[fill=white] (1,1) circle (.3);
\node at (1,1) {$R_{z_i,z_j}$};
\end{tikzpicture}&
\begin{tikzpicture}[scale=0.7]
\draw (0,0) to [out = 0, in = 180] (2,2);
\draw (0,2) to [out = 0, in = 180] (2,0);
\draw[fill=white] (0,0) circle (.35);
\draw[fill=white] (0,2) circle (.35);
\draw[fill=white] (2,0) circle (.35);
\draw[fill=white] (2,2) circle (.35);
\node at (0,0) {$c_{\beta}$};
\node at (0,2) {$c_{\alpha}$};
\node at (2,2) {$c_{\beta}$};
\node at (2,0) {$c_{\alpha}$};
\path[fill=white] (1,1) circle (.3);
\node at (1,1) {$R_{z_i,z_j}$};
\end{tikzpicture}&
\begin{tikzpicture}[scale=0.7]
\draw (0,0) to [out = 0, in = 180] (2,2);
\draw (0,2) to [out = 0, in = 180] (2,0);
\draw[fill=white] (0,0) circle (.35);
\draw[fill=white] (0,2) circle (.35);
\draw[fill=white] (2,0) circle (.35);
\draw[fill=white] (2,2) circle (.35);
\node at (0,0) {$c_{\beta}$};
\node at (0,2) {$c_{\beta}$};
\node at (2,2) {$c_{\alpha}$};
\node at (2,0) {$c_{\alpha}$};
\path[fill=white] (1,1) circle (.3);
\node at (1,1) {$R_{z_i,z_j}$};
\end{tikzpicture}&
\begin{tikzpicture}[scale=0.7]
\draw (0,0) to [out = 0, in = 180] (2,2);
\draw (0,2) to [out = 0, in = 180] (2,0);
\draw[fill=white] (0,0) circle (.35);
\draw[fill=white] (0,2) circle (.35);
\draw[fill=white] (2,0) circle (.35);
\draw[fill=white] (2,2) circle (.35);
\node at (0,0) {$c_{\alpha}$};
\node at (0,2) {$c_{\alpha}$};
\node at (2,2) {$c_{\beta}$};
\node at (2,0) {$c_{\beta}$};
\path[fill=white] (1,1) circle (.3);
\node at (1,1) {$R_{z_i,z_j}$};
\end{tikzpicture}\\
\hline
1&1
&\frac{z_i'+q z_j-(q+1)z_i'z_j}{1-z_i'z_j}
&\frac{q^{-1}z_i'+z_j-(1+q^{-1})z_i'z_j}{1-z_i'z_j}
&\frac{(1-z_i')(1-q z_j)}{1-z_i'z_j}
&\frac{(1-q^{-1}z_i')(1-z_j)}{1-z_i'z_j}\\
\hline
\end{array}\]
\caption{Boltzmann weights for colored stochastic $\Delta-\Gamma$ vertex with spectral parameters $z_i$ and $z_j$: where $\alpha<\beta$}
\label{RmatrixC2}
\end{figure}

\begin{figure}[h]
\[\begin{array}{|c|c|c|c|c|c|}
\hline
\begin{tikzpicture}[scale=0.7]
\draw (0,0) to [out = 0, in = 180] (2,2);
\draw (0,2) to [out = 0, in = 180] (2,0);
\draw[fill=white] (0,0) circle (.35);
\draw[fill=white] (0,2) circle (.35);
\draw[fill=white] (2,0) circle (.35);
\draw[fill=white] (2,2) circle (.35);
\node at (0,0) {$c_{\alpha}$};
\node at (0,2) {$c_{\alpha}$};
\node at (2,2) {$c_{\alpha}$};
\node at (2,0) {$c_{\alpha}$};
\path[fill=white] (1,1) circle (.3);
\node at (1,1) {$R_{z_i,z_j}$};
\end{tikzpicture}&
\begin{tikzpicture}[scale=0.7]
\draw (0,0) to [out = 0, in = 180] (2,2);
\draw (0,2) to [out = 0, in = 180] (2,0);
\draw[fill=white] (0,0) circle (.35);
\draw[fill=white] (0,2) circle (.35);
\draw[fill=white] (2,0) circle (.35);
\draw[fill=white] (2,2) circle (.35);
\node at (0,0) {$c_{\beta}$};
\node at (0,2) {$c_{\beta}$};
\node at (2,2) {$c_{\beta}$};
\node at (2,0) {$c_{\beta}$};
\path[fill=white] (1,1) circle (.3);
\node at (1,1) {$R_{z_i,z_j}$};
\end{tikzpicture}&
\begin{tikzpicture}[scale=0.7]
\draw (0,0) to [out = 0, in = 180] (2,2);
\draw (0,2) to [out = 0, in = 180] (2,0);
\draw[fill=white] (0,0) circle (.35);
\draw[fill=white] (0,2) circle (.35);
\draw[fill=white] (2,0) circle (.35);
\draw[fill=white] (2,2) circle (.35);
\node at (0,0) {$c_{\alpha}$};
\node at (0,2) {$c_{\beta}$};
\node at (2,2) {$c_{\alpha}$};
\node at (2,0) {$c_{\beta}$};
\path[fill=white] (1,1) circle (.3);
\node at (1,1) {$R_{z_i,z_j}$};
\end{tikzpicture}&
\begin{tikzpicture}[scale=0.7]
\draw (0,0) to [out = 0, in = 180] (2,2);
\draw (0,2) to [out = 0, in = 180] (2,0);
\draw[fill=white] (0,0) circle (.35);
\draw[fill=white] (0,2) circle (.35);
\draw[fill=white] (2,0) circle (.35);
\draw[fill=white] (2,2) circle (.35);
\node at (0,0) {$c_{\beta}$};
\node at (0,2) {$c_{\alpha}$};
\node at (2,2) {$c_{\beta}$};
\node at (2,0) {$c_{\alpha}$};
\path[fill=white] (1,1) circle (.3);
\node at (1,1) {$R_{z_i,z_j}$};
\end{tikzpicture}&
\begin{tikzpicture}[scale=0.7]
\draw (0,0) to [out = 0, in = 180] (2,2);
\draw (0,2) to [out = 0, in = 180] (2,0);
\draw[fill=white] (0,0) circle (.35);
\draw[fill=white] (0,2) circle (.35);
\draw[fill=white] (2,0) circle (.35);
\draw[fill=white] (2,2) circle (.35);
\node at (0,0) {$c_{\beta}$};
\node at (0,2) {$c_{\alpha}$};
\node at (2,2) {$c_{\alpha}$};
\node at (2,0) {$c_{\beta}$};
\path[fill=white] (1,1) circle (.3);
\node at (1,1) {$R_{z_i,z_j}$};
\end{tikzpicture}&
\begin{tikzpicture}[scale=0.7]
\draw (0,0) to [out = 0, in = 180] (2,2);
\draw (0,2) to [out = 0, in = 180] (2,0);
\draw[fill=white] (0,0) circle (.35);
\draw[fill=white] (0,2) circle (.35);
\draw[fill=white] (2,0) circle (.35);
\draw[fill=white] (2,2) circle (.35);
\node at (0,0) {$c_{\alpha}$};
\node at (0,2) {$c_{\beta}$};
\node at (2,2) {$c_{\beta}$};
\node at (2,0) {$c_{\alpha}$};
\path[fill=white] (1,1) circle (.3);
\node at (1,1) {$R_{z_i,z_j}$};
\end{tikzpicture}\\
\hline
1&1
&\frac{z_j'-z_i'}{q-(q+1)z_i'+z_i'z_j'}
&\frac{q(z_j'-z_i')}{q-(q+1)z_i'+z_i'z_j'}
&\frac{(1-z_i')(q-z_j')}{q-(q+1)z_i'+z_i'z_j'}
&\frac{(1-z_j')(q-z_i')}{q-(q+1)z_i'+z_i'z_j'}\\
\hline
\end{array}\]
\caption{Boltzmann weights for colored stochastic $\Delta-\Delta$ vertex with spectral parameters $z_i$ and $z_j$: where $\alpha<\beta$}
\label{RmatrixC3}
\end{figure}

The following theorem gives the three sets of Yang-Baxter equations for the colored stochastic symplectic ice.

\begin{theorem}\label{YBE}
For any $(X,Y)\in\{(\Delta,\Gamma),(\Gamma,\Gamma),(\Delta,\Delta)\}$ the following holds. Assume that $S$ is colored stochastic $X$ vertex with spectral parameter $z_i$, $T$ is colored stochastic $Y$ vertex with spectral parameter $z_j$, and $R$ is colored stochastic $X-Y$ vertex with spectral parameters $z_i,z_j$. Then the partition functions of the following two configurations are equal for any fixed combination of colors $a,b,c,d,e,f\in \{c_{\bar{n}},\cdots,c_0,\cdots,c_n\}$.
\begin{equation}
\hfill
\begin{tikzpicture}[baseline=(current bounding box.center)]
  \draw (0,1) to [out = 0, in = 180] (2,3) to (4,3);
  \draw (0,3) to [out = 0, in = 180] (2,1) to (4,1);
  \draw (3,0) to (3,4);
  \draw[fill=white] (0,1) circle (.3);
  \draw[fill=white] (0,3) circle (.3);
  \draw[fill=white] (3,4) circle (.3);
  \draw[fill=white] (4,3) circle (.3);
  \draw[fill=white] (4,1) circle (.3);
  \draw[fill=white] (3,0) circle (.3);
  \draw[fill=white] (2,3) circle (.3);
  \draw[fill=white] (2,1) circle (.3);
  \draw[fill=white] (3,2) circle (.3);
  \node at (0,1) {$a$};
  \node at (0,3) {$b$};
  \node at (3,4) {$c$};
  \node at (4,3) {$d$};
  \node at (4,1) {$e$};
  \node at (3,0) {$f$};
  \node at (2,3) {$g$};
  \node at (3,2) {$h$};
  \node at (2,1) {$i$};
\filldraw[black] (3,3) circle (2pt);
\node at (3,3) [anchor=south west] {$S$};
\filldraw[black] (3,1) circle (2pt);
\node at (3,1) [anchor=north west] {$T$};
\filldraw[black] (1,2) circle (2pt);
\node at (1,2) [anchor=west] {$R$};
\end{tikzpicture}\qquad\qquad
\begin{tikzpicture}[baseline=(current bounding box.center)]
  \draw (0,1) to (2,1) to [out = 0, in = 180] (4,3);
  \draw (0,3) to (2,3) to [out = 0, in = 180] (4,1);
  \draw (1,0) to (1,4);
  \draw[fill=white] (0,1) circle (.3);
  \draw[fill=white] (0,3) circle (.3);
  \draw[fill=white] (1,4) circle (.3);
  \draw[fill=white] (4,3) circle (.3);
  \draw[fill=white] (4,1) circle (.3);
  \draw[fill=white] (1,0) circle (.3);
  \draw[fill=white] (2,3) circle (.3);
  \draw[fill=white] (1,2) circle (.3);
  \draw[fill=white] (2,1) circle (.3);
  \node at (0,1) {$a$};
  \node at (0,3) {$b$};
  \node at (1,4) {$c$};
  \node at (4,3) {$d$};
  \node at (4,1) {$e$};
  \node at (1,0) {$f$};
  \node at (2,3) {$j$};
  \node at (1,2) {$k$};
  \node at (2,1) {$l$};
\filldraw[black] (1,3) circle (2pt);
\node at (1,3) [anchor=south west] {$T$};
\filldraw[black] (1,1) circle (2pt);
\node at (1,1) [anchor=north west]{$S$};
\filldraw[black] (3,2) circle (2pt);
\node at (3,2) [anchor=west] {$R$};
\end{tikzpicture}
\end{equation}
\end{theorem}
\begin{proof}
From conservation of colors for both colored stochastic $\Gamma$ vertices and colored stochastic $\Delta$ vertices (note that the directions of input and output are different for these two types of vertices), it can be checked that at most four distinct colors (including $c_0$) can appear on the boundary edges in any of the two configurations, and that the color on an inner edge must be one of the colors on boundary edges (only considering admissible configurations). Moreover, the Boltzmann weight of a vertex only depends on the relative order of the colors on its adjacent four edges. Therefore it suffices to check the result for four colors, and there are at most $4^6$ possible combinations of boundary colors. These identities are checked using a SAGE program.
\end{proof}

\subsection{Evaluation of the partition function when $\sigma(i)=\overline{\tau(i)}$}\label{Sect.3.2.5}

When $\sigma(i)=\overline{\tau(i)}$ for every $1\leq i\leq n$, the partition function $Z(\mathcal{S}_{n,L,\lambda,\sigma,\tau,z})$ has a relatively simple form, as is shown in the following theorem.

\begin{theorem}\label{ThmB}
If $\sigma$ and $\tau$ satisfy the condition that $\sigma(i)=\overline{\tau(i)}$ for every $1\leq i\leq n$, then we have
\begin{eqnarray}
    Z(\mathcal{S}_{n,L,\lambda,\sigma,\tau,z})&=&\prod_{i=1}^nz_i^L \prod_{i=1}^n (\frac{z_i'}{q})^{\lambda_i+n-i}\prod_{i=1}^n(1-q^{-1_{\sigma(i)<0}}z_i') \nonumber \\
    &\times& q^{\sum_{i=1}^n(L-n+i+\lambda_i)1_{\sigma(i)>0}+\sum_{1\leq i<j\leq n}(1_{\overline{\sigma(j)}<\sigma(i)}+1_{\sigma(j)<\sigma(i)})}.
\end{eqnarray}
In particular, if $\sigma(i)=\overline{i}$ and $\tau(i)=i$ for every $1\leq i\leq n$, then
\begin{equation}
    Z(\mathcal{S}_{n,L,\lambda,\sigma,\tau,z})=\prod_{i=1}^nz_i^L \prod_{i=1}^n (\frac{z_i'}{q})^{\lambda_i+n-i}\prod_{i=1}^n(1-\frac{z_i'}{q}) q^{\frac{n(n-1)}{2}}.
\end{equation}
\end{theorem}
\begin{proof}
We use the colored path interpretation. An illustration of the proof is shown in Figure \ref{C}, where we assume that $c_1=B,c_2=R,c_{\overline{1}}=\overline{B},c_{\overline{2}}=\overline{R}$. We say that $c_{i}$ and $c_{\overline{i}}$ are of the same color type for $1\leq i\leq n$ (and $c_0$ itself forms a color type). The collection of the particles with the same color type are viewed as a colored path. Each path enters from the left boundary, moves rightward or downward on each row of colored stochastic $\Gamma$ vertex, and moves leftward or downward on each row of colored stochastic $\Delta$ vertex. When the path enters the cap on a row of colored stochastic $\Gamma$ vertex, it will bend through the cap, change the color to its opposite, and restart on the right-most vertex of the row of colored stochastic $\Delta$ vertex just below the previous row. Finally, the colored path leaves the rectangular lattice at the bottom boundary.

Consider the colored path entering from the $2$nd row, which has color $c_{\sigma(1)}$. In order for the path to leave the domain with an opposite color (which is required by the boundary condition, as $\tau(1)=\overline{\sigma(1)}$), it has to move rightward until it goes through the cap. Then the path changes its color to $c_{\tau(1)}$ and leaves the domain at the column labeled as $\lambda_1+n$.

Now consider the colored path entering from the $4$th row, which has color $c_{\sigma(2)}$. In order for the path to leave the domain with an opposite color, it has to move rightward until it goes through the cap (as the cap connecting the first two rows has already been taken by the colored path entering from the $2$nd row). Then note that $\lambda_1+n>\lambda_2+n-1$. In order for the path to leave at the column labeled $\lambda_2+n-1$, it has to move leftward after passing the cap until it reaches the column labeled $\lambda_2+n-1$. After that it moves downward until it leaves the domain.

The rest of the colored paths can be analyzed similarly. The colored path entering from the $(2i)$th row first moves rightward until it goes through the cap (and changes the color to its opposite), then it moves leftward until it reaches the column labeled $\lambda_i+n+1-i$, and finally it moves downward until it leaves the domain.

The above analysis shows that there is only one admissible state. Computing the Boltzmann weight of this state finishes the proof.
\end{proof}

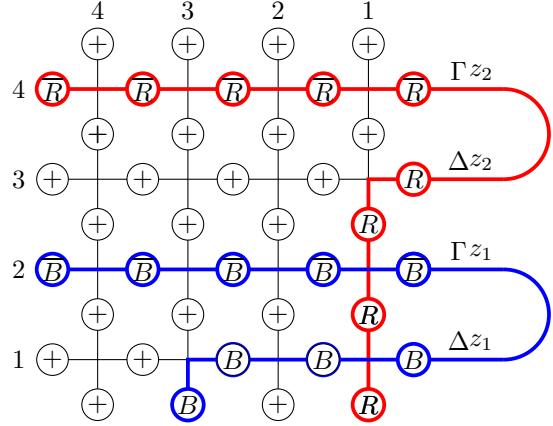
\begin{figure}[h]
\centering
  \begin{tikzpicture}[scale=0.6]
    \draw (1,0)--(1,8);
    \draw (3,0)--(3,8);
    \draw (5,0)--(5,8);
    \draw (7,0)--(7,8);
    \draw (0,1)--(10,1);
    \draw (0,3)--(10,3);
    \draw (0,5)--(10,5);
    \draw (0,7)--(10,7);
    \draw (10,1) arc(-90:90:1);
    \draw (10,5) arc(-90:90:1);
    \draw[fill=white] (3,6) circle (.35);
    \draw[fill=white] (3,4) circle (.35);
    \draw[fill=white] (1,2) circle (.35);
    \draw[fill=white] (3,2) circle (.35);
    \draw[fill=white] (5,2) circle (.35);
    \draw[fill=white] (1,0) circle (.35);
    \draw[fill=white] (3,0) circle (.35);
    \draw[fill=white] (5,0) circle (.35);
    \draw[fill=white] (7,0) circle (.35);
    \draw[fill=white] (0,5) circle (.35);
    \draw[fill=white] (2,5) circle (.35);
    \draw[fill=white] (6,5) circle (.35);
    \draw[fill=white] (4,5) circle (.35);
    \draw[fill=white] (8,5) circle (.35);
    \draw[fill=white] (0,3) circle (.35);
    \draw[fill=white] (2,3) circle (.35);
    \draw[fill=white] (4,3) circle (.35);
    \draw[fill=white] (6,3) circle (.35);
    \draw[fill=white] (4,3) circle (.35);
    \draw[fill=white] (3,0) circle (.35);

    \node at (9,7) [anchor=south] {$\Gamma$};
    \node at (9.5,7) [anchor=south]
    {$z_2$};
    \node at (9,5) [anchor=south] {$\Delta$};
    \node at (9.5,5) [anchor=south]
    {$z_2$};
    \node at (9,3) [anchor=south] {$\Gamma$};
    \node at (9.5,3) [anchor=south]
    {$z_1$};
    \node at (9,1) [anchor=south] {$\Delta$};
    \node at (9.5,1) [anchor=south]
    {$z_1$};
    
    \draw[line width=0.5mm,red,fill=white] (0,7) circle (.35);
    \node at (0,7) {$\overline{R}$};
    \draw [line width=0.5mm,red] (0.35,7.00)--(1.65,7.00);
    \draw[line width=0.5mm,red,fill=white] (2,7) circle (.35);
    \node at (2,7) {$\overline{R}$};
    \draw [line width=0.5mm,red] (2.35,7.00)--(3.65,7.00);
    \draw[line width=0.5mm,red,fill=white] (4,7) circle (.35);
    \node at (4,7) {$\overline{R}$};
    \draw [line width=0.5mm,red] (4.35,7.00)--(5.65,7.00);
    \draw[line width=0.5mm,red,fill=white] (6,7) circle (.35);
    \node at (6,7) {$\overline{R}$};
    \draw [line width=0.5mm,red] (6.35,7.00)--(7.65,7.00);
    \draw[line width=0.5mm,red,fill=white] (8,7) circle (.35);
    \node at (8,7) {$\overline{R}$};
    \draw [line width=0.5mm,red] (8.35,7.00)--(10,7.00);
    \draw [line width=0.5mm,red] (10,5) arc(-90:90:1);
    \draw [line width=0.5mm,red] (8.35,5.00)--(10,5.00);
    \draw[line width=0.5mm,red,fill=white] (8,5) circle (.35);
    \draw[line width=0.5mm,red,fill=white] (7,4) circle (.35);
    \draw[line width=0.5mm,blue,fill=white] (8,3) circle (.35);
    \node at (8,5) {$R$};
    \draw [line width=0.5mm,red] (7,4.35) to (7,5) to (7.65,5.00);
    \draw [line width=0.5mm,red] (7,3.65) to (7,2.35);
    \draw[line width=0.5mm,red,fill=white] (7,2) circle (.35);
    \node at (7,2) {$R$};
    \draw[line width=0.5mm,red] (7,1.65) to (7,0.35);
    \draw[line width=0.5mm,red,fill=white] (7,0) circle (.35);
    \node at (7,0) {$R$};
    
    \draw [line width=0.5mm,blue] (0.35,3) to (1.65,3);
    \draw [line width=0.5mm,blue] (4.35,3) to (5.65,3);
    \draw [line width=0.5mm,blue] (2.35,3) to (3.65,3);
    \draw [line width=0.5mm,blue] (6.35,3) to (7.65,3);
    \draw [line width=0.5mm,blue] (8.35,1.00)--(10,1.00);
    \draw [line width=0.5mm,blue] (8.35,3.00)--(10,3.00);
    \draw [line width=0.5mm,blue] (10,1) arc(-90:90:1);
    \draw [line width=0.5mm,blue] (6.35,1.00)--(7.65,1.00);
    \draw[line width=0.5mm,blue,fill=white] (2,3) circle (.35);
    \draw[line width=0.5mm,blue,fill=white] (0,3) circle (.35);
    \draw[line width=0.5mm,blue,fill=white] (4,3) circle (.35);
    \draw[line width=0.5mm,blue,fill=white] (6,3) circle (.35);
    \draw[line width=0.5mm,blue,fill=white] (8,3) circle (.35);
    \draw[line width=0.5mm,blue,fill=white] (6,1) circle (.35);
    \draw[line width=0.5mm,blue,fill=white] (4,1) circle (.35);
    \draw[line width=0.5mm,blue,fill=white] (3,0) circle (.35);
    \draw [line width=0.5mm,blue] (4.35,1) to (5.65,1);
    \draw[line width=0.5mm, blue] (3,0.35) to (3,1) to (3.65,1.00);
    
    \draw[fill=white] (1,8) circle (.35);
    \draw[fill=white] (3,8) circle (.35);
    \draw[fill=white] (5,8) circle (.35);
    \draw[fill=white] (7,8) circle (.35);
    \draw[fill=white] (1,4) circle (.35);
    \draw[fill=white] (1,6) circle (.35);
    \draw[fill=white] (7,6) circle (.35);
    \draw[fill=white] (5,4) circle (.35);
    \draw[fill=white] (5,6) circle (.35);
    \draw[fill=white] (0,1) circle (.35);
    \draw[fill=white] (2,1) circle (.35);
    \draw[fill=white] (4,1) circle (.35);
    \draw[fill=white] (6,1) circle (.35);
    \draw[line width=0.5mm,blue,fill=white] (8,1) circle (.35);
    \node at (1,8) {$+$};
    \node at (3,8) {$+$};
    \node at (5,8) {$+$};
    \node at (7,8) {$+$};

    \node at (1,6) {$+$};
    \node at (3,6) {$+$};
    \node at (5,6) {$+$};
    \node at (7,6) {$+$};
    \node at (1,4) {$+$};
    \node at (3,4) {$+$};
    \node at (5,4) {$+$};
    \node at (7,4) {$R$};
    \node at (1,2) {$+$};
    \node at (3,2) {$+$};
    \node at (5,2) {$+$};
    \node at (7,2) {$R$};
    \node at (1,0) {$+$};
    \node at (3,0) {$B$};
    \node at (5,0) {$+$};
    \node at (7,0) {$R$};
    \node at (0,5) {$+$};
    \node at (2,5) {$+$};
    \node at (4,5) {$+$};
    \node at (6,5) {$+$};
    \node at (1,6) {$+$};
    \node at (0,3) {$\overline{B}$};
    \node at (2,3) {$\overline{B}$};
    \node at (4,3) {$\overline{B}$};
    \node at (6,3) {$\overline{B}$};
    \node at (8,3) {$\overline{B}$};
    \node at (0,1) {$+$};
    \node at (2,1) {$+$};
    \node at (4,1) {$B$};
    \node at (6,1) {$B$};
    \node at (8,1) {$B$};
    \node at (1.00,8.75) {$ 4$};
    \node at (3.00,8.75) {$ 3$};
    \node at (5.00,8.75) {$ 2$};
    \node at (7.00,8.75) {$ 1$};
    \node at (-.75,7) {$ 4$};
    \node at (-.75,5) {$ 3$};
    \node at (-.75,3) {$ 2$};
    \node at (-.75,1) {$ 1$};
    \node at (0,-.50) {$\quad$};
  \end{tikzpicture}
  \caption{Illustration of the proof of Theorem \ref{ThmB}}
 \label{C}
\end{figure}

\subsection{The reflection equation}\label{Sect.3.2.7}

Due to the lack of the R-matrix and the Yang-Baxter equation for colored stochastic $\Gamma$ vertex and colored stochastic $\Delta$ vertex, we cannot use the caduceus relation as in the uncolored model. However, another set of relations, the reflection equation, provides an alternative way to derive recursive relations of the partition function. The following theorem gives the reflection equation.

\begin{theorem}\label{Refl}
Assume that $S$ is colored stochastic $\Gamma-\Gamma$ vertex of spectral parameters $z_i,z_j$, $T$ is colored stochastic $\Delta-\Gamma$ vertex of spectral parameters $z_i,z_j$, $S'$ is colored stochastic $\Delta-\Delta$ vertex of spectral parameters $z_j,z_i$, and $T'$ is colored stochastic $\Delta-\Gamma$ vertex of spectral parameters $z_j,z_i$. Denote by $Z(I_5(\epsilon_1,\epsilon_2,\epsilon_3,\epsilon_4))$ the partition function of the following configuration with fixed combination of colors $\epsilon_1,\epsilon_2,\epsilon_3,\epsilon_4  \in\{c_{\bar{n}},\cdots,c_0,\cdots,c_n\}$.
\begin{equation}
\label{eqn:reflect1}
\hfill
I_5=
\begin{tikzpicture}[baseline=(current bounding box.center)]
  \draw (0,0) to (4,0);
  \draw (0,1) to (1,1) to [out=0, in=180] (4,2);
  \draw (0,2) to (1,2) to [out=0, in=180] (4,3);
  \draw (0,3) to (1,3) to [out=0, in=180] (4,1);
  \node at (0,0) [anchor=east] {$\epsilon_4$};
  \node at (0,1) [anchor=east] {$\epsilon_3$};
  \node at (0,2) [anchor=east] {$\epsilon_2$};
  \node at (0,3) [anchor=east] {$\epsilon_1$};
  \draw (4,2) arc(-90:90:0.5);
  \draw (4,0) arc(-90:90:0.5);
  \filldraw[black] (4.5,0.5) circle (2pt);
  \filldraw[black] (4.5,2.5) circle (2pt);
  \filldraw[black] (2.2,2.35) circle (2pt);
  \filldraw[black] (2.75,1.65) circle (2pt);
  \node at (2.25,2.5) [anchor=south] {$S$};
  \node at (2.8, 1.7) [anchor=south] {$T$};
\end{tikzpicture}
\end{equation}
Also denote by $Z(I_6(\epsilon_1,\epsilon_2,\epsilon_3,\epsilon_4))$ the partition function of the following configuration with fixed combination of colors $\epsilon_1,\epsilon_2,\epsilon_3,\epsilon_4 \in\{c_{\bar{n}},\cdots,c_0,\cdots,c_n\}$.
\begin{equation}
\label{eqn:reflect2}
\hfill
I_6=
\begin{tikzpicture}[baseline=(current bounding box.center)]
  \draw (0,3) to (4,3);
  \draw (0,1) to (1,1) to [out=0, in=180] (4,0);
  \draw (0,2) to (1,2) to [out=0, in=180] (4,1);
  \draw (0,0) to (1,0) to [out=0, in=180] (4,2);
  \node at (0,0) [anchor=east] {$\epsilon_4$};
  \node at (0,1) [anchor=east] {$\epsilon_3$};
  \node at (0,2) [anchor=east] {$\epsilon_2$};
  \node at (0,3) [anchor=east] {$\epsilon_1$};
  \draw (4,2) arc(-90:90:0.5);
  \draw (4,0) arc(-90:90:0.5);
  \filldraw[black] (4.5,0.5) circle (2pt);
  \filldraw[black] (4.5,2.5) circle (2pt);
  \filldraw[black] (2.2,0.67) circle (2pt);
  \filldraw[black] (2.8,1.35) circle (2pt);
  \node at (2.25,0.8) [anchor=south] {$S'$};
  \node at (2.8, 1.5) [anchor=south] {$T'$};
\end{tikzpicture}
\end{equation}
Then for any fixed combination of colors $\epsilon_1,\epsilon_2,\epsilon_3,\epsilon_4\in \{c_{\bar{n}}, \cdots, c_0,\cdots, c_n\}$, we have
\begin{equation}
    Z(I_5(\epsilon_1,\epsilon_2,\epsilon_3,\epsilon_4))=Z(I_6(\epsilon_1,\epsilon_2,\epsilon_3,\epsilon_4)).
\end{equation}
\end{theorem}
\begin{proof}
We say that $c_{\alpha}$ and $c_{\beta}$ are of the same color type, if $\beta=\bar{\alpha}$ ($c_0$ itself forms a color type). From conservation of colors for the R-matrix and the cap weights, we can deduce that for any admissible state of $I_5$ or $I_6$, each color type must appear for an even number of times in $\{\epsilon_1,\epsilon_2,\epsilon_3,\epsilon_4\}$, and that the color type of an inner edge must be one of the color types of $\{\epsilon_1,\epsilon_2,\epsilon_3,\epsilon_4\}$. From this we can further deduce that at most two color types can appear in $\{\epsilon_1,\epsilon_2,\epsilon_3,\epsilon_4\}$ in any admissible state of $I_5$ or $I_6$. Moreover, we note that the Boltzmann weight for the R-matrix only depends on the relative order of colors on its four adjacent edges. Therefore it suffices to check the relation for five colors $\{c_{\bar{2}},c_{\bar{1}},c_0,c_1,c_2\}$. There are at most $5^4$ combinations of $(\epsilon_1,\epsilon_2,\epsilon_3,\epsilon_4)$ for this case. These identities have been checked using a SAGE program.
\end{proof}

\subsection{Recursive relations of the partition function}\label{Sect.3.3}

In this section, we derive recursive relations of the partition function. The recursive relations are further related to Demazure-Lusztig operators of type C in Section \ref{Sect.3.4}. The main results are Theorems \ref{Thm2C}-\ref{Thm3C}.

\begin{theorem}\label{Thm2C}
Assume that $1\leq i\leq n-1$ and $\sigma(i+1)>\sigma(i)$. Let $s_i$ be the transposition $(i,i+1)$ in $B_n$, and let $s_i z$ be the vector obtained from $z$ by interchanging $z_i,z_{i+1}$. Then the partition function of the colored stochastic symplectic ice satisfies the following recursive relation:
\begin{equation}
    q^{1_{\sigma(i+1)>0}-1_{\sigma(i)>0}}Z(\mathcal{S}_{n,L,\lambda,\sigma s_i,\tau,z})
   = -A(q,z,i)  Z(\mathcal{S}_{n,L,\lambda,\sigma,\tau,z})+B(q,z,i)Z(\mathcal{S}_{n,L,\lambda,\sigma,\tau,s_i z}),
\end{equation}
where 
\begin{equation}
A(q,z,i) =  \frac{(1-z_{i+1})(1-qz_i)}{z_{i+1}-z_i},
\end{equation}
and
\begin{equation}
 B(q,z,i) = \frac{1-(q+1)z_i+q z_i z_{i+1}}{z_{i+1}-z_i}.
\end{equation}
\end{theorem}

\begin{theorem}\label{Thm3C}
Assume that $\sigma(n)>0$. Let $s_n$ be the element of $B_n$ that changes the sign of the element at the $n$th position, and 
\begin{equation}
    s_n z=(z_1,\cdots,z_{n-1},\frac{1}{z_n'}).
\end{equation}
Then we have
\begin{equation}
    (\frac{q}{z_n})^LZ(\mathcal{S}_{n,L,\lambda,\sigma    s_n,\tau,z})=C(q,z)z_n^{-L}Z(\mathcal{S}_{n,L,\lambda,\sigma,\tau,z})-D(q,z)z_n'^L Z(\mathcal{S}_{n,L,\lambda,\sigma,\tau,s_nz}).
\end{equation}
where
\begin{equation}
    C(q,z)=\frac{(q-z_n')(z_n-1)}{q(1-z_nz_n')},
\end{equation}
\begin{equation}
    D(q,z)=\frac{q z_n+z_n'-(q+1)z_nz_n'}{q(1-z_nz_n')}.
\end{equation}
\end{theorem}

The rest of this section is devoted to the proof of Theorems \ref{Thm2C}-\ref{Thm3C}.

\subsubsection{Proof of Theorem \ref{Thm2C}}

The proof of Theorem \ref{Thm2C} is based on the Yang-Baxter equation (Theorem \ref{YBE}) and the reflection equation (Theorem \ref{Refl}). The idea is to attach two R-vertices to the left of the configuration, move them to the right using the Yang-Baxter equation, make a reflection using the reflection equation, and finally use the Yang-Baxter equation to move the braid back to the left boundary. This gives the desired recursive relation.

\begin{proof}[Proof of Theorem \ref{Thm2C}]
We attach two R-vertices to the left of $\mathcal{S}_{n,L,\lambda,\sigma,\tau,s_i z}$, as shown in the following 
\begin{equation}
    \begin{tikzpicture}[baseline=(current bounding box.center)]
  \draw (0,0) to (0.6,0) to [out=0, in=180] (3,2) to (4,2);
  \draw (0,3) to (0.6,3) to [out=0, in=180] (3,1) to (4,1);
  \draw (2.5,3) to (4,3);
  \draw (0,2) to [out=0, in=180] (2.4,0) to (3,0) to (4,0);
  \draw [dashed] (4,2) to (5,2);
  \draw [dashed] (4,1) to (5,1);
  \draw [dashed] (4,3) to (5,3);
  \draw [dashed] (4,0) to (5,0);
  \draw (5,2) to (5.5,2);
  \draw (5,1) to (5.5,1);
  \draw (5,3) to (5.5,3);
  \draw (5,0) to (5.5,0);
  \draw (5.5,2) arc(-90:90:0.5);
  \draw (5.5,0) arc(-90:90:0.5);
  \draw (3,-0.5) to (3,3.5);
  \draw (4,-0.5) to (4,3.5);
  \draw (5,-0.5) to (5,3.5);
  \filldraw[black] (6,0.5) circle (2pt);
  \filldraw[black] (6,2.5) circle (2pt);
  \filldraw[black] (2.1,1.5) circle (2pt);
  \filldraw[black] (1.5,0.5) circle (2pt);
  \filldraw[black] (3,0) circle (2pt);
  \filldraw[black] (4,0) circle (2pt);
  \filldraw[black] (5,0) circle (2pt);
  \filldraw[black] (3,1) circle (2pt);
  \filldraw[black] (4,1) circle (2pt);
  \filldraw[black] (5,1) circle (2pt);
  \filldraw[black] (3,2) circle (2pt);
  \filldraw[black] (4,2) circle (2pt);
  \filldraw[black] (5,2) circle (2pt);
  \filldraw[black] (3,3) circle (2pt);
  \filldraw[black] (4,3) circle (2pt);
  \filldraw[black] (5,3) circle (2pt);
  \node at (2.1,1.5) [anchor=south] {$T'$};
  \node at (1.5,0.5) [anchor=south] {$S'$};
  \node at (0,0) [anchor=east] {$+$};
  \node at (2.5,3) [anchor=east] {$c_{\sigma(i+1)}$};
  \node at (0,2) [anchor=east] {$+$};
  \node at (0,3) [anchor=east] {$c_{\sigma(i)}$};
  \node at (5.5,0) [anchor=south] {$z_{i+1}$};
  \node at (5.5,1) [anchor=south] {$z_{i+1}$};
  \node at (5.5,2) [anchor=south] {$z_{i}$};
  \node at (5.5,3) [anchor=south] {$z_{i}$};
  \node at (2.7,0) [anchor=south] {$\Delta$};
  \node at (2.7,1) [anchor=south] {$\Gamma$};
  \node at (2.7,2) [anchor=south] {$\Delta$};
  \node at (2.7,3) [anchor=south] {$\Gamma$};
\end{tikzpicture}
\end{equation}
where we omit the other rows of $\mathcal{S}_{n,L,\lambda,\sigma,\tau,s_i z}$, $S'$ is colored stochastic $\Delta-\Delta$ vertex of spectral parameters $z_{i},z_{i+1}$, and $T'$ is colored stochastic $\Delta-\Gamma$ R-vertex of spectral parameters $z_{i},z_{i+1}$. We denote by $Z_1$ the partition function of this new ice model.

Note that the only admissible configuration of the two R-vertices is given as follows
\begin{equation}
    \begin{tikzpicture}[baseline=(current bounding box.center)]
  \draw (0,0) to (0.6,0) to [out=0, in=180] (3,2) to (4,2);
  \draw (0,3) to (0.6,3) to [out=0, in=180] (3,1) to (4,1);
  \draw (2.5,3) to (4,3);
  \draw (0,2) to [out=0, in=180] (2.4,0) to (3,0) to (4,0);
  \draw [dashed] (4,2) to (5,2);
  \draw [dashed] (4,1) to (5,1);
  \draw [dashed] (4,3) to (5,3);
  \draw [dashed] (4,0) to (5,0);
  \draw (5,2) to (5.5,2);
  \draw (5,1) to (5.5,1);
  \draw (5,3) to (5.5,3);
  \draw (5,0) to (5.5,0);
  \draw (5.5,2) arc(-90:90:0.5);
  \draw (5.5,0) arc(-90:90:0.5);
  \draw (3,-0.5) to (3,3.5);
  \draw (4,-0.5) to (4,3.5);
  \draw (5,-0.5) to (5,3.5);
  \filldraw[black] (6,0.5) circle (2pt);
  \filldraw[black] (6,2.5) circle (2pt);
  \filldraw[black] (2.1,1.5) circle (2pt);
  \filldraw[black] (1.5,0.5) circle (2pt);
  \filldraw[black] (3,0) circle (2pt);
  \filldraw[black] (4,0) circle (2pt);
  \filldraw[black] (5,0) circle (2pt);
  \filldraw[black] (3,1) circle (2pt);
  \filldraw[black] (4,1) circle (2pt);
  \filldraw[black] (5,1) circle (2pt);
  \filldraw[black] (3,2) circle (2pt);
  \filldraw[black] (4,2) circle (2pt);
  \filldraw[black] (5,2) circle (2pt);
  \filldraw[black] (3,3) circle (2pt);
  \filldraw[black] (4,3) circle (2pt);
  \filldraw[black] (5,3) circle (2pt);
  \node at (2.1,1.5) [anchor=south] {$T'$};
  \node at (1.5,0.5) [anchor=south] {$S'$};
  \node at (0,0) [anchor=east] {$+$};
  \node at (2.5,3) [anchor=east] {$c_{\sigma(i+1)}$};
  \node at (0,2) [anchor=east] {$+$};
  \node at (0,3) [anchor=east] {$c_{\sigma(i)}$};
  \node at (2.5,0) [anchor=north] {$+$};
  \node at (2,1.2) [anchor=north] {$+$};
  \node at (2.6,1.1) [anchor=north] {$c_{\sigma(i)}$};
  \node at (2.8,2) [anchor=north] {$+$};
  \node at (5.5,0) [anchor=south] {$z_{i+1}$};
  \node at (5.5,1) [anchor=south] {$z_{i+1}$};
  \node at (5.5,2) [anchor=south] {$z_{i}$};
  \node at (5.5,3) [anchor=south] {$z_{i}$};
  \node at (2.7,0) [anchor=south] {$\Delta$};
  \node at (2.7,1) [anchor=south] {$\Gamma$};
  \node at (2.7,2) [anchor=south] {$\Delta$};
  \node at (2.7,3) [anchor=south] {$\Gamma$};
\end{tikzpicture}
\end{equation}
Therefore we have
\begin{equation}
    Z_1=R_{\Delta\Delta}(+,+,+,+;z_{i},z_{i+1})R_{\Delta\Gamma}(+,c_{\sigma(i)},+,c_{\sigma(i)};z_{i},z_{i+1})Z(\mathcal{S}_{n,L,\lambda,\sigma,\tau,s_i z}).
\end{equation}

By Theorem \ref{YBE}, we can push the two R-vertices to the right without changing the partition function. Namely, we denote by $Z_2$ the partition function of the following configuration
\begin{equation}
    \begin{tikzpicture}[baseline=(current bounding box.center)]
  \draw (-0.5,2) to (1,2);
  \draw (-0.5,1) to (1,1);
  \draw (-0.5,3) to (1,3);
  \draw (-0.5,0) to (1,0);
  \draw [dashed] (1,2) to (2,2);
  \draw [dashed] (1,1) to (2,1);
  \draw [dashed] (1,3) to (2,3);
  \draw [dashed] (1,0) to (2,0);
  \draw (2,2) to (2.5,2);
  \draw (2,1) to (2.5,1);
  \draw (2,3) to (2.5,3);
  \draw (2,0) to (2.5,0);
  \draw (0,-0.5) to (0,3.5);
  \draw (1,-0.5) to (1,3.5);
  \draw (2,-0.5) to (2,3.5);
  \filldraw[black] (0,0) circle (2pt);
  \filldraw[black] (1,0) circle (2pt);
  \filldraw[black] (2,0) circle (2pt);
  \filldraw[black] (0,1) circle (2pt);
  \filldraw[black] (1,1) circle (2pt);
  \filldraw[black] (2,1) circle (2pt);
  \filldraw[black] (0,2) circle (2pt);
  \filldraw[black] (1,2) circle (2pt);
  \filldraw[black] (2,2) circle (2pt);
  \filldraw[black] (0,3) circle (2pt);
  \filldraw[black] (1,3) circle (2pt);
  \filldraw[black] (2,3) circle (2pt);
  \node at (-0.5,0) [anchor=east] {$+$};
  \node at (-0.5,2) [anchor=east] {$c_{\sigma(i)}$};
  \node at (-0.5,1) [anchor=east] {$+$};
  \node at (-0.5,3) [anchor=east] {$c_{\sigma(i+1)}$};
  \node at (2.8,0) [anchor=south] {$z_{i}$};
  \node at (2.8,1) [anchor=south] {$z_{i+1}$};
  \node at (2.8,2) [anchor=south] {$z_{i+1}$};
  \node at (2.8,3) [anchor=south] {$z_{i}$};
  \node at (2.3,0) [anchor=south] {$\Delta$};
  \node at (2.3,1) [anchor=south] {$\Delta$};
  \node at (2.3,2) [anchor=south] {$\Gamma$};
  \node at (2.3,3) [anchor=south] {$\Gamma$};
  
 \draw (2.5,3) to (6.5,3);
  \draw (2.5,1) to (3.5,1) to [out=0, in=180] (6.5,0);
  \draw (2.5,2) to (3.5,2) to [out=0, in=180] (6.5,1);
  \draw (2.5,0) to (3.5,0) to [out=0, in=180] (6.5,2);
  \draw (6.5,2) arc(-90:90:0.5);
  \draw (6.5,0) arc(-90:90:0.5);
  \filldraw[black] (7,0.5) circle (2pt);
  \filldraw[black] (7,2.5) circle (2pt);
  \filldraw[black] (4.7,0.67) circle (2pt);
  \filldraw[black] (5.3,1.35) circle (2pt);
\end{tikzpicture}
\end{equation}
Then we have $Z_1=Z_2$.

By Theorem \ref{Refl}, $Z_2$ is equal to the partition function of the following configuration
\begin{equation}
    \begin{tikzpicture}[baseline=(current bounding box.center)]
  \draw (-0.5,2) to (1,2);
  \draw (-0.5,1) to (1,1);
  \draw (-0.5,3) to (1,3);
  \draw (-0.5,0) to (1,0);
  \draw [dashed] (1,2) to (2,2);
  \draw [dashed] (1,1) to (2,1);
  \draw [dashed] (1,3) to (2,3);
  \draw [dashed] (1,0) to (2,0);
  \draw (2,2) to (2.5,2);
  \draw (2,1) to (2.5,1);
  \draw (2,3) to (2.5,3);
  \draw (2,0) to (2.5,0);
  \draw (0,-0.5) to (0,3.5);
  \draw (1,-0.5) to (1,3.5);
  \draw (2,-0.5) to (2,3.5);
  \filldraw[black] (0,0) circle (2pt);
  \filldraw[black] (1,0) circle (2pt);
  \filldraw[black] (2,0) circle (2pt);
  \filldraw[black] (0,1) circle (2pt);
  \filldraw[black] (1,1) circle (2pt);
  \filldraw[black] (2,1) circle (2pt);
  \filldraw[black] (0,2) circle (2pt);
  \filldraw[black] (1,2) circle (2pt);
  \filldraw[black] (2,2) circle (2pt);
  \filldraw[black] (0,3) circle (2pt);
  \filldraw[black] (1,3) circle (2pt);
  \filldraw[black] (2,3) circle (2pt);
  \node at (-0.5,0) [anchor=east] {$+$};
  \node at (-0.5,2) [anchor=east] {$c_{\sigma(i)}$};
  \node at (-0.5,1) [anchor=east] {$+$};
  \node at (-0.5,3) [anchor=east] {$c_{\sigma(i+1)}$};
  \node at (2.8,0) [anchor=south] {$z_{i}$};
  \node at (2.8,1) [anchor=south] {$z_{i+1}$};
  \node at (2.8,2) [anchor=south] {$z_{i+1}$};
  \node at (2.8,3) [anchor=south] {$z_{i}$};
  \node at (2.3,0) [anchor=south] {$\Delta$};
  \node at (2.3,1) [anchor=south] {$\Delta$};
  \node at (2.3,2) [anchor=south] {$\Gamma$};
  \node at (2.3,3) [anchor=south] {$\Gamma$};
  
  \draw (2.5,0) to (6.5,0);
  \draw (2.5,1) to (3.5,1) to [out=0, in=180] (6.5,2);
  \draw (2.5,2) to (3.5,2) to [out=0, in=180] (6.5,3);
  \draw (2.5,3) to (3.5,3) to [out=0, in=180] (6.5,1);
  \draw (6.5,2) arc(-90:90:0.5);
  \draw (6.5,0) arc(-90:90:0.5);
  \filldraw[black] (7,0.5) circle (2pt);
  \filldraw[black] (7,2.5) circle (2pt);
  \filldraw[black] (4.7,2.35) circle (2pt);
  \filldraw[black] (5.25,1.65) circle (2pt);
\end{tikzpicture}
\end{equation}

Using Theorem \ref{YBE}, we push the two R-vertices back to the left without changing the partition function. Namely, if we denote the partition function of the following configuration by $Z_3$, then $Z_3=Z_2$. Here $S$ is colored stochastic $\Gamma-\Gamma$ vertex of spectral parameters $z_{i+1},z_{i}$, and $T$ is colored stochastic $\Delta-\Gamma$ vertex of spectral parameters $z_{i+1},z_{i}$.
\begin{equation}
    \begin{tikzpicture}[baseline=(current bounding box.center)]
  \draw (0,0) to (0.6,0) to [out=0, in=180] (3,2) to (4,2);
  \draw (0,3) to (0.6,3) to [out=0, in=180] (3,1) to (4,1);
  \draw (0,1) to [out=0, in=180] (2.4,3) to (3,3) to (4,3);
  \draw (2.5,0) to (4,0);
  \draw [dashed] (4,2) to (5,2);
  \draw [dashed] (4,1) to (5,1);
  \draw [dashed] (4,3) to (5,3);
  \draw [dashed] (4,0) to (5,0);
  \draw (5,2) to (5.5,2);
  \draw (5,1) to (5.5,1);
  \draw (5,3) to (5.5,3);
  \draw (5,0) to (5.5,0);
  \draw (5.5,2) arc(-90:90:0.5);
  \draw (5.5,0) arc(-90:90:0.5);
  \draw (3,-0.5) to (3,3.5);
  \draw (4,-0.5) to (4,3.5);
  \draw (5,-0.5) to (5,3.5);
  \filldraw[black] (6,0.5) circle (2pt);
  \filldraw[black] (6,2.5) circle (2pt);
  \filldraw[black] (2.1,1.5) circle (2pt);
  \filldraw[black] (1.5,2.5) circle (2pt);
  \filldraw[black] (3,0) circle (2pt);
  \filldraw[black] (4,0) circle (2pt);
  \filldraw[black] (5,0) circle (2pt);
  \filldraw[black] (3,1) circle (2pt);
  \filldraw[black] (4,1) circle (2pt);
  \filldraw[black] (5,1) circle (2pt);
  \filldraw[black] (3,2) circle (2pt);
  \filldraw[black] (4,2) circle (2pt);
  \filldraw[black] (5,2) circle (2pt);
  \filldraw[black] (3,3) circle (2pt);
  \filldraw[black] (4,3) circle (2pt);
  \filldraw[black] (5,3) circle (2pt);
  \node at (2.1,1.5) [anchor=south] {$T$};
  \node at (1.5,2.5) [anchor=south] {$S$};
  \node at (0,0) [anchor=east] {$+$};
  \node at (0,1) [anchor=east] {$c_{\sigma(i)}$};
  \node at (2.5,0) [anchor=east] {$+$};
  \node at (0,3) [anchor=east] {$c_{\sigma(i+1)}$};
  \node at (5.5,0) [anchor=south] {$z_{i}$};
  \node at (5.5,1) [anchor=south] {$z_{i}$};
  \node at (5.5,2) [anchor=south] {$z_{i+1}$};
  \node at (5.5,3) [anchor=south] {$z_{i+1}$};
  \node at (2.7,0) [anchor=south] {$\Delta$};
  \node at (2.7,1) [anchor=south] {$\Gamma$};
  \node at (2.7,2) [anchor=south] {$\Delta$};
  \node at (2.7,3) [anchor=south] {$\Gamma$};
\end{tikzpicture}
\end{equation}

Now we denote by $Z_4$ and $Z_5$ the partition functions of the following two configurations.

\begin{equation}
\begin{tikzpicture}[baseline=(current bounding box.center)]
\draw (2.5,2) arc(-90:90:0.5);
  \draw (2.5,0) arc(-90:90:0.5);
  \filldraw[black] (3,0.5) circle (2pt);
  \filldraw[black] (3,2.5) circle (2pt);
 \draw (-0.5,2) to (1,2);
  \draw (-0.5,1) to (1,1);
  \draw (-0.5,3) to (1,3);
  \draw (-0.5,0) to (1,0);
  \draw [dashed] (1,2) to (2,2);
  \draw [dashed] (1,1) to (2,1);
  \draw [dashed] (1,3) to (2,3);
  \draw [dashed] (1,0) to (2,0);
  \draw (2,2) to (2.5,2);
  \draw (2,1) to (2.5,1);
  \draw (2,3) to (2.5,3);
  \draw (2,0) to (2.5,0);
  \draw (0,-0.5) to (0,3.5);
  \draw (1,-0.5) to (1,3.5);
  \draw (2,-0.5) to (2,3.5);
  \filldraw[black] (0,0) circle (2pt);
  \filldraw[black] (1,0) circle (2pt);
  \filldraw[black] (2,0) circle (2pt);
  \filldraw[black] (0,1) circle (2pt);
  \filldraw[black] (1,1) circle (2pt);
  \filldraw[black] (2,1) circle (2pt);
  \filldraw[black] (0,2) circle (2pt);
  \filldraw[black] (1,2) circle (2pt);
  \filldraw[black] (2,2) circle (2pt);
  \filldraw[black] (0,3) circle (2pt);
  \filldraw[black] (1,3) circle (2pt);
  \filldraw[black] (2,3) circle (2pt);
  \node at (-0.5,0) [anchor=east] {$+$};
  \node at (-0.5,2) [anchor=east] {$+$};
  \node at (-0.5,1) [anchor=east] {$c_{\sigma(i)}$};
  \node at (-0.5,3) [anchor=east] {$c_{\sigma(i+1)}$};
  \node at (2.8,0) [anchor=south] {$z_{i}$};
  \node at (2.8,1) [anchor=south] {$z_{i}$};
  \node at (2.8,2) [anchor=south] {$z_{i+1}$};
  \node at (2.8,3) [anchor=south] {$z_{i+1}$};
  \node at (2.3,0) [anchor=south] {$\Delta$};
  \node at (2.3,1) [anchor=south] {$\Gamma$};
  \node at (2.3,2) [anchor=south] {$\Delta$};
  \node at (2.3,3) [anchor=south] {$\Gamma$};
\end{tikzpicture}
\quad\quad
\begin{tikzpicture}[baseline=(current bounding box.center)]
\draw (2.5,2) arc(-90:90:0.5);
  \draw (2.5,0) arc(-90:90:0.5);
  \filldraw[black] (3,0.5) circle (2pt);
  \filldraw[black] (3,2.5) circle (2pt);
  \draw (-0.5,2) to (1,2);
  \draw (-0.5,1) to (1,1);
  \draw (-0.5,3) to (1,3);
  \draw (-0.5,0) to (1,0);
  \draw [dashed] (1,2) to (2,2);
  \draw [dashed] (1,1) to (2,1);
  \draw [dashed] (1,3) to (2,3);
  \draw [dashed] (1,0) to (2,0);
  \draw (2,2) to (2.5,2);
  \draw (2,1) to (2.5,1);
  \draw (2,3) to (2.5,3);
  \draw (2,0) to (2.5,0);
  \draw (0,-0.5) to (0,3.5);
  \draw (1,-0.5) to (1,3.5);
  \draw (2,-0.5) to (2,3.5);
  \filldraw[black] (0,0) circle (2pt);
  \filldraw[black] (1,0) circle (2pt);
  \filldraw[black] (2,0) circle (2pt);
  \filldraw[black] (0,1) circle (2pt);
  \filldraw[black] (1,1) circle (2pt);
  \filldraw[black] (2,1) circle (2pt);
  \filldraw[black] (0,2) circle (2pt);
  \filldraw[black] (1,2) circle (2pt);
  \filldraw[black] (2,2) circle (2pt);
  \filldraw[black] (0,3) circle (2pt);
  \filldraw[black] (1,3) circle (2pt);
  \filldraw[black] (2,3) circle (2pt);
  \node at (-0.5,0) [anchor=east] {$+$};
  \node at (-0.5,2) [anchor=east] {$+$};
  \node at (-0.5,1) [anchor=east] {$c_{\sigma(i+1)}$};
  \node at (-0.5,3) [anchor=east] {$c_{\sigma(i)}$};
  \node at (2.8,0) [anchor=south] {$z_{i}$};
  \node at (2.8,1) [anchor=south] {$z_{i}$};
  \node at (2.8,2) [anchor=south] {$z_{i+1}$};
  \node at (2.8,3) [anchor=south] {$z_{i+1}$};
  \node at (2.3,0) [anchor=south] {$\Delta$};
  \node at (2.3,1) [anchor=south] {$\Gamma$};
  \node at (2.3,2) [anchor=south] {$\Delta$};
  \node at (2.3,3) [anchor=south] {$\Gamma$};
\end{tikzpicture}
\end{equation}
By considering all possible configurations of $S,T$, we conclude that
\begin{eqnarray}
    Z_3&=&R_{\Gamma\Gamma}(c_{\sigma(i)},c_{\sigma(i+1)},c_{\sigma(i+1)},c_{\sigma(i)};z_{i+1},z_{i})R_{\Delta\Gamma}(+,c_{\sigma(i)},+,c_{\sigma(i)};z_{i+1},z_{i})Z_4  \nonumber\\
    &&+R_{\Gamma\Gamma}(c_{\sigma(i)},c_{\sigma(i+1)},c_{\sigma(i)},c_{\sigma(i+1)};z_{i+1},z_{i})R_{\Delta\Gamma}(+,c_{\sigma(i+1)},+,c_{\sigma(i+1)};z_{i+1},z_{i})Z_5
\end{eqnarray}
Now note that
\begin{equation}
    Z_4=Z(\mathcal{S}_{n,L,\lambda,\sigma,\tau,z}),
\end{equation}
\begin{equation}
    Z_5=Z(\mathcal{S}_{n,L,\lambda,\sigma s_i,\tau,z}).
\end{equation}
Therefore we have
\begin{eqnarray}
&&R_{\Delta\Delta}(+,+,+,+;z_{i},z_{i+1})R_{\Delta\Gamma}(+,c_{\sigma(i)},+,c_{\sigma(i)};z_{i},z_{i+1})Z(\mathcal{S}_{n,L,\lambda,\sigma,\tau,s_i z}) \nonumber\\
&=&  R_{\Gamma\Gamma}(c_{\sigma(i)},c_{\sigma(i+1)},c_{\sigma(i+1)},c_{\sigma(i)};z_{i+1},z_{i})R_{\Delta\Gamma}(+,c_{\sigma(i)},+,c_{\sigma(i)};z_{i+1},z_{i})Z(\mathcal{S}_{n,L,\lambda,\sigma,\tau,z})  \nonumber\\
&&+R_{\Gamma\Gamma}(c_{\sigma(i)},c_{\sigma(i+1)},c_{\sigma(i)},c_{\sigma(i+1)};z_{i+1},z_{i})R_{\Delta\Gamma}(+,c_{\sigma(i+1)},+,c_{\sigma(i+1)};z_{i+1},z_{i})Z(\mathcal{S}_{n,L,\lambda,\sigma s_i,\tau,z}).
\end{eqnarray}
Using the Boltzmann weights for colored R-matrices and simplifying the expressions, we obtain the conclusion of the theorem.

\end{proof}

\subsubsection{Proof of Theorem \ref{Thm3C}}
The proof of Theorem \ref{Thm3C} is based on the following idea. First note that only one color (other than $+$), denoted by $R$, may appear in the $2n$th row. So we can simultaneously switch $+$ and $R$ in the $2n$th row and change Boltzmann weights in a similar way as in the uncolored case. This gives two rows of colored stochastic $\Delta$ vertices on the top, and we can attach an R-matrix to the left boundary and use the Yang Baxter equation (Theorem \ref{YBE}) to move it to the right. Using a variant of the fish relation which involves two auxiliary caps (\ref{fish2}), the original partition function is related to two new partition functions (\ref{T1})-(\ref{T2}). Changing the colors in the $2n$th row (switching $+$ and $R$, or switching $+$ and $\overline{R}$) and changing the Boltzmann weights simultaneously, the two new partition functions are further related to the two terms on the right hand side of the recursive relation.

\begin{proof}
We note that all the boundary edges on the top of the rectangular lattice carry the $+$ spin, and there are only two possible colors $c_{\sigma(n)}$ and $+$ in the $2n$th row. We write $R:=c_{\sigma(n)}$ hereafter to simplify the notations. Therefore, only the three states in Figure \ref{FishC} are involved in the $2n$th row of the lattice. Moreover, only the two states in Figure \ref{FishC2} are involved in the cap connecting the last two rows.

\begin{figure}[h]
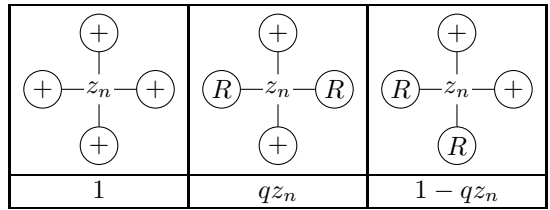

\[
\begin{array}{|c|c|c|c|c|c|}
\hline
\gammaicen{+}{+}{+}{+} &
\gammaicen{R}{+}{R}{+} &
\gammaicen{R}{+}{+}{R}\\
\hline
   1 & q z_n & 1-q z_n \\
\hline\end{array}\]
\caption{Boltzmann weights involved in the $2n$th row}
\label{FishC}
\end{figure}

\begin{figure}[h]
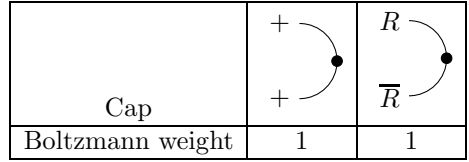

\[
\begin{array}{|c|c|c|c|c|c|}
\hline
\text{Cap} &\caps{+}{+} & \caps{R}{\overline{R}} \\
\hline
\text{Boltzmann weight}  &  1 & 1 \\
\hline\end{array}\]
\caption{Boltzmann weights involved for the cap connecting the last two rows}
\label{FishC2}
\end{figure}

Now for each admissible state, we change the color in the $2n$th row from $+$ to $R$ and from $R$ to $+$ (note that no other colors are involved in the $2n$th row for an admissible state). Meanwhile we change the Boltzmann weights for the vertices in the $2n$th row into the ones presented in Figure \ref{FishC3}, and change the Boltzmann weights for the caps connecting the last two rows into those in Figure \ref{FishC4}. Note that in the original configuration, if the colored path entering from the left of the $2n$th row doesn't go through the cap connecting the last two rows, then one $c_1$ pattern is involved (and no $c_2$ pattern is involved); otherwise neither $c_1$ nor $c_2$ is involved. Thus noting the changed Boltzmann weights for the cap, we can deduce that if we denote by $Z_1$ the partition function of the new system, then
\begin{equation}
    Z_1=(\frac{1}{qz_n})^L Z(\mathcal{S}_{n,L,\lambda,\sigma,\tau,z}).
\end{equation}

\begin{figure}[h]
\[
\begin{array}{|c|c|c|c|c|c|}
\hline
  \tt{a}_1\slash \tt{a}_2&\tt{a}_1\slash \tt{a}_2&\tt{b}_1&\tt{b}_2&\tt{d}_1&\tt{d}_2\\
\hline
\gammaicei{c_{\alpha}}{c_{\alpha}}{c_{\alpha}}{c_{\alpha}} &
\gammaicei{c_{\beta}}{c_{\beta}}{c_{\beta}}{c_{\beta}} &
\gammaicei{c_{\alpha}}{c_{\beta}}{c_{\alpha}}{c_{\beta}} &
\gammaicei{c_{\beta}}{c_{\alpha}}{c_{\beta}}{c_{\alpha}} &
\gammaicei{c_{\beta}}{c_{\beta}}{c_{\alpha}}{c_{\alpha}} &
\gammaicei{c_{\alpha}}{c_{\alpha}}{c_{\beta}}{c_{\beta}}\\
\hline
   1 & 1 & \frac{1}{z_n} & \frac{1}{qz_n} &  1-\frac{1}{z_n} & 1-\frac{1}{qz_n}\\
\hline\end{array}\]
\caption{New Boltzmann weights for the $2n$th row, where $\alpha<\beta$}
\label{FishC3}
\end{figure}
\begin{figure}[h]
\[
\begin{array}{|c|c|c|c|c|c|}
\hline
\text{New cap} &\newcaps{+}{\overline{R}} & \newcaps{R}{+} \\
\hline
\text{Boltzmann weight}  &  1 & -1 \\
\hline\end{array}\]
\caption{New Boltzmann weights for the cap connecting the last two rows}
\label{FishC4}
\end{figure}

Note that the new Boltzmann weights for the $2n$th row correspond to a colored stochastic $\Delta$ vertex with spectral parameter $\frac{1}{z_n'}$.
Now we attach an R-vertex (colored stochastic $\Delta-\Delta$ vertex with spectral parameters $\frac{1}{z_n'}$ and $z_n$) to the left boundary of the last two rows of the new system:
\begin{equation}
    \begin{tikzpicture}[baseline=(current bounding box.center)]
  \draw (0,0) to [out=0, in=-150] (1,0.5) to [out=30, in=180] (2,1);
  \draw (0,1) to [out=0, in=150] (1,0.5) to [out=-30, in=180] (2,0);
  \draw (2,0) to (3.5,0);
  \draw  [dashed] (3.5,0) to (4.5,0);
  \draw (4.5,0) to (5.5,0);
  \draw (2,1) to (3.5,1);
  \draw  [dashed] (3.5,1) to (4.5,1);
  \draw (4.5,1) to (5.5,1);
  \draw (5.5,0) to [out = 0, in = 180] (6,0.5);
  \draw (5.5,1) to [out = 0, in = 180] (6,0.5);
  \draw (2.5,-0.5) to (2.5,1.5);
  \draw (3.5,-0.5) to (3.5,1.5);
  \draw (4.5,-0.5) to (4.5,1.5);
  \filldraw[black] (2.5,0) circle (2pt);
  \filldraw[black] (3.5,0) circle (2pt);
  \filldraw[black] (4.5,0) circle (2pt);
  \filldraw[black] (2.5,1) circle (2pt);
  \filldraw[black] (3.5,1) circle (2pt);
  \filldraw[black] (4.5,1) circle (2pt);
  \filldraw[black] (6,0.5) circle (2pt);
  \node at (0,0) [anchor=east] {$+$};
  \node at (0,1) [anchor=east] {$+$};
  \filldraw[black] (1,0.5) circle (2pt);
  \node at (5,1) [anchor=south] {$\Delta$};
  \node at (5,0) [anchor=north] {$\Delta$};
  \node at (5.5,1) [anchor=south] {$\frac{1}{z_n'}$};
  \node at (5.5,0) [anchor=north] {$z_n$};
  \node at (0,0) [anchor=east] {$+$};
\end{tikzpicture}
\end{equation}
Note that the only admissible configuration of the R-matrix is given by
\begin{equation}
    \begin{tikzpicture}[baseline=(current bounding box.center)]
  \draw (0,0) to [out=0, in=-150] (1,0.5) to [out=30, in=180] (2,1);
  \draw (0,1) to [out=0, in=150] (1,0.5) to [out=-30, in=180] (2,0);
  \draw (2,0) to (3.5,0);
  \draw  [dashed] (3.5,0) to (4.5,0);
  \draw (4.5,0) to (5.5,0);
  \draw (2,1) to (3.5,1);
  \draw  [dashed] (3.5,1) to (4.5,1);
  \draw (4.5,1) to (5.5,1);
  \draw (5.5,0) to [out = 0, in = 180] (6,0.5);
  \draw (5.5,1) to [out = 0, in = 180] (6,0.5);
  \draw (2.5,-0.5) to (2.5,1.5);
  \draw (3.5,-0.5) to (3.5,1.5);
  \draw (4.5,-0.5) to (4.5,1.5);
  \filldraw[black] (2.5,0) circle (2pt);
  \filldraw[black] (3.5,0) circle (2pt);
  \filldraw[black] (4.5,0) circle (2pt);
  \filldraw[black] (2.5,1) circle (2pt);
  \filldraw[black] (3.5,1) circle (2pt);
  \filldraw[black] (4.5,1) circle (2pt);
  \filldraw[black] (6,0.5) circle (2pt);
  \node at (0,0) [anchor=east] {$+$};
  \node at (0,1) [anchor=east] {$+$};
  \filldraw[black] (1,0.5) circle (2pt);
  \node at (5,1) [anchor=south] {$\Delta$};
  \node at (5,0) [anchor=north] {$\Delta$};
  \node at (5.5,1) [anchor=south] {$\frac{1}{z_n'}$};
  \node at (5.5,0) [anchor=north] {$z_n$};
  \node at (0,0) [anchor=east] {$+$};
  \node at (0,0) [anchor=east] {$+$};
  \node at (2,0) [anchor=north] {$+$};
  \node at (2,1) [anchor=south] {$+$};
\end{tikzpicture}
\end{equation}
Therefore, the partition function of the above system is equal to $Z_1$.

By Theorem \ref{YBE}, we can push the R-vertex to the right without changing the partition function. That is, the partition function of the above system is equal to the partition function of the following
\begin{equation}
\begin{tikzpicture}[baseline=(current bounding box.center)]
\draw (3.5,0) to [out=0, in=-150] (4.5,0.5) to [out=30, in=180] (5.5,1);
 \draw (3.5,1) to [out=0, in=150] (4.5,0.5) to [out=-30, in=180] (5.5,0);
  \draw (0,0) to (1.5,0);
  \draw  [dashed] (1.5,0) to (2.5,0);
  \draw (2.5,0) to (3.5,0);
  \draw (0,1) to (1.5,1);
  \draw  [dashed] (1.5,1) to (2.5,1);
  \draw (2.5,1) to (3.5,1);
  \draw (5.5,0) to [out = 0, in = 180] (6,0.5);
  \draw (5.5,1) to [out = 0, in = 180] (6,0.5);
  \draw (0.5,-0.5) to (0.5,1.5);
  \draw (1.5,-0.5) to (1.5,1.5);
  \draw (2.5,-0.5) to (2.5,1.5);
  \filldraw[black] (0.5,0) circle (2pt);
  \filldraw[black] (1.5,0) circle (2pt);
  \filldraw[black] (2.5,0) circle (2pt);
  \filldraw[black] (0.5,1) circle (2pt);
  \filldraw[black] (1.5,1) circle (2pt);
  \filldraw[black] (2.5,1) circle (2pt);
  \filldraw[black] (6,0.5) circle (2pt);
  \node at (0,0) [anchor=east] {$+$};
  \node at (0,1) [anchor=east] {$+$};
  \filldraw[black] (4.5,0.5) circle (2pt);
  \node at (3,1) [anchor=south] {$\Delta$};
  \node at (3,0) [anchor=north] {$\Delta$};
  \node at (3.5,1) [anchor=south] {$z_n$};
  \node at (3.5,0) [anchor=north] {$\frac{1}{z_n'}$};
  \node at (0,0) [anchor=east] {$+$};
\end{tikzpicture}
\end{equation}

We introduce two types of auxiliary caps $C_1$ and $C_2$. The Boltzmann weights for these caps are shown in Figures \ref{FishC5}-\ref{FishC6}.
\begin{figure}[h]
\[
\begin{array}{|c|c|c|c|c|c|}
\hline
\text{New cap} &\capsC{+}{\overline{R}}{C_1} & \capsC{R}{+}{C_1} \\
\hline
\text{Boltzmann weight}  &  1 & -1 \\
\hline\end{array}\]
\caption{Boltzmann weights for the cap $C_1$}
\label{FishC5}
\end{figure}
\begin{figure}[h]
\[
\begin{array}{|c|c|c|c|c|c|}
\hline
\text{New cap} &\capsC{+}{R}{C_2} & \capsC{\overline{R}}{+}{C_2} \\
\hline
\text{Boltzmann weight}  &  -1 &  1 \\
\hline\end{array}\]
\caption{Boltzmann weights for the cap $C_2$}
\label{FishC6}
\end{figure}

Now let $Z(I_1(\epsilon_1,\epsilon_2))$ be the partition function of the following system for every choice of $\epsilon_1,\epsilon_2\in\{c_{\overline{n}},\cdots,c_n\}$ (where the R-vertex is the one we used above, and the cap weights are given by those in Figure \ref{FishC4}).
\begin{equation}
\hfill
I_1(\epsilon_1,\epsilon_2)=
\begin{tikzpicture}[baseline=(current bounding box.center)]
  \draw (0,0) to [out=0, in=-150] (1,0.5) to [out=30, in=180] (2,1);
  \draw (0,1) to [out=0, in=150] (1,0.5) to [out=-30, in=180] (2,0);
  \draw (2,0) to [out = 0, in = 180] (2.5,0.5);
  \draw (2,1) to [out = 0, in = 180] (2.5,0.5);
  \filldraw[black] (2.5,0.5) circle (2pt);
  \node at (0,0) [anchor=east] {$\epsilon_2$};
  \node at (0,1) [anchor=east] {$\epsilon_1$};
  \filldraw[black] (1,0.5) circle (2pt);
\end{tikzpicture}
\end{equation}
Also denote by $Z(I_2(\epsilon_1,\epsilon_2)),Z(I_2'(\epsilon_1,\epsilon_2))$ the partition functions of the following two systems for every choice of $\epsilon_1,\epsilon_2\in\{c_{\overline{n}},\cdots,c_n\}$.
\begin{equation}
\hfill
I_2(\epsilon_1,\epsilon_2)=
\begin{tikzpicture}[baseline=(current bounding box.center)]
  \draw (0,0) to [out = 0, in = 180] (0.5,0.5);
  \draw (0,1) to [out = 0, in = 180] (0.5,0.5);
  \filldraw[black] (0.5,0.5) circle (2pt);
  \node at (0,0) [anchor=east] {$\epsilon_2$};
  \node at (0,1) [anchor=east] {$\epsilon_1$};
  \node at (0.5,0.5) [anchor=west] {$C_1$};
\end{tikzpicture}
\quad\text{ and }\quad
I_2'(\epsilon_1,\epsilon_2)=
\begin{tikzpicture}[baseline=(current bounding box.center)]
  \draw (0,0) to [out = 0, in = 180] (0.5,0.5);
  \draw (0,1) to [out = 0, in = 180] (0.5,0.5);
  \filldraw[black] (0.5,0.5) circle (2pt);
  \node at (0,0) [anchor=east] {$\epsilon_2$};
  \node at (0,1) [anchor=east] {$\epsilon_1$};
   \node at (0.5,0.5) [anchor=west] {$C_2$};
\end{tikzpicture}
\end{equation}
Then we can check that for any $\epsilon_1,\epsilon_2$,
\begin{equation}\label{fish2}
    Z(I_1(\epsilon_1,\epsilon_2))=\frac{(qz_n-1)(1-z_n')}{qz_n+z_n'-(q+1)}Z(I_2(\epsilon_1,\epsilon_2))+\frac{q(z_nz_n'-1)}{q z_n+z_n'-(q+1)}Z(I_2'(\epsilon_1,\epsilon_2)).
\end{equation}

Thus if we denote by $Z_2,Z_2'$ the partition functions of the following two configurations, then
\begin{equation}\label{T1}
\begin{tikzpicture}[baseline=(current bounding box.center)]
  \draw (0,0) to (1.5,0);
  \draw  [dashed] (1.5,0) to (2.5,0);
  \draw (2.5,0) to (3.5,0);
  \draw (0,1) to (1.5,1);
  \draw  [dashed] (1.5,1) to (2.5,1);
  \draw (2.5,1) to (3.5,1);
  \draw (3.5,0) to [out = 0, in = 180] (4,0.5);
  \draw (3.5,1) to [out = 0, in = 180] (4,0.5);
  \draw (0.5,-0.5) to (0.5,1.5);
  \draw (1.5,-0.5) to (1.5,1.5);
  \draw (2.5,-0.5) to (2.5,1.5);
  \filldraw[black] (0.5,0) circle (2pt);
  \filldraw[black] (1.5,0) circle (2pt);
  \filldraw[black] (2.5,0) circle (2pt);
  \filldraw[black] (0.5,1) circle (2pt);
  \filldraw[black] (1.5,1) circle (2pt);
  \filldraw[black] (2.5,1) circle (2pt);
  \filldraw[black] (4,0.5) circle (2pt);
  \node at (0,0) [anchor=east] {$+$};
  \node at (0,1) [anchor=east] {$+$};
  \node at (3,1) [anchor=south] {$\Delta$};
  \node at (3,0) [anchor=north] {$\Delta$};
  \node at (3.5,1) [anchor=south] {$z_n$};
  \node at (3.5,0) [anchor=north] {$\frac{1}{z_n'}$};
  \node at (0,0) [anchor=east] {$+$};
  \node at (4,0.5) [anchor=west] {$C_1$};
\end{tikzpicture}
\end{equation}
\begin{equation}\label{T2}
\begin{tikzpicture}[baseline=(current bounding box.center)]
  \draw (0,0) to (1.5,0);
  \draw  [dashed] (1.5,0) to (2.5,0);
  \draw (2.5,0) to (3.5,0);
  \draw (0,1) to (1.5,1);
  \draw  [dashed] (1.5,1) to (2.5,1);
  \draw (2.5,1) to (3.5,1);
  \draw (3.5,0) to [out = 0, in = 180] (4,0.5);
  \draw (3.5,1) to [out = 0, in = 180] (4,0.5);
  \draw (0.5,-0.5) to (0.5,1.5);
  \draw (1.5,-0.5) to (1.5,1.5);
  \draw (2.5,-0.5) to (2.5,1.5);
  \filldraw[black] (0.5,0) circle (2pt);
  \filldraw[black] (1.5,0) circle (2pt);
  \filldraw[black] (2.5,0) circle (2pt);
  \filldraw[black] (0.5,1) circle (2pt);
  \filldraw[black] (1.5,1) circle (2pt);
  \filldraw[black] (2.5,1) circle (2pt);
  \filldraw[black] (4,0.5) circle (2pt);
  \node at (0,0) [anchor=east] {$+$};
  \node at (0,1) [anchor=east] {$+$};
  \node at (3,1) [anchor=south] {$\Delta$};
  \node at (3,0) [anchor=north] {$\Delta$};
  \node at (3.5,1) [anchor=south] {$z_n$};
  \node at (3.5,0) [anchor=north] {$\frac{1}{z_n'}$};
  \node at (0,0) [anchor=east] {$+$};
  \node at (4,0.5) [anchor=west] {$C_2$};
\end{tikzpicture}
\end{equation}
\begin{equation}
    Z_1=\frac{(qz_n-1)(1-z_n')}{qz_n+z_n'-(q+1)} Z_2+\frac{q(z_nz_n'-1)}{qz_n+z_n'-(q+1)}Z_2'.
\end{equation}

Finally we compute $Z_2,Z_2'$. For $Z_2$, we change the Boltzmann weights of the $2n$th row to the Boltzmann weights in Figure \ref{FishC7}. We also change the Boltzmann weights of the cap to those in Figure \ref{FishC8} (call it $C_1'$). It can be checked that the partition function doesn't change.
\begin{figure}[h]
\[
\begin{array}{|c|c|c|c|c|c|}
\hline
  \tt{a}_1\slash \tt{a}_2&\tt{a}_1\slash \tt{a}_2&\tt{b}_1&\tt{b}_2&\tt{d}_1&\tt{d}_2\\
\hline
\gammaicei{c_{\alpha}}{c_{\alpha}}{c_{\alpha}}{c_{\alpha}} &
\gammaicei{c_{\beta}}{c_{\beta}}{c_{\beta}}{c_{\beta}} &
\gammaicei{c_{\alpha}}{c_{\beta}}{c_{\alpha}}{c_{\beta}} &
\gammaicei{c_{\beta}}{c_{\alpha}}{c_{\beta}}{c_{\alpha}} &
\gammaicei{c_{\beta}}{c_{\beta}}{c_{\alpha}}{c_{\alpha}} &
\gammaicei{c_{\alpha}}{c_{\alpha}}{c_{\beta}}{c_{\beta}}\\
\hline
   1 & 1 & z_n' & \frac{1}{q}z_n' &  z_n'-1 & \frac{1}{q} z_n'-1\\
\hline\end{array}\]
\caption{New Boltzmann weights for the $2n$th row, where $\alpha<\beta$: for computing $Z_2$}
\label{FishC7}
\end{figure}
\begin{figure}[h]
\[
\begin{array}{|c|c|c|c|c|c|}
\hline
\text{New cap} &\capsC{+}{\overline{R}}{C_1'} & \capsC{R}{+}{C_1'} \\
\hline
\text{Boltzmann weight}  &  1 & 1 \\
\hline\end{array}\]
\caption{Boltzmann weights for the cap $C_1'$}
\label{FishC8}
\end{figure}

Now noting that the top boundary edges all carry the $+$ spin, we use the previous argument (changing $R$ to $+$ and $+$ to $R$, and changing the Boltzmann weights accordingly) to show that $Z_2$ is equal to $(\frac{z_n'}{q})^L$ times the partition function of the following system, which is $Z(\mathcal{S}_{n,L,\lambda,\sigma,\tau,s_n z})$.
\begin{equation}
\begin{tikzpicture}[baseline=(current bounding box.center)]
  \draw (0,0) to (1.5,0);
  \draw  [dashed] (1.5,0) to (2.5,0);
  \draw (2.5,0) to (3.5,0);
  \draw (0,1) to (1.5,1);
  \draw  [dashed] (1.5,1) to (2.5,1);
  \draw (2.5,1) to (3.5,1);
  \draw (3.5,0) arc(-90:90:0.5);
  \draw (0.5,-0.5) to (0.5,1.5);
  \draw (1.5,-0.5) to (1.5,1.5);
  \draw (2.5,-0.5) to (2.5,1.5);
  \filldraw[black] (0.5,0) circle (2pt);
  \filldraw[black] (1.5,0) circle (2pt);
  \filldraw[black] (2.5,0) circle (2pt);
  \filldraw[black] (0.5,1) circle (2pt);
  \filldraw[black] (1.5,1) circle (2pt);
  \filldraw[black] (2.5,1) circle (2pt);
  \filldraw[black] (4,0.5) circle (2pt);
  \node at (0,0) [anchor=east] {$+$};
  \node at (0,1) [anchor=east] {$R$};
  \node at (3,1) [anchor=south] {$\Gamma$};
  \node at (3,0) [anchor=north] {$\Delta$};
  \node at (3.5,1) [anchor=south] {$\frac{1}{z_n'}$};
  \node at (3.5,0) [anchor=north] {$\frac{1}{z_n'}$};
  \node at (0,0) [anchor=east] {$+$};
\end{tikzpicture}
\end{equation}

For $Z_2'$, we change the Boltzmann weights of the $2n$th row to the Boltzmann weights in Figure \ref{FishC7}. We also change the Boltzmann weights of the cap to those in Figure \ref{FishC9} (call it $C_2'$). It can be checked that the partition function changes by a factor of $-1$.
\begin{figure}[h]
\[
\begin{array}{|c|c|c|c|c|c|}
\hline
\text{New cap} &\capsC{+}{R}{C_2'} & \capsC{\overline{R}}{+}{C_2'} \\
\hline
\text{Boltzmann weight}  &  1 & 1 \\
\hline\end{array}\]
\caption{Boltzmann weights for the cap $C_2'$}
\label{FishC9}
\end{figure}
Now noting again that the top boundary edges all carry the $+$ spin, we use the previous argument (this time we change $\overline{R}$ to $+$ and $+$ to $\overline{R}$, and change the Boltzmann weights accordingly) to show that $Z_2'$ is equal to $-(z_n')^L$ times the partition function of the following system, which is $Z(\mathcal{S}_{n,L,\lambda, \sigma s_n,\tau,s_n z})$.
\begin{equation}
\begin{tikzpicture}[baseline=(current bounding box.center)]
  \draw (0,0) to (1.5,0);
  \draw  [dashed] (1.5,0) to (2.5,0);
  \draw (2.5,0) to (3.5,0);
  \draw (0,1) to (1.5,1);
  \draw  [dashed] (1.5,1) to (2.5,1);
  \draw (2.5,1) to (3.5,1);
  \draw (3.5,0) arc(-90:90:0.5);
  \draw (0.5,-0.5) to (0.5,1.5);
  \draw (1.5,-0.5) to (1.5,1.5);
  \draw (2.5,-0.5) to (2.5,1.5);
  \filldraw[black] (0.5,0) circle (2pt);
  \filldraw[black] (1.5,0) circle (2pt);
  \filldraw[black] (2.5,0) circle (2pt);
  \filldraw[black] (0.5,1) circle (2pt);
  \filldraw[black] (1.5,1) circle (2pt);
  \filldraw[black] (2.5,1) circle (2pt);
  \filldraw[black] (4,0.5) circle (2pt);
  \node at (0,0) [anchor=east] {$+$};
  \node at (0,1) [anchor=east] {$\overline{R}$};
  \node at (3,1) [anchor=south] {$\Gamma$};
  \node at (3,0) [anchor=north] {$\Delta$};
  \node at (3.5,1) [anchor=south] {$\frac{1}{z_n'}$};
  \node at (3.5,0) [anchor=north] {$\frac{1}{z_n'}$};
  \node at (0,0) [anchor=east] {$+$};
\end{tikzpicture}
\end{equation}

Therefore we conclude that
\begin{eqnarray}
(\frac{1}{z_n})^L  Z(\mathcal{S}_{n,L,\lambda,\sigma,\tau,z})&=&\frac{(qz_n-1)(1-z_n')}{z_n'+qz_n-(q+1)}z_n'^L Z(\mathcal{S}_{n,L,\lambda,\sigma,\tau,s_nz})  \nonumber\\
&&-\frac{q(z_nz_n'-1)}{qz_n+z_n'-(q+1)}q^Lz_n'^L Z(\mathcal{S}_{n,L,\lambda,\sigma s_n,\tau,s_nz}).
\end{eqnarray}
By rearranging and changing $z$ to $s_n z$ we reach the conclusion of the theorem.
\end{proof}

\subsection{Relation to Demazure-Lusztig operators of type C}\label{Sect.3.4}
The recursive relations for colored stochastic symplectic ice shown in Section \ref{Sect.3.3} are related to Demazure-Lusztig operators of type C.
We explain this connection in this section.

Viewed as operators on rational functions of $u=(u_1,\cdots,u_n)$, Demazure-Lusztig operators of type C can be given as follows (see \cite{BBL,Pus1,Pus2} for details). For $1\leq i\leq n-1$, define
\begin{equation}
    s_i (u_1,\cdots,u_n)=(u_1,\cdots,u_{i+1},u_i,\cdots,u_n),
\end{equation}
that is, $s_i$ transposes $u_i$ and $u_{i+1}$. Also define
\begin{equation}
    s_n (u_1,\cdots,u_n)=(u_1,\cdots,u_{n-1},u_n^{-1}).
\end{equation}
For every $1\leq i\leq n$, and any rational function $f(u)$ of $u$, we let
\begin{equation}
    s_if(u):=f(s_i u).
\end{equation}
Then Demazure-Lusztig operators (with parameter $v$) $\mathcal{L}_{i,v}$ are given by
\begin{equation}
    \mathcal{L}_{i,v}(f)=\frac{1-v}{u^{\alpha_i}-1}f+\frac{v u^{\alpha_i}-1}{u^{\alpha_i}-1}s_i(f),
\end{equation}
where $\{\alpha_i\}$ are the simple roots of type $C_n$, that is, $\alpha_i=\epsilon_i-\epsilon_{i+1}$ for $1\leq i\leq n-1$ and $\alpha_n=2\epsilon_n$. Here $\epsilon_i$ is the $n$-dimensional vector such that its $i$th coordinate is $1$ and the other coordinates are $0$, for every $1\leq i\leq n$.

We also let $\hat{\mathcal{L}}_{i,v}=\mathcal{L}_{i,v}-v+1$. Note that from the quadratic relation for $\mathcal{L}_{i,v}$
\begin{equation}
    \mathcal{L}_{i,v}^2=(v-1)\mathcal{L}_{i,v}+v,
\end{equation}
we obtain
\begin{equation}
    \hat{\mathcal{L}}_{i,v}\mathcal{L}_{i,v}=v.
\end{equation}

In order to relate the recursive relations to Demazure-Lusztig operator of type C, we make the following change of variables. We let
\begin{equation}
    u_i=\frac{1-q z_i}{1-z_i}
\end{equation}
for every $1\leq i\leq n$. Then we have
\begin{equation}
    z_i=\frac{1-u_i}{q-u_i},
\end{equation}
\begin{equation}
    z_i'=\frac{1-q u_i}{1-u_i}.
\end{equation}
Under this change of variables, we obtain that for every $1\leq i \leq n-1$,
\begin{equation}
    A(q,z,i)=\frac{(q-1)u_i}{u_i-u_{i+1}},
\end{equation}
\begin{equation}
    B(q,z,i)=\frac{q u_i-u_{i+1}}{u_i-u_{i+1}},
\end{equation}
\begin{equation}
    C(q,z)=\frac{1-q}{q(1-u_n^2)},
\end{equation}
\begin{equation}
    D(q,z)=\frac{1-q u_n^2}{q(1-u_n^2)},
\end{equation}
where $A(q,z,i)$ and $B(q,z,i)$ are as in Theorem \ref{Thm2C}, and $C(q,z)$ and $D(q,z)$ are as in Theorem \ref{Thm3C}.

Now we let 
\begin{equation}
    Z(\tilde{\mathcal{S}}_{n,L,\lambda,\sigma,\tau,u})=Z(\mathcal{S}_{n,L,\lambda,\sigma,\tau,z})\frac{q^{\sum_{i=1}^n (n-i) 1_{\sigma(i)>0}+(L+1)\sum_{i=1}^n 1_{\sigma(i)<0}}}{\prod_{i=1}^n z_i^L}
\end{equation}
as a function of $u=(u_1,\cdots,u_n)$.
Then Theorems \ref{Thm2C}-\ref{Thm3C} translate into the following 

\begin{theorem}
For $1\leq i\leq n-1$, if $\sigma(i+1)>\sigma(i)$, we have
\begin{equation}
    Z(\tilde{\mathcal{S}}_{n,L,\lambda,\sigma s_i,\tau,u})=\hat{\mathcal{L}}_{i,q}(Z(\tilde{\mathcal{S}}_{n,L,\lambda,\sigma,\tau,u})).
\end{equation}
Moreover, if $\sigma(n)>0$, we have
\begin{equation}
    Z(\tilde{\mathcal{S}}_{n,L,\lambda,\sigma s_n,\tau,u})=-\mathcal{L}_{n,q}(Z(\tilde{\mathcal{S}}_{n,L,\lambda,\sigma,\tau,u})).
\end{equation}
\end{theorem}

\section{Another colored model for the stochastic symplectic ice}\label{Sect.4}

In this section, we present a different colored model for the stochastic symplectic ice. In this model, the set of colors is $\{c_1,\cdots,c_n\}$. We also denote by $c_0=+$. 

The model and related Boltzmann weights are introduced in Section \ref{Sect.4.1}. Then the Boltzmann weights for the R-matrices are introduced in Section \ref{Sect.4.2}, and the Yang-Baxter equation is proved there. Finally in Section \ref{Sect.4.3} the reflection equation is introduced, based on which the recursive relations of the partition function are derived.

\subsection{The colored model}\label{Sect.4.1}
We introduce the new colored stochastic symplectic ice in this section. The set of colors for this model is given by $\{c_1,\cdots,c_n\}$. 

In this model, there are also two types of vertices termed ``colored stochastic $\Gamma$ vertex'' and ``colored stochastic $\Delta$ vertex''. The model depends on $n$ spectral parameters $z_1,\cdots,z_n$ and a deformation parameter $q$, and we take 
\begin{equation*}
    z_i'=q+1-\frac{1}{z_i}.
\end{equation*}
The Boltzmann weights for these two types of vertices are listed in Figures \ref{GammaN}-\ref{DeltaN}. 

\begin{figure}[h]
\[
\begin{array}{|c|c|c|c|c|c|}
\hline
  \tt{a}_1\slash\tt{a}_2&\tt{a}_2\slash\tt{a}_2&\tt{b}_1&\tt{b}_2&\tt{c}_1&\tt{c}_2\\
\hline
\gammaice{c_{\alpha}}{c_{\alpha}}{c_{\alpha}}{c_{\alpha}} &
  \gammaice{c_{\beta}}{c_{\beta}}{c_{\beta}}{c_{\beta}} &
  \gammaice{c_{\alpha}}{c_{\beta}}{c_{\alpha}}{c_{\beta}} &
  \gammaice{c_{\beta}}{c_{\alpha}}{c_{\beta}}{c_{\alpha}} &
  \gammaice{c_{\beta}}{c_{\alpha}}{c_{\alpha}}{c_{\beta}} &
  \gammaice{c_{\alpha}}{c_{\beta}}{c_{\beta}}{c_{\alpha}}\\
\hline
   1 & 1 & z_i & q z_i &1-q z_i & 1-z_i\\
\hline\end{array}\]
\caption{Boltzmann weights for colored stochastic $\Gamma$ vertex with spectral parameter $z_i$, where $\alpha<\beta$}
\label{GammaN}
\end{figure}

\begin{figure}[h]
\[
\begin{array}{|c|c|c|c|c|c|}
\hline
  \tt{a}_1\slash\tt{a}_2&\tt{a}_1\slash\tt{a}_2&\tt{b}_1&\tt{b}_2&\tt{d}_1&\tt{d}_2\\
\hline
\gammaice{c_{\alpha}}{c_{\alpha}}{c_{\alpha}}{c_{\alpha}} &
  \gammaice{c_{\beta}}{c_{\beta}}{c_{\beta}}{c_{\beta}} &
  \gammaice{c_{\alpha}}{c_{\beta}}{c_{\alpha}}{c_{\beta}} &
  \gammaice{c_{\beta}}{c_{\alpha}}{c_{\beta}}{c_{\alpha}} &
  \gammaice{c_{\beta}}{c_{\beta}}{c_{\alpha}}{c_{\alpha}} &
  \gammaice{c_{\alpha}}{c_{\alpha}}{c_{\beta}}{c_{\beta}}\\
\hline
   1 & 1 & z_i' & \frac{1}{q} z_i' &  1-z_i' & 1-\frac{1}{q}z_i'\\
\hline\end{array}\]
\caption{Boltzmann weights for colored stochastic $\Delta$ vertex with spectral parameter $z_i$, where $\alpha<\beta$ and $z_i'=q+1-\frac{1}{z_i}$}
\label{DeltaN}
\end{figure}

The basic set-up of the new model is similar to that of the colored model given in Section \ref{Sect.3}. The difference lies in the assignment of boundary conditions: now we specify two permutations $\sigma,\tau$ from the symmetric group $S_n$, assign the color $c_{\tau(i)}$ to each column labeled $\lambda_i+n+1-i$ for $1\leq i\leq n$ at the bottom, and assign the color $c_{\sigma(i)}$ to the left boundary of the $i$th row of colored stochastic $\Gamma$ vertices. The Boltzmann weights for the caps are given in Figure \ref{ColorCapN}. 

\begin{figure}[h]
\[
\begin{array}{|c|c|c|c|c|c|}
\hline
\text{Cap} &\caps{+}{+}  & \caps{c_{\alpha}}{c_{\alpha}} \\
\hline
\text{Boltzmann weight}  &  1  & 1 \\
\hline\end{array}\]
\caption{Boltzmann weights of the caps for the new colored model, where $\alpha\in \{1,2,\cdots,n\}$}
\label{ColorCapN}
\end{figure}

Hereafter we denote by $\mathcal{U}_{n,L,\lambda,\sigma,\tau,z}$ the collection of admissible configurations with the corresponding data. We also denote by $Z(\mathcal{U}_{n,L,\lambda,\sigma,\tau,z})$ the partition function. We assume that $L\geq \lambda_1+n$, too.

We note that the Boltzmann weights for this model are also stochastic, which allows a probabilistic interpretation of the model when the condition (\ref{Co1}) is satisfied. The colored model can be similarly interpreted as stochastic dynamics as the colored model given in Section \ref{Sect.3}. The main difference is that in this model, the particles don't change color when they are reflected at the caps.

\subsection{The R-matrix and the Yang-Baxter equation}\label{Sect.4.2}
For this model, we also find three sets of Yang-Baxter equations. The corresponding R-matrices are termed ``colored stochastic $\Gamma-\Gamma$ vertex'', ``colored stochastic $\Delta-\Gamma$ vertex'' and ``colored stochastic $\Delta-\Delta$ vertex'', too. The Boltzmann weights for the three types of R-matrices are given in Figures \ref{RmatrixC1N}-\ref{RmatrixC3N}.

\begin{figure}[h]
\[\begin{array}{|c|c|c|c|c|c|}
\hline
\begin{tikzpicture}[scale=0.7]
\draw (0,0) to [out = 0, in = 180] (2,2);
\draw (0,2) to [out = 0, in = 180] (2,0);
\draw[fill=white] (0,0) circle (.35);
\draw[fill=white] (0,2) circle (.35);
\draw[fill=white] (2,0) circle (.35);
\draw[fill=white] (2,2) circle (.35);
\node at (0,0) {$c_{\alpha}$};
\node at (0,2) {$c_{\alpha}$};
\node at (2,2) {$c_{\alpha}$};
\node at (2,0) {$c_{\alpha}$};
\path[fill=white] (1,1) circle (.3);
\node at (1,1) {$R_{z_i,z_j}$};
\end{tikzpicture}&
\begin{tikzpicture}[scale=0.7]
\draw (0,0) to [out = 0, in = 180] (2,2);
\draw (0,2) to [out = 0, in = 180] (2,0);
\draw[fill=white] (0,0) circle (.35);
\draw[fill=white] (0,2) circle (.35);
\draw[fill=white] (2,0) circle (.35);
\draw[fill=white] (2,2) circle (.35);
\node at (0,0) {$c_{\beta}$};
\node at (0,2) {$c_{\beta}$};
\node at (2,2) {$c_{\beta}$};
\node at (2,0) {$c_{\beta}$};
\path[fill=white] (1,1) circle (.3);
\node at (1,1) {$R_{z_i,z_j}$};
\end{tikzpicture}&
\begin{tikzpicture}[scale=0.7]
\draw (0,0) to [out = 0, in = 180] (2,2);
\draw (0,2) to [out = 0, in = 180] (2,0);
\draw[fill=white] (0,0) circle (.35);
\draw[fill=white] (0,2) circle (.35);
\draw[fill=white] (2,0) circle (.35);
\draw[fill=white] (2,2) circle (.35);
\node at (0,0) {$c_{\alpha}$};
\node at (0,2) {$c_{\beta}$};
\node at (2,2) {$c_{\alpha}$};
\node at (2,0) {$c_{\beta}$};
\path[fill=white] (1,1) circle (.3);
\node at (1,1) {$R_{z_i,z_j}$};
\end{tikzpicture}&
\begin{tikzpicture}[scale=0.7]
\draw (0,0) to [out = 0, in = 180] (2,2);
\draw (0,2) to [out = 0, in = 180] (2,0);
\draw[fill=white] (0,0) circle (.35);
\draw[fill=white] (0,2) circle (.35);
\draw[fill=white] (2,0) circle (.35);
\draw[fill=white] (2,2) circle (.35);
\node at (0,0) {$c_{\beta}$};
\node at (0,2) {$c_{\alpha}$};
\node at (2,2) {$c_{\beta}$};
\node at (2,0) {$c_{\alpha}$};
\path[fill=white] (1,1) circle (.3);
\node at (1,1) {$R_{z_i,z_j}$};
\end{tikzpicture}&
\begin{tikzpicture}[scale=0.7]
\draw (0,0) to [out = 0, in = 180] (2,2);
\draw (0,2) to [out = 0, in = 180] (2,0);
\draw[fill=white] (0,0) circle (.35);
\draw[fill=white] (0,2) circle (.35);
\draw[fill=white] (2,0) circle (.35);
\draw[fill=white] (2,2) circle (.35);
\node at (0,0) {$c_{\beta}$};
\node at (0,2) {$c_{\alpha}$};
\node at (2,2) {$c_{\alpha}$};
\node at (2,0) {$c_{\beta}$};
\path[fill=white] (1,1) circle (.3);
\node at (1,1) {$R_{z_i,z_j}$};
\end{tikzpicture}&
\begin{tikzpicture}[scale=0.7]
\draw (0,0) to [out = 0, in = 180] (2,2);
\draw (0,2) to [out = 0, in = 180] (2,0);
\draw[fill=white] (0,0) circle (.35);
\draw[fill=white] (0,2) circle (.35);
\draw[fill=white] (2,0) circle (.35);
\draw[fill=white] (2,2) circle (.35);
\node at (0,0) {$c_{\alpha}$};
\node at (0,2) {$c_{\beta}$};
\node at (2,2) {$c_{\beta}$};
\node at (2,0) {$c_{\alpha}$};
\path[fill=white] (1,1) circle (.3);
\node at (1,1) {$R_{z_i,z_j}$};
\end{tikzpicture}\\
\hline
1&1
&\frac{z_i-z_j}{1-(q+1)z_j+q z_i z_j}
&\frac{q(z_i-z_j)}{1-(q+1)z_j+q z_i z_j}
&\frac{(1-q z_i)(1-z_j)}{1-(q+1)z_j+q z_i z_j}
&\frac{(1-z_i)(1-q z_j)}{1-(q+1)z_j+q z_i z_j}\\
\hline
\end{array}\]
\caption{Boltzmann weights for colored stochastic $\Gamma-\Gamma$ vertex with spectral parameters $z_i$ and $z_j$: where $\alpha<\beta$}
\label{RmatrixC1N}
\end{figure}

\begin{figure}[h]
\[\begin{array}{|c|c|c|c|c|c|}
\hline
\begin{tikzpicture}[scale=0.7]
\draw (0,0) to [out = 0, in = 180] (2,2);
\draw (0,2) to [out = 0, in = 180] (2,0);
\draw[fill=white] (0,0) circle (.35);
\draw[fill=white] (0,2) circle (.35);
\draw[fill=white] (2,0) circle (.35);
\draw[fill=white] (2,2) circle (.35);
\node at (0,0) {$c_{\alpha}$};
\node at (0,2) {$c_{\alpha}$};
\node at (2,2) {$c_{\alpha}$};
\node at (2,0) {$c_{\alpha}$};
\path[fill=white] (1,1) circle (.3);
\node at (1,1) {$R_{z_i,z_j}$};
\end{tikzpicture}&
\begin{tikzpicture}[scale=0.7]
\draw (0,0) to [out = 0, in = 180] (2,2);
\draw (0,2) to [out = 0, in = 180] (2,0);
\draw[fill=white] (0,0) circle (.35);
\draw[fill=white] (0,2) circle (.35);
\draw[fill=white] (2,0) circle (.35);
\draw[fill=white] (2,2) circle (.35);
\node at (0,0) {$c_{\beta}$};
\node at (0,2) {$c_{\beta}$};
\node at (2,2) {$c_{\beta}$};
\node at (2,0) {$c_{\beta}$};
\path[fill=white] (1,1) circle (.3);
\node at (1,1) {$R_{z_i,z_j}$};
\end{tikzpicture}&
\begin{tikzpicture}[scale=0.7]
\draw (0,0) to [out = 0, in = 180] (2,2);
\draw (0,2) to [out = 0, in = 180] (2,0);
\draw[fill=white] (0,0) circle (.35);
\draw[fill=white] (0,2) circle (.35);
\draw[fill=white] (2,0) circle (.35);
\draw[fill=white] (2,2) circle (.35);
\node at (0,0) {$c_{\alpha}$};
\node at (0,2) {$c_{\beta}$};
\node at (2,2) {$c_{\alpha}$};
\node at (2,0) {$c_{\beta}$};
\path[fill=white] (1,1) circle (.3);
\node at (1,1) {$R_{z_i,z_j}$};
\end{tikzpicture}&
\begin{tikzpicture}[scale=0.7]
\draw (0,0) to [out = 0, in = 180] (2,2);
\draw (0,2) to [out = 0, in = 180] (2,0);
\draw[fill=white] (0,0) circle (.35);
\draw[fill=white] (0,2) circle (.35);
\draw[fill=white] (2,0) circle (.35);
\draw[fill=white] (2,2) circle (.35);
\node at (0,0) {$c_{\beta}$};
\node at (0,2) {$c_{\alpha}$};
\node at (2,2) {$c_{\beta}$};
\node at (2,0) {$c_{\alpha}$};
\path[fill=white] (1,1) circle (.3);
\node at (1,1) {$R_{z_i,z_j}$};
\end{tikzpicture}&
\begin{tikzpicture}[scale=0.7]
\draw (0,0) to [out = 0, in = 180] (2,2);
\draw (0,2) to [out = 0, in = 180] (2,0);
\draw[fill=white] (0,0) circle (.35);
\draw[fill=white] (0,2) circle (.35);
\draw[fill=white] (2,0) circle (.35);
\draw[fill=white] (2,2) circle (.35);
\node at (0,0) {$c_{\beta}$};
\node at (0,2) {$c_{\beta}$};
\node at (2,2) {$c_{\alpha}$};
\node at (2,0) {$c_{\alpha}$};
\path[fill=white] (1,1) circle (.3);
\node at (1,1) {$R_{z_i,z_j}$};
\end{tikzpicture}&
\begin{tikzpicture}[scale=0.7]
\draw (0,0) to [out = 0, in = 180] (2,2);
\draw (0,2) to [out = 0, in = 180] (2,0);
\draw[fill=white] (0,0) circle (.35);
\draw[fill=white] (0,2) circle (.35);
\draw[fill=white] (2,0) circle (.35);
\draw[fill=white] (2,2) circle (.35);
\node at (0,0) {$c_{\alpha}$};
\node at (0,2) {$c_{\alpha}$};
\node at (2,2) {$c_{\beta}$};
\node at (2,0) {$c_{\beta}$};
\path[fill=white] (1,1) circle (.3);
\node at (1,1) {$R_{z_i,z_j}$};
\end{tikzpicture}\\
\hline
1&1
&\frac{z_i'+q z_j-(q+1)z_i'z_j}{1-z_i'z_j}
&\frac{q^{-1}z_i'+z_j-(1+q^{-1})z_i'z_j}{1-z_i'z_j}
&\frac{(1-z_i')(1-q z_j)}{1-z_i'z_j}
&\frac{(1-q^{-1}z_i')(1-z_j)}{1-z_i'z_j}\\
\hline
\end{array}\]
\caption{Boltzmann weights for colored stochastic $\Delta-\Gamma$ vertex with spectral parameters $z_i$ and $z_j$: where $\alpha<\beta$}
\label{RmatrixC2N}
\end{figure}

\begin{figure}[h]
\[\begin{array}{|c|c|c|c|c|c|}
\hline
\begin{tikzpicture}[scale=0.7]
\draw (0,0) to [out = 0, in = 180] (2,2);
\draw (0,2) to [out = 0, in = 180] (2,0);
\draw[fill=white] (0,0) circle (.35);
\draw[fill=white] (0,2) circle (.35);
\draw[fill=white] (2,0) circle (.35);
\draw[fill=white] (2,2) circle (.35);
\node at (0,0) {$c_{\alpha}$};
\node at (0,2) {$c_{\alpha}$};
\node at (2,2) {$c_{\alpha}$};
\node at (2,0) {$c_{\alpha}$};
\path[fill=white] (1,1) circle (.3);
\node at (1,1) {$R_{z_i,z_j}$};
\end{tikzpicture}&
\begin{tikzpicture}[scale=0.7]
\draw (0,0) to [out = 0, in = 180] (2,2);
\draw (0,2) to [out = 0, in = 180] (2,0);
\draw[fill=white] (0,0) circle (.35);
\draw[fill=white] (0,2) circle (.35);
\draw[fill=white] (2,0) circle (.35);
\draw[fill=white] (2,2) circle (.35);
\node at (0,0) {$c_{\beta}$};
\node at (0,2) {$c_{\beta}$};
\node at (2,2) {$c_{\beta}$};
\node at (2,0) {$c_{\beta}$};
\path[fill=white] (1,1) circle (.3);
\node at (1,1) {$R_{z_i,z_j}$};
\end{tikzpicture}&
\begin{tikzpicture}[scale=0.7]
\draw (0,0) to [out = 0, in = 180] (2,2);
\draw (0,2) to [out = 0, in = 180] (2,0);
\draw[fill=white] (0,0) circle (.35);
\draw[fill=white] (0,2) circle (.35);
\draw[fill=white] (2,0) circle (.35);
\draw[fill=white] (2,2) circle (.35);
\node at (0,0) {$c_{\alpha}$};
\node at (0,2) {$c_{\beta}$};
\node at (2,2) {$c_{\alpha}$};
\node at (2,0) {$c_{\beta}$};
\path[fill=white] (1,1) circle (.3);
\node at (1,1) {$R_{z_i,z_j}$};
\end{tikzpicture}&
\begin{tikzpicture}[scale=0.7]
\draw (0,0) to [out = 0, in = 180] (2,2);
\draw (0,2) to [out = 0, in = 180] (2,0);
\draw[fill=white] (0,0) circle (.35);
\draw[fill=white] (0,2) circle (.35);
\draw[fill=white] (2,0) circle (.35);
\draw[fill=white] (2,2) circle (.35);
\node at (0,0) {$c_{\beta}$};
\node at (0,2) {$c_{\alpha}$};
\node at (2,2) {$c_{\beta}$};
\node at (2,0) {$c_{\alpha}$};
\path[fill=white] (1,1) circle (.3);
\node at (1,1) {$R_{z_i,z_j}$};
\end{tikzpicture}&
\begin{tikzpicture}[scale=0.7]
\draw (0,0) to [out = 0, in = 180] (2,2);
\draw (0,2) to [out = 0, in = 180] (2,0);
\draw[fill=white] (0,0) circle (.35);
\draw[fill=white] (0,2) circle (.35);
\draw[fill=white] (2,0) circle (.35);
\draw[fill=white] (2,2) circle (.35);
\node at (0,0) {$c_{\beta}$};
\node at (0,2) {$c_{\alpha}$};
\node at (2,2) {$c_{\alpha}$};
\node at (2,0) {$c_{\beta}$};
\path[fill=white] (1,1) circle (.3);
\node at (1,1) {$R_{z_i,z_j}$};
\end{tikzpicture}&
\begin{tikzpicture}[scale=0.7]
\draw (0,0) to [out = 0, in = 180] (2,2);
\draw (0,2) to [out = 0, in = 180] (2,0);
\draw[fill=white] (0,0) circle (.35);
\draw[fill=white] (0,2) circle (.35);
\draw[fill=white] (2,0) circle (.35);
\draw[fill=white] (2,2) circle (.35);
\node at (0,0) {$c_{\alpha}$};
\node at (0,2) {$c_{\beta}$};
\node at (2,2) {$c_{\beta}$};
\node at (2,0) {$c_{\alpha}$};
\path[fill=white] (1,1) circle (.3);
\node at (1,1) {$R_{z_i,z_j}$};
\end{tikzpicture}\\
\hline
1&1
&\frac{z_j'-z_i'}{q-(q+1)z_i'+z_i'z_j'}
&\frac{q(z_j'-z_i')}{q-(q+1)z_i'+z_i'z_j'}
&\frac{(1-z_i')(q-z_j')}{q-(q+1)z_i'+z_i'z_j'}
&\frac{(1-z_j')(q-z_i')}{q-(q+1)z_i'+z_i'z_j'}\\
\hline
\end{array}\]
\caption{Boltzmann weights for colored stochastic $\Delta-\Delta$ vertex with spectral parameters $z_i$ and $z_j$: where $\alpha<\beta$}
\label{RmatrixC3N}
\end{figure}

The following theorem gives the three sets of Yang-Baxter equations for the new colored model.

\begin{theorem}\label{YBE2N}
For any $(X,Y)\in\{(\Delta,\Gamma),(\Gamma,\Gamma),(\Delta,\Delta)\}$ the following holds. Assume that $S$ is colored stochastic $X$ vertex with spectral parameter $z_i$, $T$ is colored stochastic $Y$ vertex with spectral parameter $z_j$, and $R$ is colored stochastic $X-Y$ vertex with spectral parameters $z_i,z_j$. Then the partition functions of the following two configurations are equal for any fixed combination of colors $a,b,c,d,e,f\in \{c_0,\cdots,c_n\}$.
\begin{equation}
\hfill
\begin{tikzpicture}[baseline=(current bounding box.center)]
  \draw (0,1) to [out = 0, in = 180] (2,3) to (4,3);
  \draw (0,3) to [out = 0, in = 180] (2,1) to (4,1);
  \draw (3,0) to (3,4);
  \draw[fill=white] (0,1) circle (.3);
  \draw[fill=white] (0,3) circle (.3);
  \draw[fill=white] (3,4) circle (.3);
  \draw[fill=white] (4,3) circle (.3);
  \draw[fill=white] (4,1) circle (.3);
  \draw[fill=white] (3,0) circle (.3);
  \draw[fill=white] (2,3) circle (.3);
  \draw[fill=white] (2,1) circle (.3);
  \draw[fill=white] (3,2) circle (.3);
  \node at (0,1) {$a$};
  \node at (0,3) {$b$};
  \node at (3,4) {$c$};
  \node at (4,3) {$d$};
  \node at (4,1) {$e$};
  \node at (3,0) {$f$};
  \node at (2,3) {$g$};
  \node at (3,2) {$h$};
  \node at (2,1) {$i$};
\filldraw[black] (3,3) circle (2pt);
\node at (3,3) [anchor=south west] {$S$};
\filldraw[black] (3,1) circle (2pt);
\node at (3,1) [anchor=north west] {$T$};
\filldraw[black] (1,2) circle (2pt);
\node at (1,2) [anchor=west] {$R$};
\end{tikzpicture}\qquad\qquad
\begin{tikzpicture}[baseline=(current bounding box.center)]
  \draw (0,1) to (2,1) to [out = 0, in = 180] (4,3);
  \draw (0,3) to (2,3) to [out = 0, in = 180] (4,1);
  \draw (1,0) to (1,4);
  \draw[fill=white] (0,1) circle (.3);
  \draw[fill=white] (0,3) circle (.3);
  \draw[fill=white] (1,4) circle (.3);
  \draw[fill=white] (4,3) circle (.3);
  \draw[fill=white] (4,1) circle (.3);
  \draw[fill=white] (1,0) circle (.3);
  \draw[fill=white] (2,3) circle (.3);
  \draw[fill=white] (1,2) circle (.3);
  \draw[fill=white] (2,1) circle (.3);
  \node at (0,1) {$a$};
  \node at (0,3) {$b$};
  \node at (1,4) {$c$};
  \node at (4,3) {$d$};
  \node at (4,1) {$e$};
  \node at (1,0) {$f$};
  \node at (2,3) {$j$};
  \node at (1,2) {$k$};
  \node at (2,1) {$l$};
\filldraw[black] (1,3) circle (2pt);
\node at (1,3) [anchor=south west] {$T$};
\filldraw[black] (1,1) circle (2pt);
\node at (1,1) [anchor=north west]{$S$};
\filldraw[black] (3,2) circle (2pt);
\node at (3,2) [anchor=west] {$R$};
\end{tikzpicture}
\end{equation}
\end{theorem}
\begin{proof}
From conservation of colors for both colored stochastic $\Gamma$ vertices and colored stochastic $\Delta$ vertices (note that the directions of input and output are different for these two types of vertices), it can be checked that at most four distinct colors (including $c_0$) can appear on the boundary edges in any of the two configurations, and that the color on an inner edge must be one of the colors on boundary edges (only considering admissible configurations). Moreover, the Boltzmann weight of a vertex only depends on the relative order of the colors on its adjacent four edges. Therefore it suffices to check the result for four colors, and there are at most $4^6$ possible combinations of boundary colors. These identities are checked using a SAGE program.
\end{proof}

\subsection{The reflection equation and recursive relations of the partition function}\label{Sect.4.3}

For the new model, we also have the reflection equation:

\begin{theorem}\label{Refl2N}
Assume that $S$ is colored stochastic $\Gamma-\Gamma$ vertex of spectral parameters $z_i,z_j$, $T$ is colored stochastic $\Delta-\Gamma$ vertex of spectral parameters $z_i,z_j$, $S'$ is colored stochastic $\Delta-\Delta$ vertex of spectral parameters $z_j,z_i$, and $T'$ is colored stochastic $\Delta-\Gamma$ vertex of spectral parameters $z_j,z_i$. Denote by $Z(I_5(\epsilon_1,\epsilon_2,\epsilon_3,\epsilon_4))$ the partition function of the following configuration with fixed combination of colors $\epsilon_1,\epsilon_2,\epsilon_3,\epsilon_4  \in\{c_0,\cdots,c_n\}$.
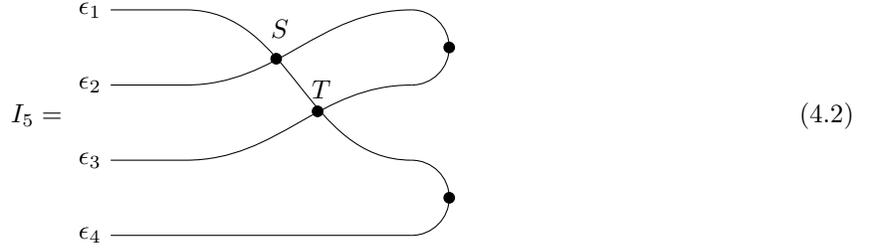
\begin{equation}
\hfill
I_5=
\begin{tikzpicture}[baseline=(current bounding box.center)]
  \draw (0,0) to (4,0);
  \draw (0,1) to (1,1) to [out=0, in=180] (4,2);
  \draw (0,2) to (1,2) to [out=0, in=180] (4,3);
  \draw (0,3) to (1,3) to [out=0, in=180] (4,1);
  \node at (0,0) [anchor=east] {$\epsilon_4$};
  \node at (0,1) [anchor=east] {$\epsilon_3$};
  \node at (0,2) [anchor=east] {$\epsilon_2$};
  \node at (0,3) [anchor=east] {$\epsilon_1$};
  \draw (4,2) arc(-90:90:0.5);
  \draw (4,0) arc(-90:90:0.5);
  \filldraw[black] (4.5,0.5) circle (2pt);
  \filldraw[black] (4.5,2.5) circle (2pt);
  \filldraw[black] (2.2,2.35) circle (2pt);
  \filldraw[black] (2.75,1.65) circle (2pt);
  \node at (2.25,2.5) [anchor=south] {$S$};
  \node at (2.8, 1.7) [anchor=south] {$T$};
\end{tikzpicture}
\end{equation}
Also denote by $Z(I_6(\epsilon_1,\epsilon_2,\epsilon_3,\epsilon_4))$ the partition function of the following configuration with fixed combination of colors $\epsilon_1,\epsilon_2,\epsilon_3,\epsilon_4 \in\{c_0,\cdots,c_n\}$.
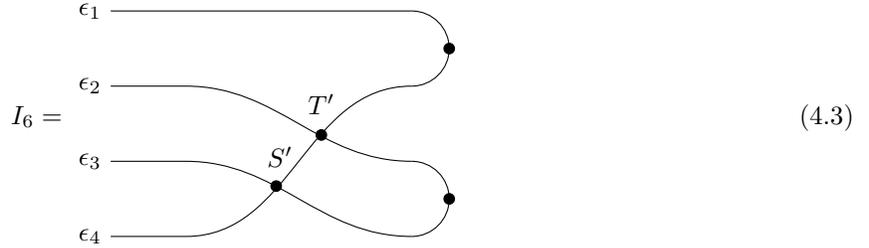
\begin{equation}
\label{eqn:reflect2N}
\hfill
I_6=
\begin{tikzpicture}[baseline=(current bounding box.center)]
  \draw (0,3) to (4,3);
  \draw (0,1) to (1,1) to [out=0, in=180] (4,0);
  \draw (0,2) to (1,2) to [out=0, in=180] (4,1);
  \draw (0,0) to (1,0) to [out=0, in=180] (4,2);
  \node at (0,0) [anchor=east] {$\epsilon_4$};
  \node at (0,1) [anchor=east] {$\epsilon_3$};
  \node at (0,2) [anchor=east] {$\epsilon_2$};
  \node at (0,3) [anchor=east] {$\epsilon_1$};
  \draw (4,2) arc(-90:90:0.5);
  \draw (4,0) arc(-90:90:0.5);
  \filldraw[black] (4.5,0.5) circle (2pt);
  \filldraw[black] (4.5,2.5) circle (2pt);
  \filldraw[black] (2.2,0.67) circle (2pt);
  \filldraw[black] (2.8,1.35) circle (2pt);
  \node at (2.25,0.8) [anchor=south] {$S'$};
  \node at (2.8, 1.5) [anchor=south] {$T'$};
\end{tikzpicture}
\end{equation}
Then for any fixed combination of colors $\epsilon_1,\epsilon_2,\epsilon_3,\epsilon_4\in \{c_0,\cdots, c_n\}$, we have
\begin{equation}
    Z(I_5(\epsilon_1,\epsilon_2,\epsilon_3,\epsilon_4))=Z(I_6(\epsilon_1,\epsilon_2,\epsilon_3,\epsilon_4)).
\end{equation}
\end{theorem}
\begin{proof}
From conservation of colors for the R-matrix and the cap weights, we can deduce that for any admissible state of $I_5$ or $I_6$, each color must appear for an even number of times in $\{\epsilon_1,\epsilon_2,\epsilon_3,\epsilon_4\}$, and that the color of an inner edge must be one of the colors of $\{\epsilon_1,\epsilon_2,\epsilon_3,\epsilon_4\}$. From this we can further deduce that at most two colors can appear in $\{\epsilon_1,\epsilon_2,\epsilon_3,\epsilon_4\}$ in any admissible state of $I_5$ or $I_6$. Moreover, we note that the Boltzmann weight for the R-matrix only depends on the relative order of colors on its four adjacent edges. Therefore it suffices to check the relation for three colors $\{c_0,c_1,c_2\}$. There are at most $3^4$ combinations of $(\epsilon_1,\epsilon_2,\epsilon_3,\epsilon_4)$ for this case. These identities have been checked using a SAGE program.
\end{proof}

Based on the Yang-Baxter equation (Theorem \ref{YBE2N}) and the reflection equation (Theorem \ref{Refl2N}), we can establish the following theorem on the recursive relations for the partition function $Z(\mathcal{U}_{n,L,\lambda,\sigma,\tau,z})$. The proof of Theorem \ref{Thm2C2} is similar to that of Theorem \ref{Thm2C} and we omit it.

\begin{theorem}\label{Thm2C2}
Assume that $1\leq i\leq n-1$ and $\sigma(i+1)>\sigma(i)$. Let $s_i$ be the transposition $(i,i+1)$ in the symmetric group $S_n$, and let $s_i z$ be the vector obtained from $z$ by interchanging $z_i,z_{i+1}$. Then the partition function of the new colored model satisfies the following recursive relation:
\begin{equation}
    Z(\mathcal{U}_{n,L,\lambda,\sigma s_i,\tau,z})
   = -A(q,z,i)  Z(\mathcal{U}_{n,L,\lambda,\sigma,\tau,z})+B(q,z,i)Z(\mathcal{U}_{n,L,\lambda,\sigma,\tau,s_i z}),
\end{equation}
where 
\begin{equation}
A(q,z,i) =  \frac{(1-z_{i+1})(1-qz_i)}{z_{i+1}-z_i},
\end{equation}
and
\begin{equation}
 B(q,z,i) = \frac{1-(q+1)z_i+q z_i z_{i+1}}{z_{i+1}-z_i}.
\end{equation}
\end{theorem}

\newpage
\bibliographystyle{acm}
\bibliography{Symplectic.bib}

\end{document}